\documentclass[10pt,twocolumn,twoside]{IEEEtran}

\usepackage{multirow}  
\usepackage{graphicx}          
\usepackage{amssymb}
\usepackage{amsmath}
\usepackage{enumerate}
\usepackage{color,cite}

\usepackage{amsfonts,amssymb}
\usepackage{bm}
\usepackage{amsmath}

\usepackage{ntheorem}
\usepackage{algorithm}
\usepackage{algorithmic}

\usepackage[font=small]{caption}
\usepackage{float}

\usepackage{graphicx}
\usepackage{subfigure}

\newtheorem{mydef}{Definition}
\newtheorem{assumption}{Assumption}
\newtheorem{lemma}{Lemma}
\newtheorem{mythm}{Theorem}
\newtheorem{remark}{Remark}
\newtheorem{example}{Example}
\newtheorem{corollary}{Corollary}
\newenvironment{proof}{{\noindent \textbf{Proof:}}}{\hfill $\square$ \par}

\ifCLASSINFOpdf
\else
\fi
\begin{document}
%
\title{Single-Leader-Multiple-Followers Stackelberg Security Game with Hypergame Framework}
%
%
%

\author{Zhaoyang~Cheng,
        Guanpu~Chen,
        and~Yiguang~Hong,~\IEEEmembership{Fellow,~IEEE}
 \thanks{This work was supported by Shanghai Municipal Science and
Technology Major Project (No. 2021SHZDZX0100), and by the
National Natural Science Foundation of China (Nos. 62173250 and 61733018).}
\thanks{Z. Cheng is with Key Laboratory of Systems and Control, Academy of Mathematics and Systems Science, Beijing, 100190, China, and is also with School of Mathematical Sciences, University of Chinese Academy of Sciences,Beijing, 100190, China 
(e-mail: chengzhaoyang@amss.ac.cn).
}
\thanks{G. Chen is with JD Explore Academy, Beijing, 100176, China, and is also with Key Laboratory of Systems and Control, Academy of Mathematics and Systems Science, Beijing, China
(e-mail: chengp@amss.ac.cn).
}
\thanks{Y. Hong is with Department of Control Science and Engineering \&
Shanghai Research Institute for Intelligent Autonomous Systems, Tongji
University, Shanghai, 201804, China, and is also with Key Laboratory of Systems
and Control, Academy of Mathematics and Systems Science, Beijing, 100190, China  (e-mail: yghong@iss.ac.cn).}
\thanks{}}

%
%

\markboth{ }%
{}
%



\maketitle

\begin{abstract}

In this paper, we employ a hypergame framework to analyze the single-leader-multiple-followers (SLMF)  Stackelberg security game with two typical misinformed situations: misperception and deception. We provide a stability criterion with the help of hyper Nash equilibrium (HNE) to investigate both strategic stability and cognitive stability of equilibria in SLMF games with misinformation.  In fact,  we find mild  stable  conditions such that the equilibria with misperception and deception can become HNE. Moreover,  we discuss the robustness  of the equilibria to reveal whether players have the ability to keep their profits under the influence of some misinformation. 

\end{abstract}


%
\IEEEpeerreviewmaketitle

\section{Introduction}
%
%
%
%


\IEEEPARstart{S}{ecurity} games describe  the situation between defenders and  attackers, with applications in many fields such as cyber-physical system (CPS), infrastructure protection, and  counterterrorism problems  \cite{korzhyk2011stackelberg,8017519,kar2017trends,9036917,hausken2011defending}.
The Stackelberg security game is one of the significant categories to characterize  practical conflict \cite{korzhyk2011stackelberg,8017519,kar2017trends}. As a  fundamental model in \cite{korzhyk2011stackelberg}, the leader is a defender to prevent invading, while the follower is an attacker to implement malicious behaviors after observing the leader's action.
In addition,  security models with multiple followers are also important since followers can not be usually treated as a  monolithic party, considering that they may  have different preferences, capabilities, and operational strategies \cite{9036917}.

Misinformation occurs in lots of security games  \cite{schlenker2018deceiving,chen2020learning,albanese2016deceiving,8705326}, which may lead to  players' different observations. Specifically, misperception and deception are two typical misinformed situations \cite{8750848}. On the one hand, misperception is usually caused by  external disturbances, bounded rationality,  or accidental errors  \cite{5199443,8691466,chen2021distributed}. For instance,  limited attention  of players in   the Internet of Things (IoT)  increases cyber risks of the community \cite{8691466}. Accordingly, the equilibrium of Stackelberg security games with misperception can be described by the misperception strong Stackelberg equilibrium (MSSE) \cite{8705326}, where no player can change others' observations.
On the other hand, deception usually  results from    belief manipulation,  concealment,   or  camouflages \cite{bakker2020hypergames,zhang2018online,6587320}.  For instance,   in CPS,   the network administrator might change the system's TCP/IP stack and obfuscate the services running on the port, while the hacker is probing the system \cite{schlenker2018deceiving}. Accordingly, the equilibrium of Stackelberg security games with deception can be defined as the deception strong Stackelberg equilibrium (DSSE) \cite{ijcai2019-75}, and the hoaxer can manipulate others' observations, which is different from  the situation with misperception.

In fact, both DSSE and MSSE reflect players' strategic stability, that is, each player has no will to change its own strategy unilaterally. However, due to misinformation, players' cognitive stability \cite{sasaki2015modeling,sasaki2017generalized} is crucial for whether players trust their observations of the game. Actually, players' suspicions of their cognitions may ruin the balance or even lead to the collapse of the model. For example, players may realize the biased misperception and intend to explore the truth \cite{cranford2020toward,8691466}, or the hoaxer does not prefer the current deception along with unsatisfactory benefits \cite{ortmann2002costs,HELLER2019223}. The most existent works have not paid enough attention to cognitive stability of equilibrium analysis in security games, including players changing their communication neighbors \cite{manshaei2013game}, keeping their current cooperators \cite{saad2011coalitional}, or maintaining power systems \cite{an2020stackelberg}.

Fortunately, hypergames provide an effective tool to analyze both strategic and cognitive stability of games with misinformation. Roughly speaking, hypergames describe complex situations when players have different understandings by decomposing a game into multiple subjective games \cite{kovach2015hypergame}.
Hyper Nash equilibrium (HNE) \cite{sasaki2008preservation} is a core concept in hypergames, which represents the best response in each player's subjective game. Once achieving a HNE, each player not only rejects changing its strategy unilaterally, but also trusts its observation of the game since others' strategies are consistent with its own cognition. Such analogous discussions on cognitive stability with HNE have been applied in various circumstances, including resource allocation, military conflicts, and economics \cite{kovach2015hypergame,HELLER2019223}.

Therefore, the motivation of this paper is to employ  a hypergame framework and HNE to investigate the strategic stability and cognitive stability of MSSE and DSSE  in   SLMF Stackelberg security games with misinformation.

 \noindent\textbf{Related Works:} Security games with misinformation have been widely investigated. 
Misperception is one of the typical situations with imperfect observations among players, such as in the infrastructure protection with disturbed invaders \cite{korzhyk2011stackelberg,hausken2011protecting,zhang2019modeling}, and in CPS with small random fluctuations \cite{5199443,kraemer2007human}.
Deception is another typical situation where some players mislead others’ observations. In online negotiations, participants aim to interfere others' beliefs \cite{8573896,albanese2016deceiving}, while in cyber deception games, the network administrator obfuscates the
services running on the port to deceive hackers \cite{schlenker2018deceiving,albanese2016deceiving,8705326}. Additionally, in signaling games, a defender sends deceptive signals to a receiver \cite{feng2017signaling,8573896}.

Different properties of equilibria appear in security games. 
As we know, the Nash equilibrium (NE) is one of the well-known concepts, and  \cite{manshaei2013game,saad2011coalitional} analyzed the conditions when players prefer to keep their current cooperators by NE in wireless
security problems. With the Stackelberg equilibrium (SE),  \cite{ijcai2019-75,ijcai2017-516,feng2019using} analyzed counterterrorism problems with imitative criminals or different attack types in the single-leader-single-follower (SLSF) security game, while \cite{an2020stackelberg} used the cost-based SE for the
optimal allocation of players’ investment resources in power systems. Moreover,  robustness in equilibria with misinformation was also widely discussed in security games. \cite{7577775} considered the sensor networks' robustness capacity for resisting interference signals in Denial-of-Service (DoS) attack problems, while \cite{bakker2020hypergames} analyzed the defender's robustness against the attacker's manipulation in CPS.

Besides strategic stability of equilibria,  cognitive stability has also attracted extensive attention. \cite{sasaki2017generalized} utilized cognitive stability to describe whether each player trusts its own cognition in financial markets without awareness. \cite{ortmann2002costs,HELLER2019223} showed that the misinformation might ruin the balance among players' cognitions of the game, and players' preference changes with the deception in economics and psychology with deception. Also,  \cite{sasaki2008preservation,sasaki2015modeling} revealed a relationship between HNE and cognitive stability in human interactive situations.

 \noindent\textbf{Contributions:} The main contributions  are  as  follows:

\begin{itemize}
\item  We provide a novel second-level hypergame model for     SLMF Stackelberg security games with misperception and deception, and present   a  stability evaluation criterion by HNE. Moreover, compared with  the current stability analysis in security games \cite{manshaei2013game,an2020stackelberg,saad2011coalitional},  the stability criterion based on HNE reflects that players not only avoid changing their strategies unilaterally, but also tend to believe their own cognitions of the game.

\item We   show two different stable conditions such that  MSSE and DSSE can become HNE. In such stable conditions, a HNE as an evaluation criterion covers the stable states  when players do not realize the inherent misperception \cite{sasaki2008preservation} or the hoaxer has no will to change its manipulation under deception \cite{HELLER2019223}. Furthermore, we   show the broad applicability of the obtained stable conditions by verifying them in  typical circumstances \cite{korzhyk2011stackelberg,hausken2011protecting,zhang2019modeling}.

\item With the help of HNE, we also investigate the robustness of  MSSE and  DSSE   to reveal players' capacities to keep their profits.  We give  lower bounds of  MSSE and   DSSE  to describe whether  players can safely ignore the misperception and easily implement deception in different misinformed situations, respectively.

\end{itemize}

The rest of this paper is organized as follows: Section \uppercase\expandafter{\romannumeral2} introduces the SLMF Stackelberg security game with  misperception and deception, and also describes a second-level Stackelberg hypergame model. Section \uppercase\expandafter{\romannumeral3} provides equilibrium analysis with a stability evaluation criterion.  Then Section \uppercase\expandafter{\romannumeral4} gives sufficient conditions such that MSSE and DSSE are   HNE, while  Section \uppercase\expandafter{\romannumeral5}   analyzes the robustness of
the derived equilibria. Additionally, Section \uppercase\expandafter{\romannumeral6} gives   numerical simulations to illustrate our results. Finally, Section \uppercase\expandafter{\romannumeral7} summarizes this paper. A summary of important notations is provided in  Table \ref{ta::1}.

\begin{table}[htbp]
\renewcommand{\arraystretch}{1.3}
\caption{IMPORTANT NOTATIONS}
\centering
\begin{tabular}{|l|l|}
\hline{Notations} & Description \\
\hline$n$ & Number of followers. \\
\hline$l$ & The leader. \\
\hline$\mathbf{P}$ & Set of followers. \\
\hline${\Omega}_l$ & Strategy set of the leader. \\
\hline$\mathbf{\Omega}_f$ & Strategy set of followers. \\
\hline$U_l$ & Utility function of the leader. \\
\hline$\mathbf{U}_f$ & Set of followers' utility functions. \\
\hline$\mathcal{G}$ & SLMF  Stackelberg security game.\\
\hline$\theta_0$ & True value of  a parameter in $\mathcal{G}$. \\
\hline$\theta'$ & Misinformation of $\theta_0$. \\
\hline$\Theta$ & Set of all possible $\theta'$. \\
\hline$\theta^*$ & Optimal deception of $\theta'$ by the leader. \\

\hline$\mathcal{H}^1$ & First-level hypergame. \\
\hline$\mathcal{H}^2(\theta')$ & Second-level hypergame with misperception $\theta'$. \\
\hline$\mathcal{H}^2(\Theta)$ & Second-level hypergame with deception in $\Theta$. \\
\hline $\text{BR}_i$ & Best response of the $i$th follower. \\
\hline$\mathbf{1}_n$ & Row vector with all elements of one. \\
\hline$I_n$ & $n\times n$ identity matrix. \\
\hline$\text{int}(\cdot)$ & Interior of the set. \\
\hline$\chi(\cdot)$ & Indicative function. \\

\hline$\parallel\cdot\parallel$ &  Euclidean norm. \\
\hline
\end{tabular}

\end{table}\label{ta::1}

\section{Security Game and Hypergame}

In this section, we formulate two kinds  of SLMF Stackelberg security games with misperception and deception,  respectively, and we provide a second-level Stackelberg hypergame for SLMF  games with misinformation.

\subsection{SLMF Stackelberg Security Game}

In the Stackelberg security game model \cite{korzhyk2011stackelberg}, the leader is a defender, and the follower is an attacker to attack $K$ targets. The attacker chooses to attack a certain target, while the defender tries to prevent attacks by covering targets with resources from a feasible set. In practice, there may be multiple followers with different preferences, capabilities, and strategies. For instance, there are   multiple layers \cite{8017519} and edge caching devices \cite{9036917} in wireless networks. Also, governments may be confronted with numerous  attackers in counterterrorism problems \cite{hausken2011defending}.

Define the SLMF Stackelberg security game by $\mathcal{G}=\big\{\{l\}\cup\mathbf{P}, \Omega_l\times\mathbf{\Omega}_f, \{U_l\}\cup\mathbf{U}_f\big\}$, where $l$ is the leader and $\mathbf{P}=\{1,\dots,n\}$ is the set of followers. $\Omega_l\subseteq\mathbb{R}^K$ is the strategy set of the leader and $\mathbf{\Omega}_f=\Omega_1\times \dots \times \Omega_n$, where  $\Omega_i\subseteq \mathbb{R}^K$ is the strategy set of the $i$th follower.  Also, $U_l: \Omega_l\times\boldsymbol{\Omega}_f \to \mathbb{R}$ is the leader's utility function and $\mathbf{U}_f=\{U_1,\dots,U_n\}$, where  $U_i:\Omega_l\times \Omega_i\to \mathbb{R}$ is the utility function of the $i$th follower. The leader's utility is influenced by all players actions, while each follower's utility only relies on its own action and the leader's action. On this basis, each player aims at maximizing its own utility function. Denote
the target set as $T=\{t_1,\dots,t_K\}$, which each follower (attacker) aims at attacking but the leader (defender) tries to protect. Then the SLMF game model $\mathcal{G}$ can be presented concisely in Fig. \ref{fi::31}. Suppose that the leader has a resource $R_l\in\mathbb{R}$ to assign on each target, \textit{i.e.}, $x^{k}$ on target $k$ with $\sum\limits_{k=1}^Kx^k= R_l$. Then the leader's strategy is $x=[x^1,\dots,x^K]^T$ and its strategy set is denoted by
\begin{equation}\label{eq::Omega_l}
\begin{aligned}
\Omega_l=\{x| \sum\limits_{k=1}^Kx^k= R_l, x^k\geqslant 0,\forall k=1\dots,K \}.
\end{aligned}
\end{equation}
Similarly, the $i$th follower has a resource $R_i\in \mathbb{R}$ with the strategy $y_i=[y_i^1,\dots,y_i^K]^T$, where $\sum\limits_{k=1}^Ky_i^k= R_i$. Then the strategy set is denoted by
\begin{equation}\label{eq::Omega_i}
\begin{aligned}
\Omega_i=\{y_i| \sum\limits_{k=1}^Ky_i^k= R_i, y_i^k\geqslant 0,\forall k=1\dots,K \},\forall i\in \mathbf{P}.
\end{aligned}
\end{equation}
Moreover, let $\boldsymbol{y}=[y_1^T,\dots,y_n^T]^T$ be the strategy profile of all followers.

\begin{remark}
Many attack-defense mechanisms can be modeled by the SLMF Stackelberg security game $\mathcal{G}$. For example, in CPS \cite{chen2009game}, the invasion type is the DoS attack, while the intrusion detection system (IDS), as a defender, monitors the network with a probability distribution. In counterterrorism problems \cite{zhang2019modeling}, terrorists select different attacking options, like assassination, armed assault, or hijacking, while the government, as a defender, allocates budgets on cities to defend against terrorists.
\end{remark}

\begin{figure}[h]
\centering
\includegraphics[width = 5.5cm ]{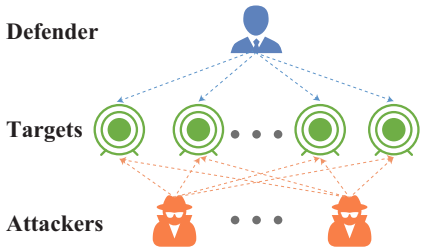}
\caption{SLMF Stackelberg security game model.}
\label{fi::31}
\end{figure}

As  discussed  in \cite{korzhyk2011stackelberg}, let $ U_l^c(t_k)$ be the leader's utility when the leader allocates per unit of resource to target $t_k$ with per unit attacking resource. $ U_l^u(t_k)$ is the leader's utility when the leader does not allocate per unit of resources to target $t_k$ with per unit attacking resource.  Given a strategy profile $(x,\boldsymbol{y})$, the leader's utility function is
\begin{equation}\label{eq::LFUl}
\begin{aligned}
U_l(x,\boldsymbol{y})\!=\!\sum\limits_{k=1}^K \!(\sum\limits_{i=1}^n  y^k_i)\left(x^k U_l^c(t_k)\!+\!(R_l-x^k) U_l^u(t_k)\right),
\end{aligned}
\end{equation}
where $\sum\limits_{i=1}^ny^k_i$ reflects the influence of all followers on target $k$. Similarly, the $i$th follower's utility  consists of $U_i^c(t_k)$ and $ U_i^u(t_k)$. Given the strategy profile $(x,\boldsymbol{y})$, the $i$th follower's utility function is
\begin{equation}\label{eq::LFUi}
\begin{aligned}
U_{i}(x,y_i)=\sum\limits_{k=1}^K y_i^k\left(x^k U_i^c(t_k)+(R_l-x^k) U_i^u(t_k)\right).
\end{aligned}
\end{equation}


If the followers cannot observe the actions of the leader and all players make decisions simultaneously, then we can consider  the Nash Equilibrium (NE) \cite{korzhyk2011stackelberg,DBLP:conf/ijcai/KorzhykCP11}.

\begin{mydef}
A strategy profile $(x^*,\boldsymbol{y}^*)$ is said to be a NE of the SLMF  game $\mathcal{G}$ if
$$
\begin{aligned}
&x^*\in\mathop{\text{\emph{argmax}}}\limits_{x\in \Omega_l} U_l(x,y^*), \\
& y_i^*\in\mathop{\text{\emph{argmax}}}\limits_{y_i\in \Omega_i} U_i(x^*,y_i), \forall i\in \mathbf{P}.
\end{aligned}
$$
\end{mydef}
On the other hand,  when the leader implements an allocation,  followers   determine their strategies after observing the leader's strategy.
Denote the $i$th follower's best response to the leader's strategy $x$ by
$$
\begin{aligned}
\text{BR}_i(x)= \mathop{\text{\text{argmax}}}\limits_{y_i\in \Omega_i}U_i(x,y_i),\forall  i\in \mathbf{P},
\end{aligned}
$$
and $\boldsymbol{\textbf{BR}}(x)=\text{BR}_1(x)\times \dots \times \text{BR}_n(x)$. Without loss of generality,  followers can break ties optimally for the leader if there are multiple best responses.
In this case, we introduce the Strong Stackelberg Equilibrium (SSE) \cite{von2010leadership}.

\begin{mydef}
A strategy profile $(x^*,\boldsymbol{y}^*)$ is said to be a SSE of the SLMF  game $\mathcal{G}$ if
$$
\begin{aligned}
&(x^*, \boldsymbol{y}^*)\in\mathop{\emph{argmax}}\limits_{x\in \Omega_l, \boldsymbol{y}\in \boldsymbol{\textbf{\emph{BR}}}(x)}U_l(x,\boldsymbol{y} ).
\end{aligned}
$$
\end{mydef}

\begin{remark}
According to \cite{ijcai2019-75}, the SSE in the 
SLSF security game can be reducible to  $\{x^*,k^*\}$, where the leader's SSE strategy is $x^*$ and the follower attacks target $t_{k^*}$. However,  multiple followers may attack different targets, and their SSE strategies can not be   reduced to a single target. Additionally, in the SLMF game, each follower's strategy set is a subset of $\mathbb{R}^K$, and is more complex than a subset of $\mathbb{R}$ in \cite{8017519}. Although our model has   similar utility functions as \cite{zhang2019modeling}, we focus  on that each follower   has its own resource, \textit{i.e.}, $\sum\limits_{k=1}^Ky_i^k= R_i$, instead of that all followers allocate a total resource $\sum\limits_{i=1}^n\sum\limits_{k=1}^Ky_i^k= R$.
\end{remark}

Then we  discuss  SLMF games with misinformation, when  players have different cognitions. Specifically, misperception and deception are two typical misinformed situations.

\subsection{Misperception and Deception}
We first consider misperception  for a situation when there are  imperfect or  prejudiced observations/understandings of the game among  players. It is  caused by passive factors with players' biased cognitions. For example, in communication channels such as sensor systems, external disturbances may cause imperfect observations \cite{5199443}, while players with bounded rationality may have  prejudiced observations in the  IoT \cite{8691466}. Moreover, in the computer and information security, accidental errors from small random fluctuations also bring players imprecise observations \cite{kraemer2007human}.

Consider that the followers have prejudiced observations of the security game $\mathcal{G}$, while the leader realizes the situation. To describe different observations,  we consider a parameter $\theta_0\in\mathbb{R}^m$ in the SLMF game $\mathcal{G}$, and followers' prejudiced observations of $\theta_0$ are $\theta'\in\mathbb{R}^m$, where $\theta_0$ and $\theta'$ only affect followers' utility functions. Then the SLMF game with  misperception can be denoted by
$\mathcal{G}_M(\theta_0,\theta')=\big\{\{l\}\cup\mathbf{P}, \Omega_l\times\mathbf{\Omega}_f, \{U_l\}\cup\mathbf{U}_f,\{\theta_0,\theta'\} \big\}$.  Specifically,  $\mathbf{U}_f=\{U_1,\dots,U_n\}$, where $U_i:\Omega_l\times \Omega_i\times\{\theta_0,\theta'\}\to \mathbb{R}$ is the utility function of the $i$th follower. Here  we rewrite the security model without misperception as $\mathcal{G}=\mathcal{G}_M(\theta_0,\theta_0).$

In addition, given the strategy profile $(x,\boldsymbol{y})$, the leader's actual utility function is $U_l(x,\boldsymbol{y})$, which is the same as (\ref{eq::LFUl}). Also the $i$th follower's actual utility function is $U_{i}(x,y_i,\theta_0)$, which is exactly (\ref{eq::LFUi}).  However, known to the leader, the $i$th follower believes that its own utility function is  as follow:
\begin{equation}\label{eq::misperception-utility-2}
\begin{aligned}
U_{i}(x,y_i,\theta')&=\sum\limits_{k=1}^Ky_i^k\left(x^k U_i^c(\theta',t_k)+(R_l-x^k) U_i^u(\theta',t_k)\right)\!.
\end{aligned}
\end{equation}

Correspondingly,  the $i$th follower's best response to the leader strategy $x$ under the prejudiced observation $\theta'$ is
$$
\begin{aligned}
\text{BR}_i(x,\theta')= \mathop{\text{argmax}}\limits_{y_i\in \Omega_i}U_i(x,y_i,\theta'), \forall  i\in\mathbf{P},
\end{aligned}
$$
and $\boldsymbol{\textbf{BR}}(x,\theta')=\text{BR}_1(x,\theta')\times \dots \times \text{BR}_n(x,\theta')$. Similar to SSE, the leader implements an allocation, and afterward, followers determine their strategies after observing the leader's strategy. Therefore, following the security game with misperception \cite{8705326}, the equilibrium with misperception can be denoted by the Misperception Strong Stackelberg Equilibrium (MSSE).

\begin{mydef}
A strategy profile $(x^*, \boldsymbol{y}^*)$ is said to be a MSSE of the SLMF  game with misperception $\mathcal{G}_M(\theta_0,\theta')$   if
$$
\begin{aligned}
(x^*, \boldsymbol{y}^*)\in\mathop{\emph{argmax}}\limits_{x\in \Omega_l,\boldsymbol{y}\in \boldsymbol{\textbf{\emph{BR}}}(x,\theta')} U_l(x,\boldsymbol{y}).
\end{aligned}
$$
\end{mydef}

Next, we address deception  for another situation when some players mislead others' cognitions with selfish or malevolent  motivation.  Unlike misperception, deception is caused by active factors among players  with players’ manipulated cognitions. For instance, each player is explicitly interested in convincing the others to hold some particular beliefs such as  in authentication protocols and online negotiations \cite{bakker2020hypergames,8573896}. Moreover, the leader may tend to deceive followers like a network administrator (leader) and a hacker (follower) in CPS \cite{schlenker2018deceiving}, while the network administrator might change a system's TCP/IP stack and obfuscate the services running on the port \cite{albanese2016deceiving,8705326}.

In this situation, the leader can manipulate  the followers' observation, while followers are not aware. Set $\theta_0\in\mathbb{R}^m$ as the true value of the parameter in $\mathcal{G}$. Take  $\Theta\subseteq\mathbb{R}^m$ as the deceptive set, while  the leader  manipulates   followers' observations of the parameter as $\theta'\in\Theta$ to maximize its own utility function. Denote the SLMF Stackelberg security game with deception by $\mathcal{G}_D(\Theta)=\big\{\{l\}\cup\mathbf{P}, \Omega_l\times\mathbf{\Omega}_f\times\Theta, \{U_l\}\cup\mathbf{U}_f ,\theta_0\big\}$. Specifically,
$\mathbf{U}_f=\{U_1,\dots,U_n\}$, where $U_i:\Omega_l\times \Omega_i\times\Theta\to \mathbb{R}$ is the utility function of the $i$th follower. Here  we rewrite the security model without deception as $\mathcal{G}=\mathcal{G}_D(\{\theta_0\}).$

Given the strategy profile $(x,\boldsymbol{y},\theta')$, players' actual utility functions  are $U_l(x,\boldsymbol{y})$ and $U_{i}(x,y_i,\theta_0)$, $\forall i\in \mathbf{P}$. Since followers' observations are manipulated as $\theta'$, the $i$th follower regard its own utility function  as $U_{i}(x,y_i,\theta')$, which is generated by (\ref{eq::misperception-utility-2}) and the domain of $U_i$ contains $\Theta$ instead of $\{\theta_0,\theta'\}$.

Therefore, the leader manipulates followers' observations of the parameter as $\theta^*\in\Theta$ to maximize its own utility function at first. Then, similar to  SSE and  MSSE, the leader provides its own strategy, and afterward, each follower acts according to the observation $\theta^*$ and the leader's strategy.
Thus, based on the SLSF security game with deception \cite{schlenker2018deceiving,ijcai2019-75},  the equilibrium with deception can be defined as the Deception Strong Stackelberg Equilibrium (DSSE).

\begin{mydef}
A strategy profile $(x^*, \boldsymbol{y}^*,\theta^*)$ is said to be a DSSE of the SLMF  game with deception $\mathcal{G}_D(\Theta)$ if
$$
\begin{aligned}
&(x^*, \boldsymbol{y}^*)\in\mathop{\text{\emph{argmax}}}\limits_{x\in \Omega_l,\boldsymbol{y}\in \boldsymbol{\textbf{\emph{BR}}}(x,\theta^*)}{U_l(x,\boldsymbol{y} )} ,
\end{aligned}
$$
where $ \theta^* \in\mathop{\text{\emph{argmax}}}\limits_{\theta' \in \Theta}\max\limits_{x\in \Omega_l, y\in \boldsymbol{\textbf{\emph{BR}}}(x,\theta')}U_l(x,\boldsymbol{y} )$ is the optimal deception of the leader.
\end{mydef}
Different from MSSE, DSSE describes a decision with players' manipulated cognitions, and the leader can manipulate followers' observations of the parameter $\theta_0$. The following assumptions have been  widely employed in security games with deception  \cite{korzhyk2011stackelberg,schlenker2018deceiving,ijcai2017-516,hu2013existence,ijcai2019-75,zhang2017strategic,zhang2020game,feng2017signaling,von2013information,chen2009game}.

\begin{assumption}\label{as::ass4}
$\Theta$ is compact and convex, while $\emph{int}(\Theta)$ is nonempty and $\theta_0\in\Theta$.
\end{assumption}
\begin{assumption}\label{as::ass5}

For $ i\in\mathbf{P},k=1\dots, K$, $U_i^c(\theta',t_k)$ and $U_i^c(\theta',t_k)$ are differentiable in $\theta'\in \Theta$.
\end{assumption}

 \begin{assumption}\label{as::ass1}
  For $k= 1,\dots , K$, $U_l^c(t_k)> U_l^u(t_k)$.
\end{assumption}
\begin{assumption}\label{as::ass3}
For  $\theta'\in \Theta, i\in\mathbf{P},k=1\dots, K$,  $U_i^c(\theta',t_k)<U_i^u(\theta',t_k)$.
\end{assumption}

\begin{assumption}\label{as::ass6}
There exists $k$, such that for $i\in\mathbf{P}, l\neq k$, $U_i^c(\theta_0,t_k)\geqslant U_i^u(\theta_0,t_l)$.
\end{assumption}

Assumptions \ref{as::ass4} and \ref{as::ass5}   guarantee the existence of a DSSE of $\mathcal{G}_D(\Theta)$ \cite{hu2013existence,schlenker2018deceiving}, which are  also adopted in real-world security problems such as unmanned aerial vehicles (UAVs) security games \cite{zhang2017strategic,zhang2020game} and moving target defense (MTD) problems \cite{feng2017signaling}. Furthermore, Assumption \ref{as::ass1} indicates that, for the leader, the unit utility for defending a target is larger than that without defending. Assumption \ref{as::ass3} indicates that, for each follower, the unit utility for attacking a target is larger than that without attacking. They are consistent with the fact that the leader tends to resist attacks and followers tend to implement invasions \cite{korzhyk2011stackelberg,ijcai2019-75,ijcai2017-516}. Moreover, Assumption \ref{as::ass6} refers to the situation when there exists a most attractive  target for followers  \cite{von2013information,chen2009game}, which describes a relationship among different targets, different from Assumption \ref{as::ass3}.

\begin{remark}
In  many practical situations,  the leader has direct access to others' cognition. 
For example, in  CPS, a network administrator can obfuscate the
services running on the port, while the administrator knows all the information of the services  \cite{albanese2016deceiving,8705326}. In UAVs security problems, the defender may show wrong targets' locations to UAVs, where both true and wrong locations are detected by the defender  \cite{zhang2017strategic}. Additionally, in an industrial control system, Stuxnet, as a leader, can directly obtain the access to the system and feed fake data to disguise malicious actions \cite{bakker2020hypergames}.
\end{remark}

\subsection{Hypergame}

The hypergame theory describes different cognitions among players for the strategic interactions in situations with misinformation.
It covers misperception or deception, where players may have biased misperception or can manipulate others' observations.  The main idea of the hypergame is to decompose complex situations with misinformation  into multiple subjective games.  
 According to \cite{kovach2015hypergame}, each player in a hypergame may 
\begin{itemize}
\item have a misled or false understanding of other players' preferences or utility functions;
\item have an incorrect comprehension of other players' strategy sets;
\item be not aware of   every one of all players;
\item have any combination of the above.
\end{itemize}
In fact, the hypergame theory has been applied in various circumstances, such as CPS and  economic behaviors \cite{HELLER2019223,bakker2020hypergames,sasaki2015modeling}. Benefiting from the subjective games decomposed by hypergames, the relationship among players' strategies, incorporation of opponents' cognitions, and fears of being outguessed  are further explored in situations with misinformation.

\begin{remark}

Standard Stackelberg games describe players' different acting sequences, where the leader implements strategy first and followers act after observing the leader's actions. Hypergames focus on players' different cognitions with multiple subjective games. Players play different subjective games and may also know others' cognition.
Actually, standard Stackelberg games may be regarded as a special hypergame, where followers know the leader's game and follow its action, and the leader also knows this fact. Since the hypergame is good at describing complex situations with misinformation, it helps us analyze both the strategic and cognitive stability of equilibria with misperception and deception.

\end{remark}

There are different levels for describing different cognitive environments \cite{wang1988modeling}. 
For instance, the first-level hypergame describes the situation when players are playing different games, but no one realizes the fact. Correspondingly, the second-level hypergame occurs when at least one player is aware that different games are played.
Then we aim at employing the second-level Stackelberg hypergame to analyze SLMF games with misperception and deception and providing a criterion for evaluating the stability of the equilibrium with misinformation.

As we know, both misperception and deception can cause observation errors of the parameter $\theta_0$ in the game. Take  $\Theta \subseteq\mathbb{R}^m$ as the cognitive set of all followers' possible observations and $\theta_0\in\mathbb{R}^m$ as the true value of the parameter in $\mathcal{G}$.  Denote the game under the observation $\theta'\in\Theta$ by $\mathcal{G}(\theta')=\big\{\{l\}\cup\mathbf{P}, \Omega_l\times\mathbf{\Omega}_f, \{U_l\}\cup\mathbf{U}_f,\theta' \big\}$  with  $\mathbf{U}_f=\{U_1,\dots,U_n\}$, where $U_i:\Omega_l\times \Omega_i\times\Theta\to \mathbb{R}$ is the utility function of the $i$th follower.  In addition, given the strategy profile $(x,\boldsymbol{y})$, players' actual utility functions are $U_l(x,\boldsymbol{y})$ and $U_{i}(x,y_i,\theta_0)$, $\forall i\in \mathbf{P}$. However, in all players' views,  the $i$th follower's utility function    is $U_{i}(x,y_i,\theta')$. Here we rewrite the SLMF  game model without any observation error in Section \uppercase\expandafter{\romannumeral2} as $\mathcal{G}=\mathcal{G}(\theta_0)$.

Consider the first-level hypergame to describe a complex   situation in the  SLMF game $\mathcal{G}$ when there are observed differences among players, but no one is aware.
Concretely, suppose that all followers' observations of the parameter are $\theta'\in\Theta$, while  the leader's observation is $\theta_0$.  For the leader, denote  $\mathcal{G}_l=\mathcal{G}(\theta_0)$ as the game of the leader's self-cognition. For $i\in \mathbf{P}$, denote  $\mathcal{G}_i=\mathcal{G}(\theta')$ as the game of the $i$th follower's self-cognition. Then  the situation can be defined as  $\mathcal{H}^1=\{\mathcal{G}(\theta_0),\big(\mathcal{G}(\theta')\big)_{i\in \mathbf{P}}\}$, which is  a first-level hypergame as shown in Fig. \ref{fi::hyper1level}.

\begin{figure}[h]
\centering
\includegraphics[width =5.7cm ]{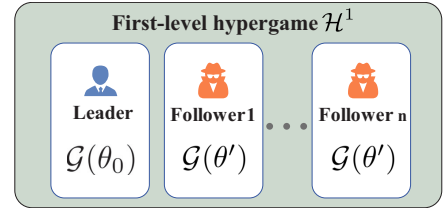}
\caption{First-level hypergame $\mathcal{H}^1$.}
\label{fi::hyper1level}
\end{figure}

Moreover, we employ the second-level hypergame to describe   different misinformed situations when all followers do not realize  the  observed differences and  the leader is aware of the fact.  On the one hand, in the leader's view, for $ i\in\mathbf{P}$, denote $\mathcal{G}_{il}$ as the $i$th follower's game under the leader's perception. Thus, $\mathcal{G}_{il}=\mathcal{G}(\theta') $ since the leader knows the $i$th follower's observation is $\theta'$. Also, denote $\mathcal{G}_{ll}$ as the leader's game under its own perception. Then $\mathcal{G}_{ll}= \mathcal{G}(\theta_0)$ since the leader's own observation is $\theta_0$. Thus, $\mathcal{H}_l^1=\{\mathcal{G}_{ll},(\mathcal{G}_{il})_{i\in \mathbf{P}}\}$ is a novel first-level hypergame perceived by the leader. On the other hand, in the view of the $i$th follower, denote $\mathcal{G}_{li}$ as the leader's game  and $\mathcal{G}_{ii}$ as its own game under the $i$th follower's perception.
Thus, $\mathcal{G}_{li}= \mathcal{G}(\theta')$ and $\mathcal{G}_{ii}= \mathcal{G}(\theta')$ since the $i$th follower is not conscious with $\theta'\neq\theta_0$. Thus, $\mathcal{H}^1_i= \{\mathcal{G}_{li},\mathcal{G}_{ii}\}$  is another first-level hypergame perceived by the $i$th follower.
Notice that, for all $i\in\mathbf{P}$, $\mathcal{H}^1_i$ and $\mathcal{H}_l^1$ are different since the leader notices the cognitive set $\Theta$ but followers do not. Therefore, the different first-level hypergames perceived by all players form a Stackelberg  hypergame $\mathcal{H}^2(\Theta)=\{\mathcal{H}^1_l,(\mathcal{H}^1_i)_{i\in\mathbf{P}}\}$, which is also a second-level  hypergame as shown in Fig. \ref{fi::hyper}.
\begin{figure}[h]
\centering
\includegraphics[width =8.5cm ]{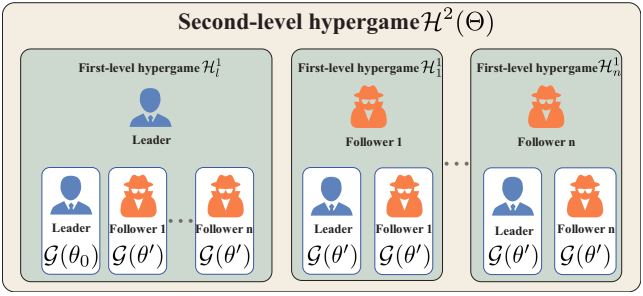}
\caption{Second-level hypergame $\mathcal{H}^2(\Theta)$.}
\label{fi::hyper}
\end{figure}

Clearly, $\mathcal{H}^2(\Theta)$ can describe SLMF games with both deception and misperception. For instance, regarding $\Theta$ as the deceptive set,  we can  rewrite the SLMF  game  with deception in Section \uppercase\expandafter{\romannumeral2} as $\mathcal{G}_D(\Theta)=\mathcal{H}^2(\Theta)$. Especially, let $\Theta=\{\theta'\}$, and then the SLMF  game  with misperception in Section  \uppercase\expandafter{\romannumeral2} can be written  as $\mathcal{G}_M(\theta_0,\theta')=\mathcal{H}^2(\theta')$.

In all players' views, the leader  chooses a strategy with the utility function $U_l(x,\boldsymbol{y})$, while the $i$th follower   makes a decision with the utility function $U_i(x,y_i,\theta')$, which leads to a concept called Hyper Nash Equilibrium (HNE) \cite{sasaki2008preservation}.

\begin{mydef}
A strategy profile $(x^*,\boldsymbol{y}^*)$ is said to be a HNE of $\mathcal{H}^2(\Theta)$ with any fixed $\theta'\in\Theta$ if
$$
\begin{aligned}
&x^*\in\mathop{\emph{argmax}}\limits_{x\in \Omega_l} U_l(x,\boldsymbol{y}^*), \\
& y_i^*\in\mathop{\emph{argmax}}\limits_{y_i\in \Omega_i} U_i(y_i,x^*,\theta' ), \forall i \in \mathbf{P}.
\end{aligned}
$$
\end{mydef}

In fact,   HNE is such a strategy profile that is the best response strategy in everyone's subjective game. Each player does not change its strategy unilaterally if all the players play HNE strategies. It is the same as NE if there is no cognitive difference. In addition, at   HNE, each player can not realize that its cognition is different from others, since others' strategies are consistent with its own anticipation. Then all players have  no incentive to update their observations of the parameter $\theta_0$. Hence, the HNE in our hypergame is a desired equilibrium with cognitive stability, which was similarly described in \cite{sasaki2015modeling,sasaki2017generalized}.  Compared with previous discussions on equilibrium stability in security games \cite{manshaei2013game,an2020stackelberg,saad2011coalitional},    HNE helps  analyze both the strategic stability and cognitive stability  of SLMF games with misinformation, and provides a unified framework for misperception and deception.

\begin{remark}
The signaling game is a leader-follower game with deception, where the leader sends deceptive signals to followers \cite{feng2017signaling,8573896}.
It usually focuses on whether players can achieve the equilibrium with deception, which is actually a learning process to update cognition for followers.
Different from the signaling game, the hypergame concerns the cognitive stability of games with misinformation \cite{sasaki2015modeling,sasaki2017generalized}. In the cases of misperception or deception, the cognitive stability plays an important role since it shows the conditions that each player trusts its own current cognition. Otherwise, the player's anticipation may not be consistent with others’ strategies, and thus, the player may be suspicious about its cognition. For instance, the hoaxer may not believe in the current deception along with benefits, while the victim (follower) may be aware of the biased cognition if the opponents' best response strategies are ``impractical''. Thus, we adopt the hypergame to analyze the cognitive stability of the SLMF game with misinformation.
\end{remark}

\section{Equilibrium Analysis}

In this section, we analyze the equilibria of the proposed formulations in security games.

With complete information, the following lemma verifies the existence of  NE and SSE in the SLMF game $\mathcal{G}$.
\begin{lemma}\label{le::NESSE-exists}
There exists a NE and a SSE of $\mathcal{G}$.
\end{lemma}

\begin{proof}
Recalling (\ref{eq::Omega_l}) and (\ref{eq::Omega_i}), for  $i\in \mathbf{P}$, $\Omega_l$ and $\Omega_i$ are compact  and convex. $U_l(\cdot, \boldsymbol{y})$ and $U_i(x,\cdot)$ are linear for fixed $\boldsymbol{y}\in \mathbf{\Omega}_f$ and   = $x\in \Omega_l$, respectively. 
By Theorem 2.1 in  \cite{carmona2012existence}, there exists a NE of $\mathcal{G}$. 
Also, based on \cite{korzhyk2011stackelberg}, since $\boldsymbol{\textbf{BR}}(x)$ is  compact and $U_l$ and $U_i$ is continuous, there exists a SSE of $\mathcal{G}$.
\end{proof}
With misperception, the following lemma shows the existence of a MSSE of $\mathcal{H}^2(\theta')$, whose proof can be given by  replacing $U_i(x,y_i)$ with $U_i(x,y_i,\theta')$  in Lemma \ref{le::NESSE-exists}.
\begin{lemma}\label{th::MSSE-exists-section2}
For any $\theta'\in \mathbb{R}^m$, there exists   a MSSE  of $\mathcal{H}^2(\theta')$.
\end{lemma}

Moreover, with deception, the next lemma shows the  existence of a DSSE in $\mathcal{H}^2(\Theta)$.
\begin{lemma}\label{th::DSSE-existence-section-3}
Under Assumptions  \ref{as::ass4} and \ref{as::ass5}, there exists   a DSSE  of  $\mathcal{H}^2(\Theta)$.
\end{lemma}

\begin{proof}
Denote
$
\big(x(\theta'),\boldsymbol{y}(\theta')\big) = \mathop{\text{argmax}}\limits_{x\in\Omega_l, \boldsymbol{y}\in \boldsymbol{\textbf{BR}}(x,\theta')}U_l (x,\boldsymbol{y})
$ for any $\theta' \in \Theta$.
Take $L(\theta')= U_l \big(x(\theta'),\boldsymbol{y}(\theta')\big)$ and  $L^*=\mathop{sup}\limits_{\theta'\in\Theta} L (\theta')$. Then there exists a sequence $\{\theta_j\}_{j=1}^\infty$ such that $L(\theta_j)>L^*-\frac{1}{j}$. Since $\Theta $ is compact, there exists  a convergent  subsequence $\{\theta_{j_m}\}_{m=1}^\infty$, where $\lim\limits_{m\to\infty} \theta_{j_m}=\theta^* \in\Theta$. Thus,
$
L^*\geqslant U_l\big(x(\theta_{j_m}),\boldsymbol{y}(\theta_{j_m})\big)>L^*-\frac{1}{j_m}.
$
By the continuity of $U_l$, $\lim\limits_{m\to \infty}U_l\big(x(\theta_{j_m}),\boldsymbol{y}(\theta_{j_m})\big)=L^*$.
Also, there is a convergent subsequence $\{(x(\theta_{j_{m_q}}),(\theta_{j_{m_q}}))\}_{q=1}^\infty$, where
$ \lim\limits_{q\to\infty} (x(\theta_{j_{m_q}}),(\theta_{j_{m_q}}))=(x^*, \boldsymbol{y}^*)$.
Then
$U_l(x^*,\boldsymbol{y}^*)= \max\limits_{\theta' \in \Theta}\max\limits_{x\in \Omega_l, \boldsymbol{y}\in \boldsymbol{\textbf{BR}}(x,\theta')}U_l(x,\boldsymbol{y} )
$. By  Lemma 17.30 in \cite{aliprantis2006infinite},
$\boldsymbol{y}^*\in \boldsymbol{\textbf{BR}}(x^*,\theta^*)$. Then $(x^*,\boldsymbol{y}^*,\theta^*)$ is a DSSE of $\mathcal{H}^2(\Theta)$.
\end{proof}
The following example indicates that Assumptions  \ref{as::ass4} and \ref{as::ass5} are fundamental in Lemma \ref{th::DSSE-existence-section-3}, since there may be no existence of DSSE without Assumptions  \ref{as::ass4} and \ref{as::ass5}. 

\begin{example}\label{ex::1}
Consider a SLSF game with $R_l=R_1=1$, $K=2$, and $\Theta=(0,1)$. Take $U_l^c(t_1)=1$, $U_1^c(\theta',t_1)=\theta'-1$, $U_1^u(\theta',t_1)=\theta'$, and $U_l^u(t_1)=U_l^c(t_2)=U_l^u(t_2)=U_1^c(\theta',t_2)=U_1^u(\theta',t_2)=0$. Then for any $\theta'\in \Theta$, players take $x=[\theta',1-\theta']^T$ and $y_1=[1,0]^T$. Then, the leader's profit is $U_l(x,y_1)=\theta'$. Since $\Theta$ is not closed, there is no DSSE.
\end{example}


Furthermore, the following theorem shows the existence of a HNE in $\mathcal{H}^2(\Theta)$.
\begin{mythm}\label{le::HNE-exists}
Under Assumptions  \ref{as::ass4} and \ref{as::ass5}, there exists   a HNE   of $\mathcal{H}^2(\Theta)$.
\end{mythm}
\begin{proof}
For any fixed $\theta'\in \Theta$, in the leader's view, the $i$th follower acts with $U_i(x,y_i,\theta')$. For any $x\in \Omega_l$ and $\boldsymbol{y}\in \mathbf{\Omega}_f$, denote $F(x, \boldsymbol{y})=\{(\hat{x},\hat{\boldsymbol{y}})|\hat{x}\in\mathop{\emph{\rm{argmax}}}\limits_{x'\in \Omega_l} U_l(x',\boldsymbol{y}), \hat{\boldsymbol{y}}\in\textbf{BR}(x,\theta')\} $.  Take $(x_j,\boldsymbol{y}_j)$ as a covergent sequence, where $
\lim\limits_{j\to\infty}\big(x_{j},(\boldsymbol{y})_{j}\big)=(x^{*},\boldsymbol{y}^*)$. There exists $(\hat{x}_j,\hat{\boldsymbol{y}}_j)\in F(x_j, (\boldsymbol{y})_j)$. Since $\Omega_l\times\mathbf{\Omega}_f$ is compact, there exists a convergent subsequence $\{\big(\hat{x}_{j_m},(\hat{\boldsymbol{y}})_{j_m}\big)\}_{m=1}^\infty$, where $
\lim\limits_{m\to\infty}\big(\hat{x}_{j_m},(\hat{\boldsymbol{y}})_{j_m}\big)=(\hat{x}^{*},\hat{\boldsymbol{y}}^*)$. By the continuity of $U_l$, 
$$
\begin{aligned}
U_l(\hat{x}^*,\!\boldsymbol{y}^*)\!=\!\lim\limits_{m \to \infty}U_l\big(\hat{x}_{j_m},\!({\boldsymbol{y}})_{j_m}\big)\!=\!\lim\limits_{m \to \infty} \max\limits_{x'\in\Omega_l}U_l\big(x',\!({\boldsymbol{y}})_{j_m}\big).
\end{aligned}
$$
According to  Lemma 17.30 in \cite{aliprantis2006infinite},
$$
\begin{aligned}
U_l(\hat{x}^*,\boldsymbol{y}^*)= \max\limits_{x'\in\Omega_l}\lim\limits_{m \to \infty}U_l\big(x',(\hat{\boldsymbol{y}})_{j_m}\big)= \max\limits_{x'\in\Omega_l}U_l\big( x',({\boldsymbol{y}})^*\big).
\end{aligned}
$$
Thus, $\hat{x}^*\in\mathop{\emph{\rm{argmax}}}\limits_{x'\in \Omega_l} U_l(x',({\boldsymbol{y}})^*)$. Similarly, $\hat{\boldsymbol{y}}^*\in\textbf{BR}(x^*,\theta')$. Then $(\hat{x}^*,\hat{\boldsymbol{y}}^*)\in F(x^*, \boldsymbol{y}^*)$. According to Theorem A.14 in  \cite{carmona2012existence}, there exists $(x', \boldsymbol{y}')$ such that $(x', \boldsymbol{y}')\in F(x', \boldsymbol{y}')$. Then $(x', \boldsymbol{y}')$ is the  best response strategy for everyone in the leader's view. Also, in the $i$th follower's view, $x^*\in\mathop{\emph{argmax}}\limits_{x\in \Omega_l} U_l(x,y^*)$, and  $  y'_i\in\text{BR}_i(x',\theta')$. Thus, $(x', \boldsymbol{y}')$ is a HNE of $\mathcal{H}^2(\Theta)$.
\end{proof}

The following example indicates that misinformation may lead to some players' suspicions on the observation of the game, since others' strategies do not match their cognitions.

\begin{example}\label{ex::2}
Consider a SLMF game with $n=K=2$, $R_l=R_1=R_2=1$, $\Theta=\{0,1\}$, and $\theta_0=0$. Take $U_l^u(t_1)=2$, $U_l^u(t_2)=3$, $U_2^u(\theta',t_1)=U_2^u(\theta',t_2)=0$, and $U_l^c(t_1)=U_l^c(t_2)=U_2^c(\theta',t_1)=U_2^c(\theta',t_2)=1$. Denote $U_1^c(\theta',t_1)=U_1^u(\theta',t_1)=\theta'$ and $U_1^c(\theta',t_2)=U_1^u(\theta',t_2)=1-\theta'$. Then the leader's optimal deception is $\theta^*=1$, and players take $x^*=[0,1]^T,y^*_1=[1,0]^T$, and $y^*_2=[0,1]^T$. Notice that the first follower may not observe the leader's strategy before attacking \cite{korzhyk2011stackelberg}. In the first follower's view, $x^*$ is not the best response strategy to $\boldsymbol{y}^*$, and players should have taken  $x'=[0,1]^T,y'_1=[0,1]^T$ with $\theta_0=0$. Clearly, $\theta^*$ brings more benefits to the leader since $U_l(x^*,y_1^*,y_2^*)>U_l(x',y'_1,y^*_2)$. Then the first follower realizes the misinformation and tends to update its observation.
\end{example}
According to  Example \ref{ex::2},  players’ suspicions on their cognitions may cause them to update their observations, and even make the game model collapse. Thus, cognitive stability is crucial for SLMF  games with misinformation. To this end, we aim at analyzing the cognitive stability and strategic stability of MSSE and DSSE with the help of HNE.


\section{Stability Analysis}

It is known that HNE   describes a stable state    when each player does not update  its own cognitions and strategies. In this section, we explore  conditions to reveal how the MSSE and   DSSE can become a HNE in the   Stackelberg hypergame.


\subsection{Stable Conditions}

With misperception, MSSE is called stable when players do not realize the inherent misperception, while, with deception, DSSE is called stable when the leader has no will to change its manipulation on followers' cognitions.

First, in order to   evaluate the stability of  MSSE,   given $\boldsymbol{y}\in \mathbf{\Omega}_f , \theta' \in \mathbb{R}^m$, define 
$$
\begin{aligned}
\text{SOL}(\boldsymbol{y},\theta')=\{\boldsymbol{y}'\in \mathbf{\Omega}_f, \lambda>0|A_1(\theta') \boldsymbol{y}'= \lambda B\boldsymbol{y},A_2 \boldsymbol{y}' =0\},
\end{aligned}
$$ 
where 
$$
\begin{aligned}
&A_1(\theta')=[A_1(\theta',1),\dots, A_1(\theta',n)],\\
&A_1\!(\theta'\!, \!i)\!\!=\!\!\text{diag}\!\!\left(\!\frac{U_1^u\!(\theta'\!,t_1\!)\!\!-\!\!U_1^c\!(\theta'\!,t_1\!)}{U_l^c\!(t_1\!)\!-\!U_l^u\!(t_1\!)}\!,\! \dots\!,\! \frac{U_1^u\!(\theta'\!,t_K\!)\!\!-\!\!U_1^c\!(\theta'\!,t_K\!)}{U_l^c\!(t_K\!)\!-\!U_l^u\!(t_K\!)}\!\right)\!,\\
&A_2= \text{diag}\big(\mathbf{1}_{nK}\!-\!\chi(y)\big),\quad B=[I_{K}, I_{K}, \dots, I_{K}].
\end{aligned}
$$
Notice that $A_2\boldsymbol{y}'=0$ is equivalent to $(\boldsymbol{y}')_i^k=0$ if $y_i^k=0$ for any $i\in\mathbf{P}, k=1,\dots, K$.  Here, $\chi(\cdot)$ is the indicative function where $\chi(x )=0$ iff $x=0$.

Let $(x_{\textit{\tiny MSSE}}, \boldsymbol{y}_{\textit{\tiny MSSE}})$ be a MSSE  of $\mathcal{H}^2(\theta')$. In the following, we give a result about the   MSSE of $\mathcal{H}^2(\theta')$,  whose proof   can be found in Appendix \ref{ap::th::h1}.

\begin{mythm}\label{th::h1}
Under Assumptions \ref{as::ass1} and \ref{as::ass3},  if $\text{\emph{SOL}}(\boldsymbol{y}_{\textit{\tiny MSSE}},\theta')$ is nonempty, then $(x_{\textit{\tiny MSSE}}, \boldsymbol{y}_{\textit{\tiny MSSE}})$ is also a HNE.
\end{mythm}

Theorem \ref{th::h1} implies that a MSSE strategy is stable in such a  condition,  since  such a decision-making process prevents players from realizing the inherent
misperception. Concretely, for each follower, the leader's strategy is consistent with each follower's anticipation, and its MSSE strategy is also the best response strategy in its own subjective game.  Thus,  they can not be aware of their cognitive errors in $\mathcal{H}^2(\theta')$, which also conforms with the two-players game model in \cite{sasaki2008preservation}. Additionally,  for the leader, Theorem \ref{th::h1} also indicates that, even if followers cannot observe the consequences of the leader's strategy, the leader can safely play a MSSE strategy since it is still the best response strategy. Furthermore, no matter how many followers move simultaneously, the  conclusion holds in the SLMF game with misperception $\mathcal{H}^2(\theta')$, which covers the situation  in \cite{korzhyk2011stackelberg}.

Second, in the view of deception, our major concern is the stability of  DSSE. The following theorem gives a sufficient condition to guarantee that a  DSSE of $\mathcal{H}^2(\Theta)$ is a HNE,  whose proof   can be found in Appendix \ref{ap::th::h2}.

\begin{figure*}
\centering
\subfigure[$n=5$]{
\label{fi::n5}
 \includegraphics[width=4.8cm]{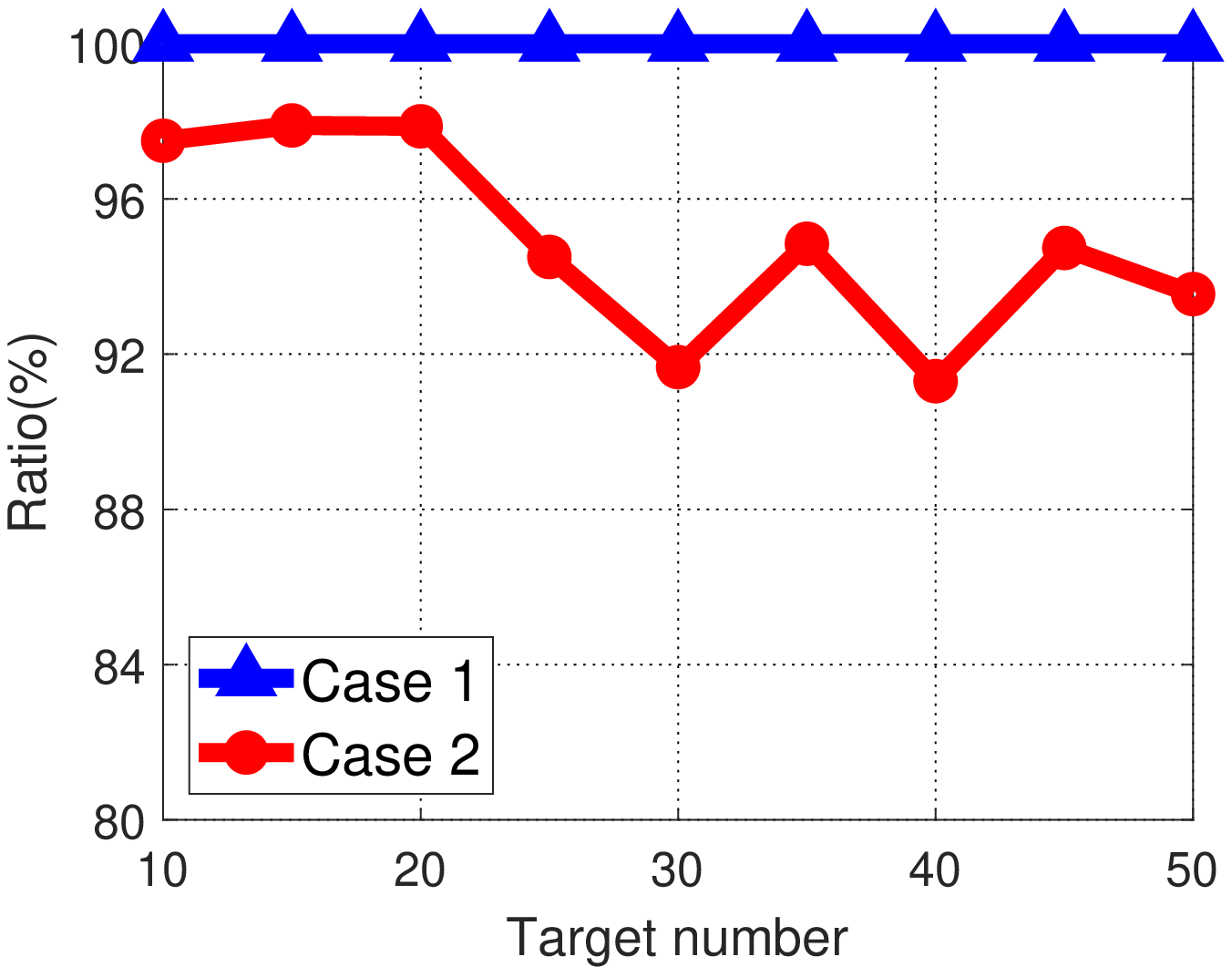}}
 \subfigure[$n=10$]
 { \label{fi::n10}
  \includegraphics[width=4.8cm]{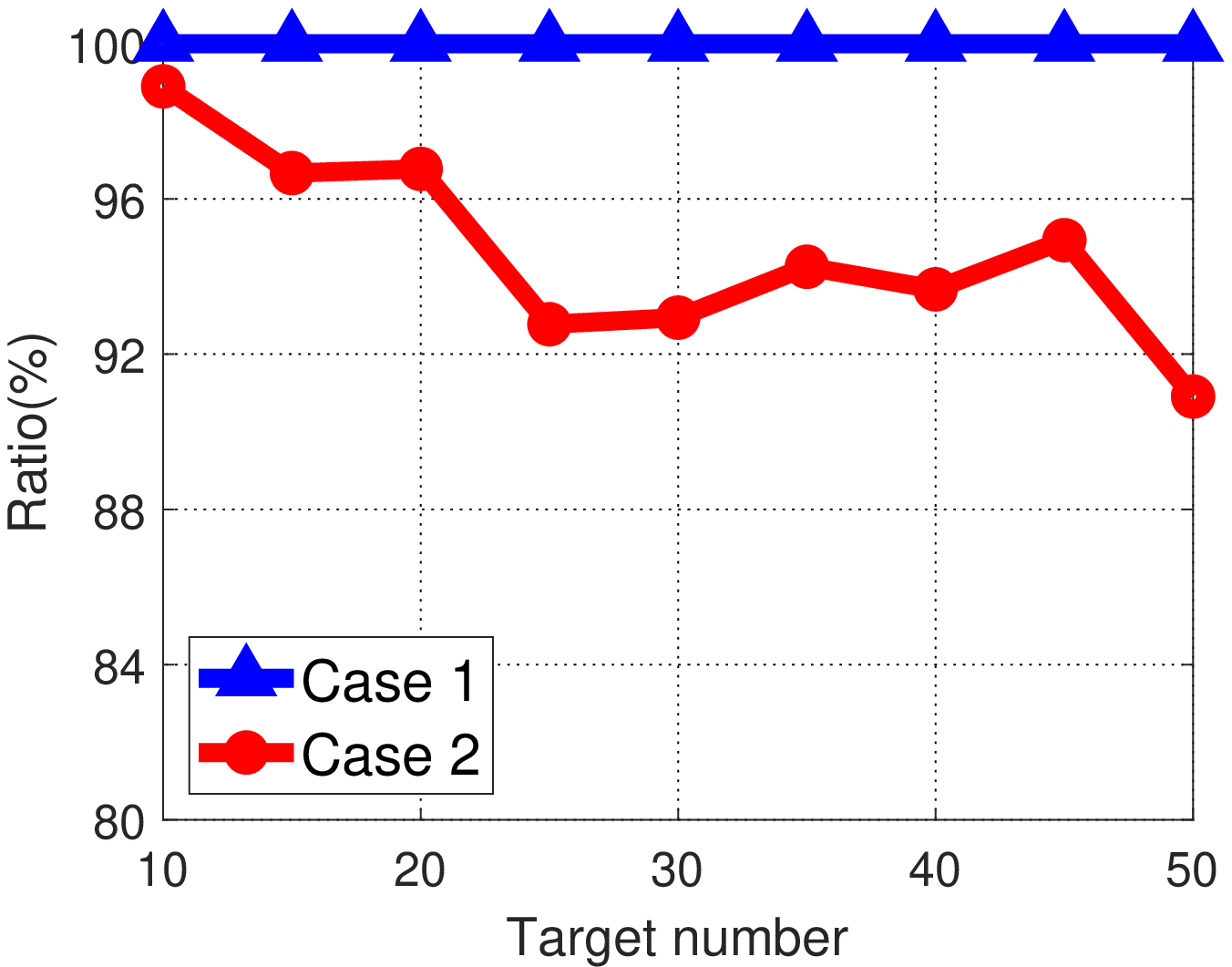}}
   \subfigure[$n=15$]
 { \label{fi::n15}
  \includegraphics[width=4.8cm]{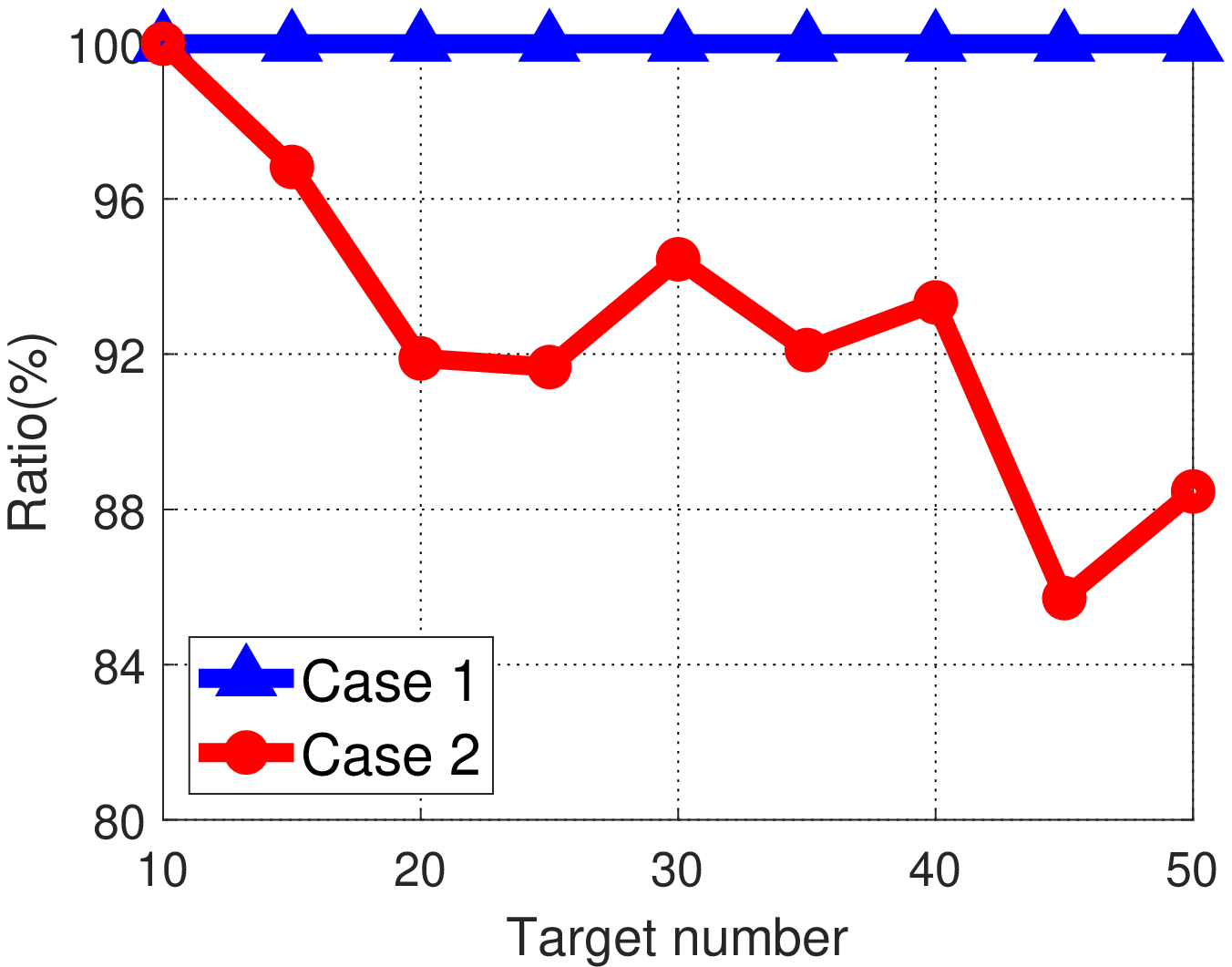}}
       \subfigure[$K=5$]
 { \label{fi::K5}
  \includegraphics[width=4.8cm]{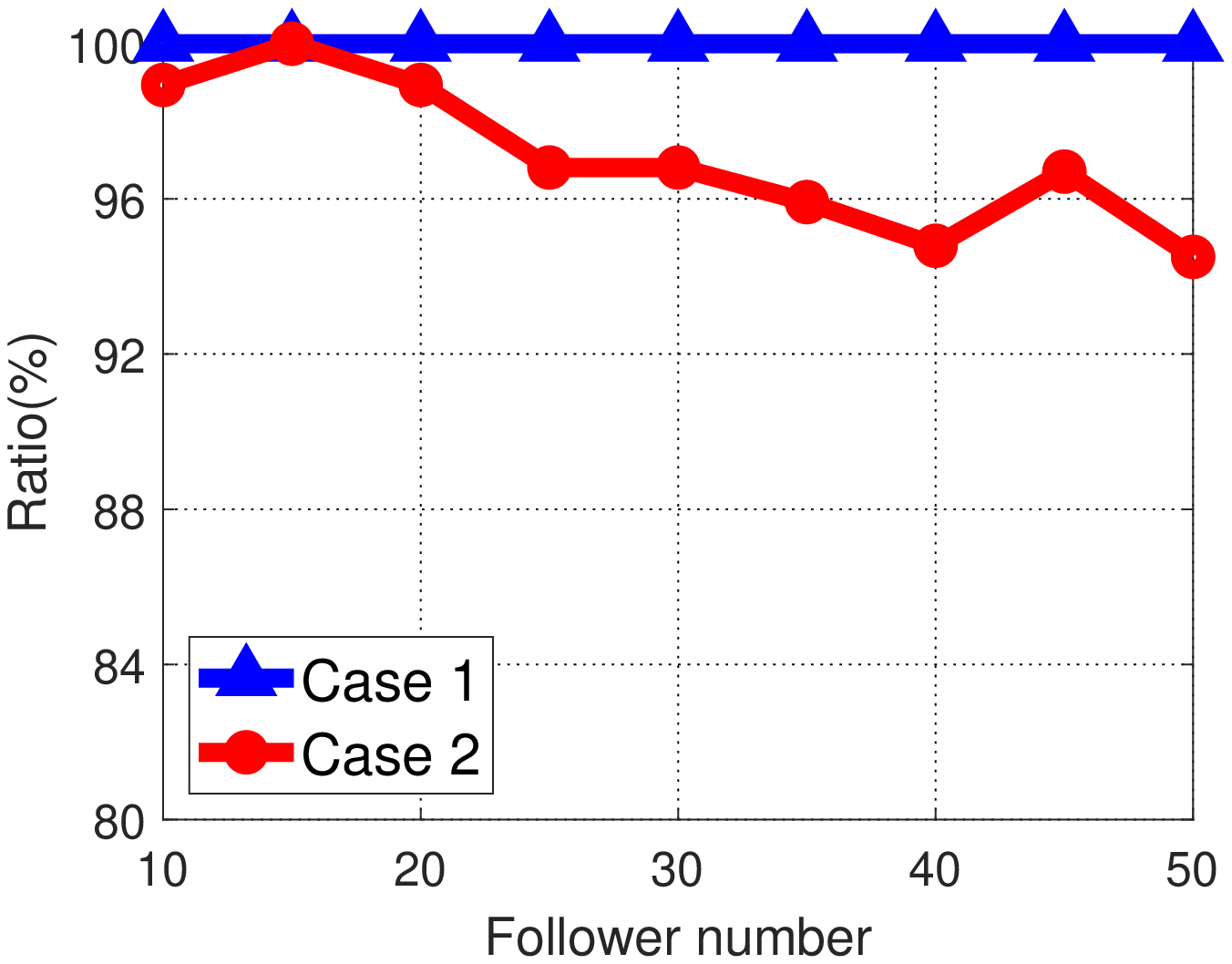}}
         \subfigure[$K=10$]
 { \label{fi::K10}
  \includegraphics[width=4.8cm]{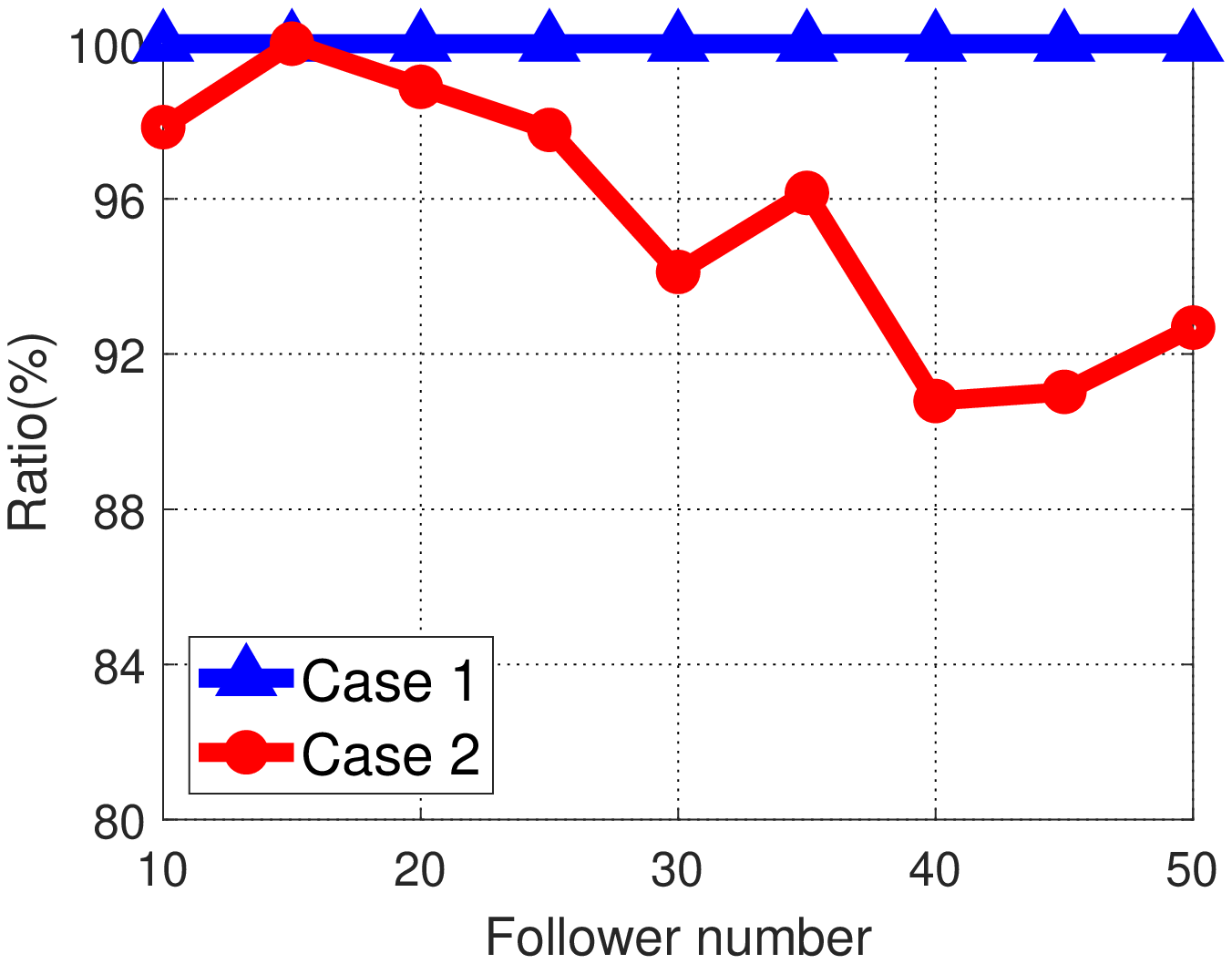}}
     \subfigure[$K=15$]
 { \label{fi::K15}
  \includegraphics[width=4.8cm]{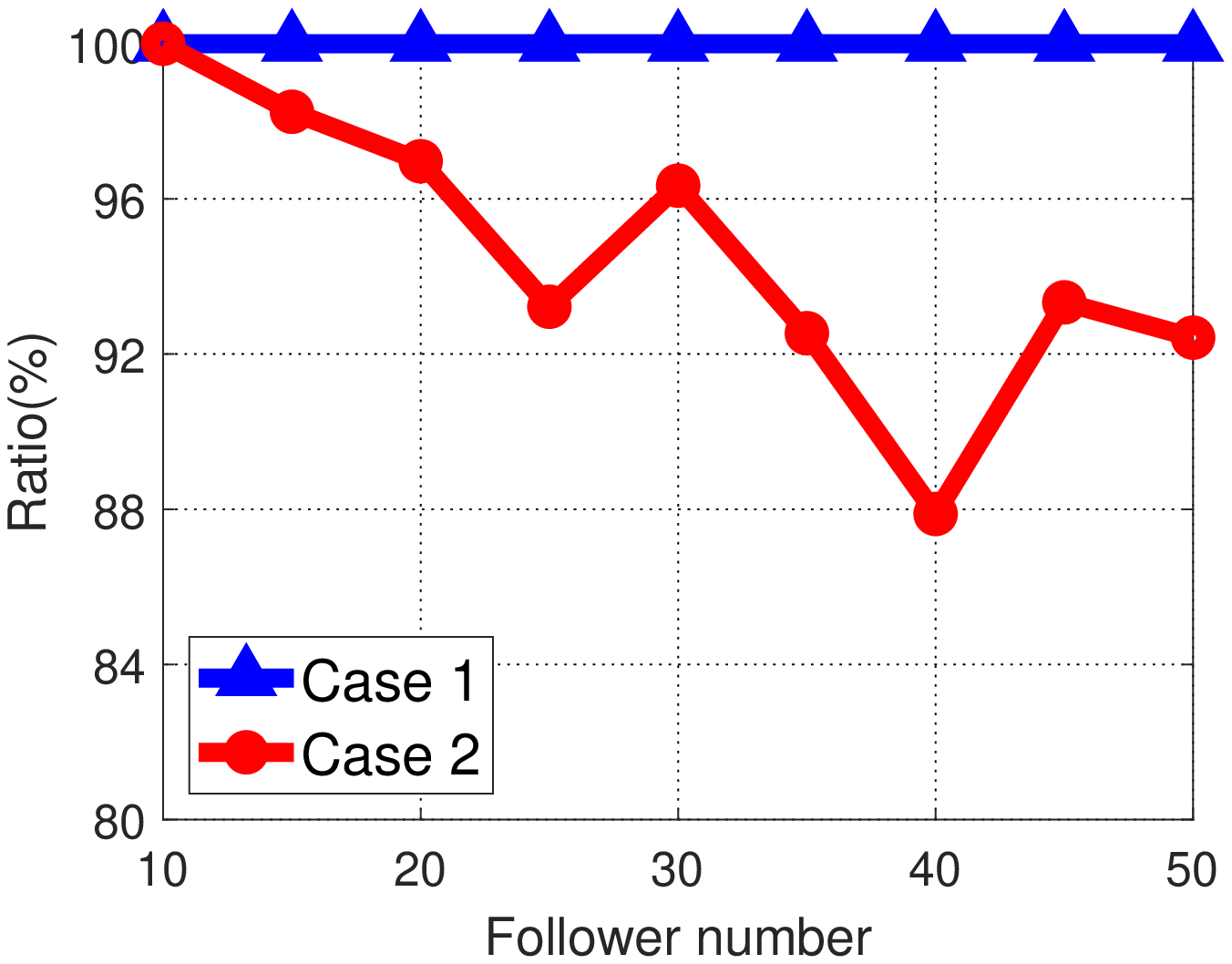}}
 \caption{Ratios of the two cases in misperception. 
The $x$-axis is for the target number in (a)-(c), and for the follower number in (d)-(f). $y$-axis is for ratios of \textit{Case 1} and \textit{Case 2} in 100 instances, which are depicted in blue and red, respectively.}
 \label{fig:subfig}
 \end{figure*}

\begin{mythm}\label{th::h2}
Under Assumptions \ref{as::ass4}-\ref{as::ass3}, if $K_{max}\in \mathop{\text{\emph{argmax}}}\limits_{k=1,\dots,K} U_l^c(t_k)$ and the leader is able to trick followers into attacking target $t_{K_{max}}$, then the corresponding DSSE  of $\mathcal{H}^2(\Theta)$ is a HNE.
\end{mythm}

Theorem \ref{th::h2} indicates that a DSSE strategy is stable in such a sufficient condition, since the leader has no will to change its manipulation on followers' cognitions. The deception brings the leader the most benefit since its DSSE strategy is the best response strategy. Also, similar to Theorem \ref{th::h1}, followers are not able to find that their observations of $\theta_0$ are misled, because the leader acts as they expect. Thus, followers do not update their observations, and the leader can safely deceive followers without being found out. Then the deception result is stable, as the discussion in \cite{HELLER2019223}. Additionally, no matter how many followers cannot observe the consequence of the leader's strategy, the leader can select a DSSE strategy because it is at least a HNE strategy for itself. Thus, the leader has no need to worry whether followers can observe its decision, which also consists with the analysis of two-players games in \cite{Bondi_Oh_Xu_Fang_Dilkina_Tambe_2020}.

\begin{remark}
Clearly, MSSE describes a situation caused by passive factors with players' biased cognitions, while  DSSE describes another situation caused by active factors among players with players' manipulated cognitions. Additionally, it is called stable with misperception when players do not realize the inherent misperception, while   it is called stable with deception when the leader has no will to change its manipulation on followers' cognitions. On the other hand, the   evaluation processes of   stability with misperception and deception  are different. Since deception happens in a cognitive set, we use   the deception's influences on players to investigate  the stability   of  DSSE, while for misperception, $\text{SOL}(\boldsymbol{y},\theta')$ gives how the  misperception affects players' strategies.
\end{remark}

\subsection{Typical Cases}
Here we investigate several typical cases to further explain the proposed stable conditions in   Theorems \ref{th::h1} and \ref{th::h2}, in order to show    them can be widely applied to many practical problems.

\subsubsection{Players' Perspective}

Consider the SLSF  game  in  \cite{korzhyk2011stackelberg}. The follower takes strategy $y_1\in {\Omega}_1$, and its utility function is $U_1(x,y,\theta')$. 
For any $\theta'\in\mathbb{R}^m$,   $\text{SOL}(y,\theta')$ can be converted to
$
\text{SOL}({y},\theta')=\left\{{y}'\in \mathbf{\Omega}_f, \lambda>0| A_1(\theta',1) \boldsymbol{y}'= \lambda \boldsymbol{y}\right\}.
$
It is easy to verify that  $\text{SOL}({y},\theta')$ is always nonempty. Then we have  the following result,  regarded as an extension of  the Theorem 3.9 in \cite{korzhyk2011stackelberg}.

\begin{corollary}\label{co::h1-2}
Under Assumptions \ref{as::ass1}  and \ref{as::ass3}, for any $\theta'\in \mathbb{R}^m$, any MSSE  is  a HNE of the SLSF game $\mathcal{H}^2(\theta')$. 
\end{corollary}

In addition, players' DSSE strategies are with a certain $\theta^*$, and the following result follows directly. 
\begin{corollary}\label{co::h1-3}
Under Assumptions \ref{as::ass4}-\ref{as::ass3},  any DSSE   is a HNE of the SLSF  game $\mathcal{H}^2(\Theta)$. 
\end{corollary}

\subsubsection{Targets' Perspective}

Consider the case that all followers prefer to attack the same target. Notice that our SLMF game is with independent targets. 
Attacks on one target do not affect others.  If followers attack the same target independently, $A_1(\theta') \boldsymbol{y}'= \lambda B\boldsymbol{y}$ covers the solution to $A_2 \boldsymbol{y}' =0$. Thus, $\text{SOL}(\boldsymbol{y},\theta')$ can be converted to
\begin{equation}\label{eq::case2-now}
\begin{aligned}
\text{SOL}(\boldsymbol{y},\theta')=\left\{\boldsymbol{y}'\in \mathbf{\Omega}_f, \lambda>0| A_1(\theta') \boldsymbol{y}'= \lambda B\boldsymbol{y}\right\}.
\end{aligned}
\end{equation}
Also, it is easy to see that  (\ref{eq::case2-now}) is always nonempty. Then we get the following result,   which  consists with the `near decomposability'  \cite{hausken2011protecting}.

\begin{corollary}\label{co::h2-2}
Under Assumptions \ref{as::ass1} and \ref{as::ass3}, for any $\theta'\in \mathbb{R}^m$,  a MSSE  of $\mathcal{H}^2(\theta')$ is also a HNE if all followers attack the same target.
\end{corollary}

\subsubsection{With Same Perception}

If there is no cognitive difference, followers' observations of $\theta_0$ are true, \textit{i.e.}, $\theta'=\theta_0$, and all players are involved in an identical game $\mathcal{G}$. Then the model turns into the SLMF  game in \cite{feng2019using}. 
Take $(x_{\textit{\tiny SSE}},\boldsymbol{y}_{ \textit{\tiny SSE}})$ as the SSE   of $\mathcal{G}$. The next result reveals a relationship between SSE and NE with the same perception.

\begin{corollary}\label{co::h3}
Under Assumptions \ref{as::ass4}, \ref{as::ass1} and \ref{as::ass3}, if $\text{\emph{SOL}}(\boldsymbol{y}_{\textit{\tiny SSE}},\theta_0)$ is nonempty, then $(x_{\textit{\tiny SSE}},\boldsymbol{y}_{ \textit{\tiny SSE}})$  is also a NE of $\mathcal{G}$.
\end{corollary}

\section{Robustness analysis}

\begin{figure*}
\centering

\subfigure[$a=0.2$]{
\label{fi::mis_compare_not}
 \includegraphics[width=4.8cm]{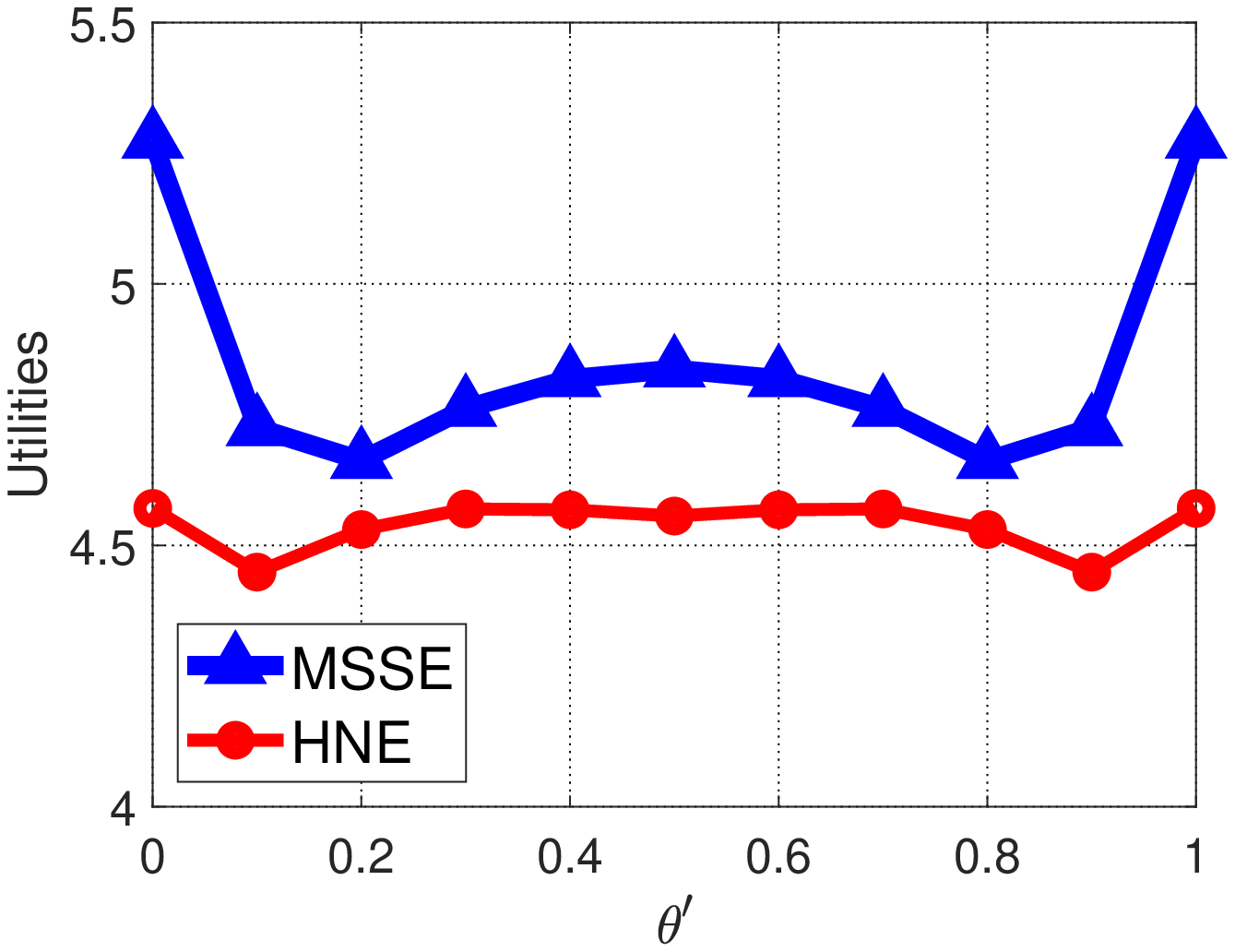}}
\subfigure[$a=0.3$]{
\label{fi::mis_compare_1}
 \includegraphics[width=4.8cm]{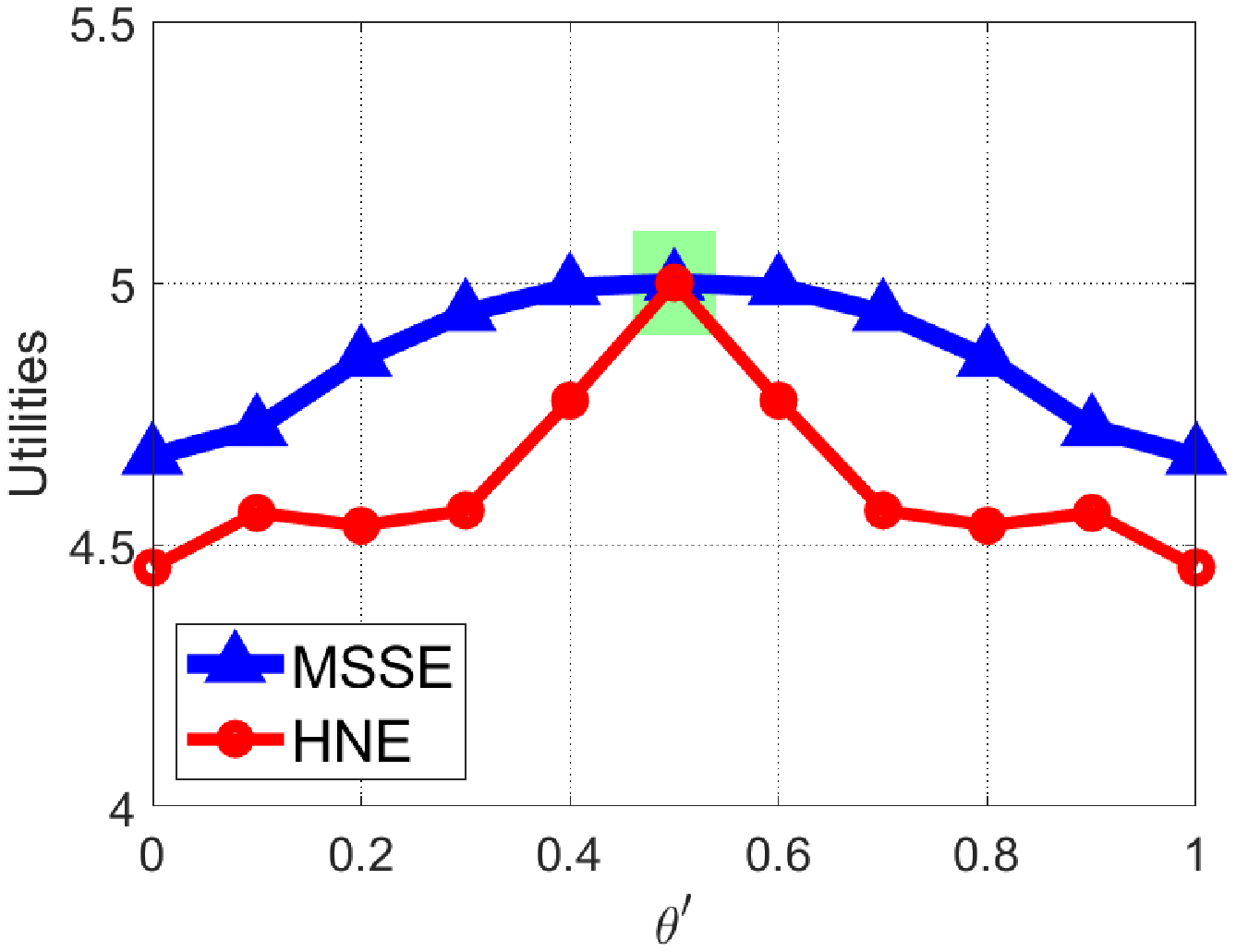}}
\subfigure[$a=0.4$]{
\label{fi::mis_compare_more}
 \includegraphics[width=4.8cm]{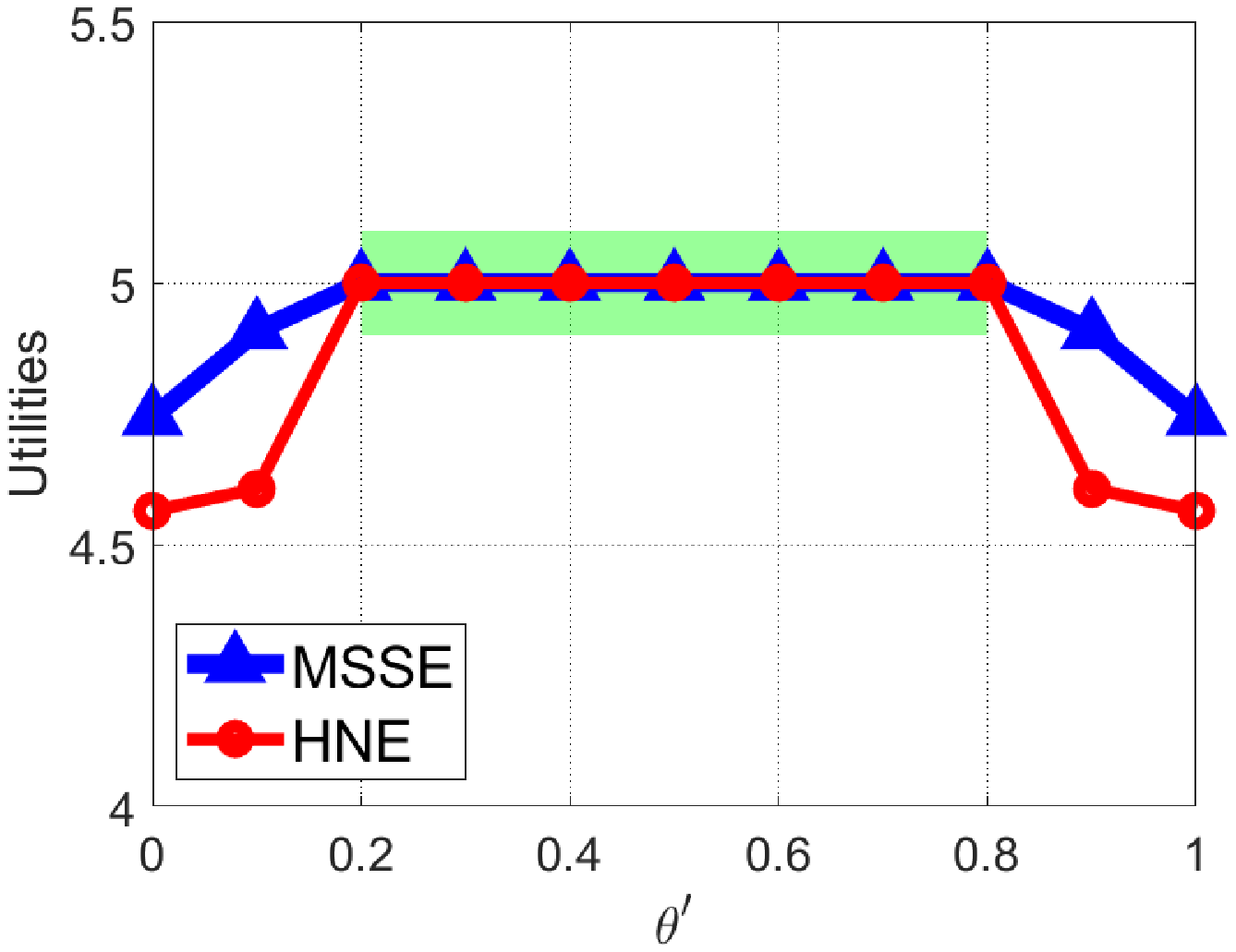}}
 \caption{To find cognitively stable MSSE with different environment settings. The blue line describes the leader's utility $U_L$ of MSSE with different $\theta'\in \Theta$, while the red line stands for the leader's utility $U_L$ of HNE with different $\theta'\in \Theta$. The light green region shows when  MSSE is HNE.
}
 \label{fig:mis_com}
 \end{figure*}

In this section, we   discuss the robustness of  MSSE and   DSSE. As a complement  to  HNE,  we focus on the misinformation's influence on players' actual utility functions, which refers to players'  capacities to keep their profits.

Conveniently, for any $x\in\Omega_l,\theta\in \Theta,i\in\mathbf{P}, k=1,\dots,K$, let 
$g_i(x,\theta, k)=x^k U_i^c(\theta,t_k)+(R_l-x^k) U_i^u(\theta,t_k). $ Correspondingly, denote $$\Gamma^1_i(x,\theta)=\mathop{\text{argmax}}\limits_{k=1,\dots,K}g_i(x,\theta, k),$$ $$\Gamma^2_i(x,\theta)=\mathop{\text{argmax}}\limits_{k=1,\dots,K,k\notin \Gamma^1_i(x,\theta)}g_i(x,\theta, k).$$ Intuitively,  $\Gamma^1_i(x,\theta)$ and $\Gamma^2_i(x,\theta)$ represent the corresponding target sets of the  two  most attractive utilities to the $i$th follower under the leader's strategy $x$ and the observation $\theta$. Moreover,  let
\begin{equation}\label{eq::gri}
\begin{aligned}
\hat{g}_i^r=g_i(x_{\textit{\tiny SSE}},\theta_0,k),k\in \Gamma^r_i(x_{\textit{\tiny SSE}},\theta_0),r\in\{1,2\},
\end{aligned}
\end{equation}
\begin{equation}\label{eq::deltei}
\begin{aligned}
\bigtriangledown^*_i=\max\limits_{k\in\Gamma^1_i(x_{\textit{\tiny SSE}},\theta_0)}\parallel\bigtriangledown_{\theta}g_i(x_{\textit{\tiny SSE}},\theta_0,k)\parallel.
\end{aligned}
\end{equation}

\subsection{For MSSE}

In this situation, the imprecise observation  mainly affects followers, and  additionally, the followers' decisions under  different observations also  reflect the game's performance of resisting misperception. For instance, in cyber-physical security problems,  external perturbation influences the followers' profits through  their observations \cite{bakker2020hypergames}. If the followers' profits do not change under the perturbation, the game has a strong  anti-jamming capacity. In the view of bounded rationality such as computational constraints, emotion, and habitual thoughts, players have an inherent observation error  \cite{8691466}. If the profits of followers remain unchanged under bounded rationality, the game is said to be robust to the inherent systemic uncertainty. Moreover, for the accidental error, it reveals the tolerance of the model for the random internal uncertainty, if the accidental error  does not change the followers' profits \cite{kraemer2007human}.

For $\theta'\in\mathbb{R}^m$, let $\boldsymbol{y}_{\textit{\tiny MSSE}}(\theta')$ be the followers' the MSSE strategy of $\mathcal{H}^2(\theta')$ and $(x_{\textit{\tiny SSE}},\boldsymbol{y}_{ \textit{\tiny SSE}})$ be the SSE   of $\mathcal{G}(\theta_0)$. We are interested in the subset $\delta_{\theta}\subseteq\Theta$ such that
\begin{equation}\label{eq::robustofMSSE}
U_i\big(x_{\textit{\tiny SSE}},(\boldsymbol{y}_{\textit{\tiny SSE}})_i,\theta_0\big)\!=\!U_i\big(x_{ \textit{\tiny SSE}},(\boldsymbol{y}_{ \textit{\tiny MSSE}})_i(\theta'),\theta_0\big),\forall \theta\in\delta_\theta, i\in\mathbf{P}\!,
\end{equation}
which is regarded as the robustness set of MSSE. The following result shows the robustness of MSSE, whose  proof  can be found in Appendix \ref{ap::th::robust-followers-2}.

\begin{mythm}\label{th::robust-followers-1}
Under Assumptions \ref{as::ass4}-\ref{as::ass3},

\noindent 1) there exists a convex subset $\delta_{\theta}\subseteq\Theta$ satisfying (\ref{eq::robustofMSSE}) with nonempty $\emph{int}(\delta_{\theta})$; 

\noindent 2) moreover,  if $U_i^c(\theta,t_k)$ and $U_i^u(\theta,t_k)$ are convex and $\varsigma$-Lipschitz continuous  in $\theta\in \Theta$ for all $i\in\mathbf{P},k=1,\dots K $,  there exists   $\delta_\theta=\left\{\theta\in\Theta:\parallel\theta-\theta_0\parallel<\Delta\theta\right\}$ satisfying (\ref{eq::robustofMSSE}) such that
$$
\begin{aligned}
\Delta\theta\!=\!\min\limits_{i\in\mathbf{P}}\frac{\hat{g}_i^1- \hat{g}^2_i}{\bigtriangledown^*_i+\varsigma R_l},
\end{aligned}
$$
where $\hat{g}_i^1$ and $\hat{g}_i^2$ are from (\ref{eq::gri}), and $\bigtriangledown^*_i$ is according to (\ref{eq::deltei}).

\end{mythm}

Conclusion 1) of Theorem \ref{th::robust-followers-1} shows that there is always a nonempty subset of the observation parameter such that the MSSE   is robust for the followers.
Additionally, 2) of Theorem \ref{th::robust-followers-1} gives a lower bound  for players  to ignore the misperception  if the game model satisfies the convexity and Lipschitz continuity.


\subsection{For DSSE}

In the deception situation, the leader deceives followers by manipulating followers’ observations, and followers are unaware of the deception. Obviously, the leader will not deceive if the implementation does not increase its own profit \cite{ijcai2017-516}, and it always needs to spend energy  for deception \cite{7577775}. Therefore,  the deceiver decides to cheat  when the rewards exceed the lower bound of the deceptive energy. Moreover, the ridiculous and outrageous  deception may cause followers' suspicions, which may   lead to the collapse of  the model \cite{ortmann2002costs}.

Let $(x_{\textit{\tiny DSSE}},\boldsymbol{y}_{\textit{\tiny DSSE}},\theta^*)$ be the DSSE   of $\mathcal{H}^2(\delta_{\theta})$ for the deceptive set $\delta_{\theta}$.
We are interested in the subset $\delta_{\theta}\subseteq\Theta$ such that
\begin{equation}\label{eq::robustDSSE}
U_l(x_{\textit{\tiny SSE}},\boldsymbol{y}_{\textit{\tiny SSE}})=U_l(x_{\textit{\tiny DSSE}},\boldsymbol{y}_{\textit{\tiny DSSE}}),\ \forall \theta\in\delta_\theta,
\end{equation}
which is regarded as the robustness set of DSSE. The following theorem reveals the robustness of DSSE, whose  proof  can be found in Appendix \ref{ap::th::robust-leader-1}.

\begin{mythm}\label{th::robust-leader-1}
Under Assumptions  \ref{as::ass4}-\ref{as::ass6},

\noindent 1) there exists a convex subset $\delta_{\theta}\in\Theta$ satisfying (\ref{eq::robustDSSE}) with nonempty $\emph{int}(\delta_{\theta})$;

\noindent 2) moreover, if $U_i^c(\theta,t_k)$ and $U_i^u(\theta,t_k)$ are convex and $\varsigma$-Lipschitz continuous  in $\theta\in \Theta$ for all $k=1,\dots K, i\in\mathbf{P}$, there exists $\delta_\theta=\left\{\theta\in\Theta:\parallel\theta-\theta_0\parallel<\Delta\theta\right\}$ satisfying (\ref{eq::robustDSSE}) such that
$$
\begin{aligned}
\Delta\theta=\min\limits_{i\in\mathbf{P}}\frac{\hat{g}_i^1- \hat{g}^2_i}{2\varsigma R_l},
\end{aligned}
$$
where $\hat{g}_i^1$ and $\hat{g}_i^2$ are from (\ref{eq::gri}), and $\bigtriangledown^*_i$ is according to (\ref{eq::deltei}).
\end{mythm}

 \begin{figure*}
\centering
\subfigure[$n=1$]{
\label{fi::Dn1}
 \includegraphics[width=4.8cm]{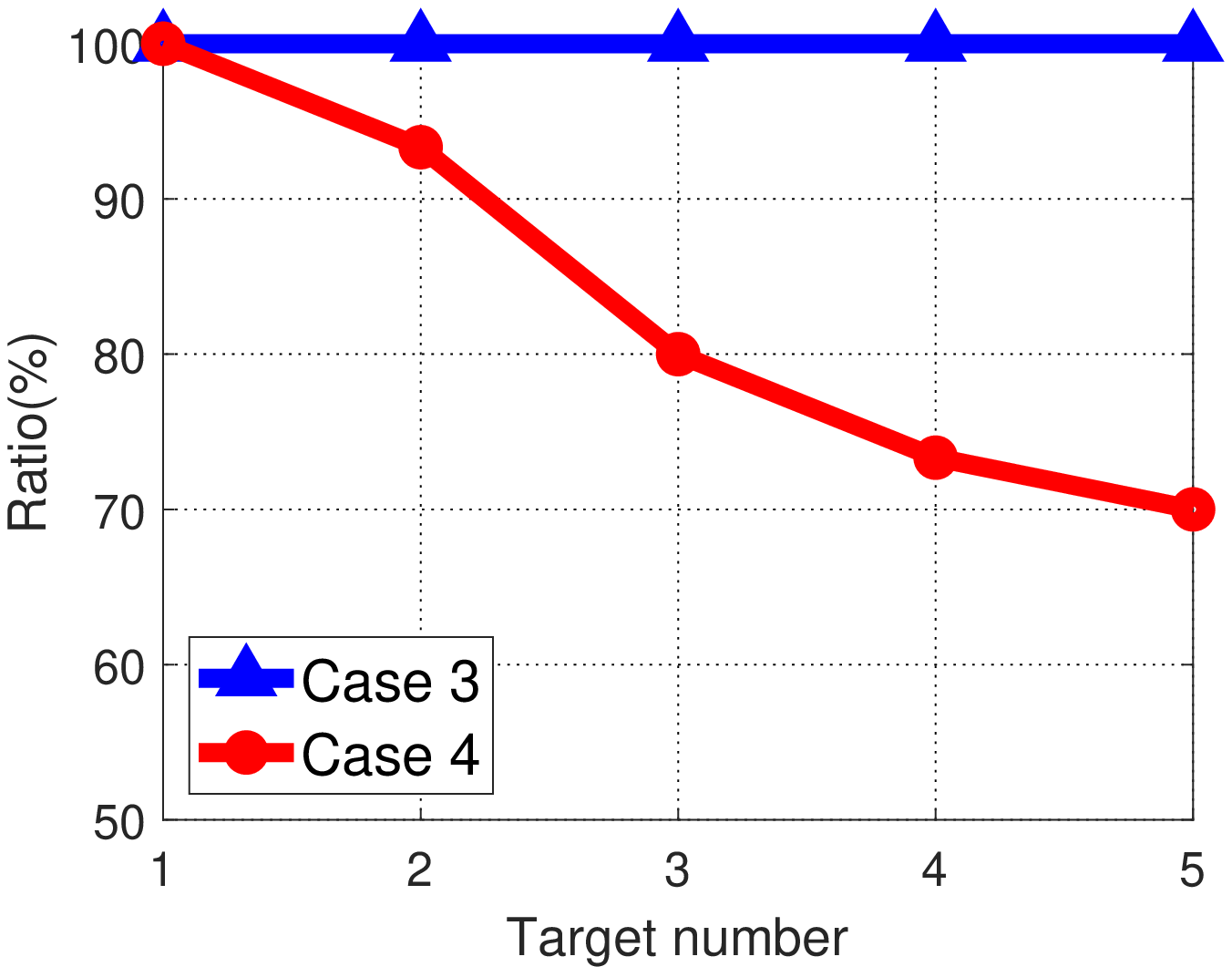}}
 \subfigure[$n=3$]
 { \label{fi::Dn3}
  \includegraphics[width=4.8cm]{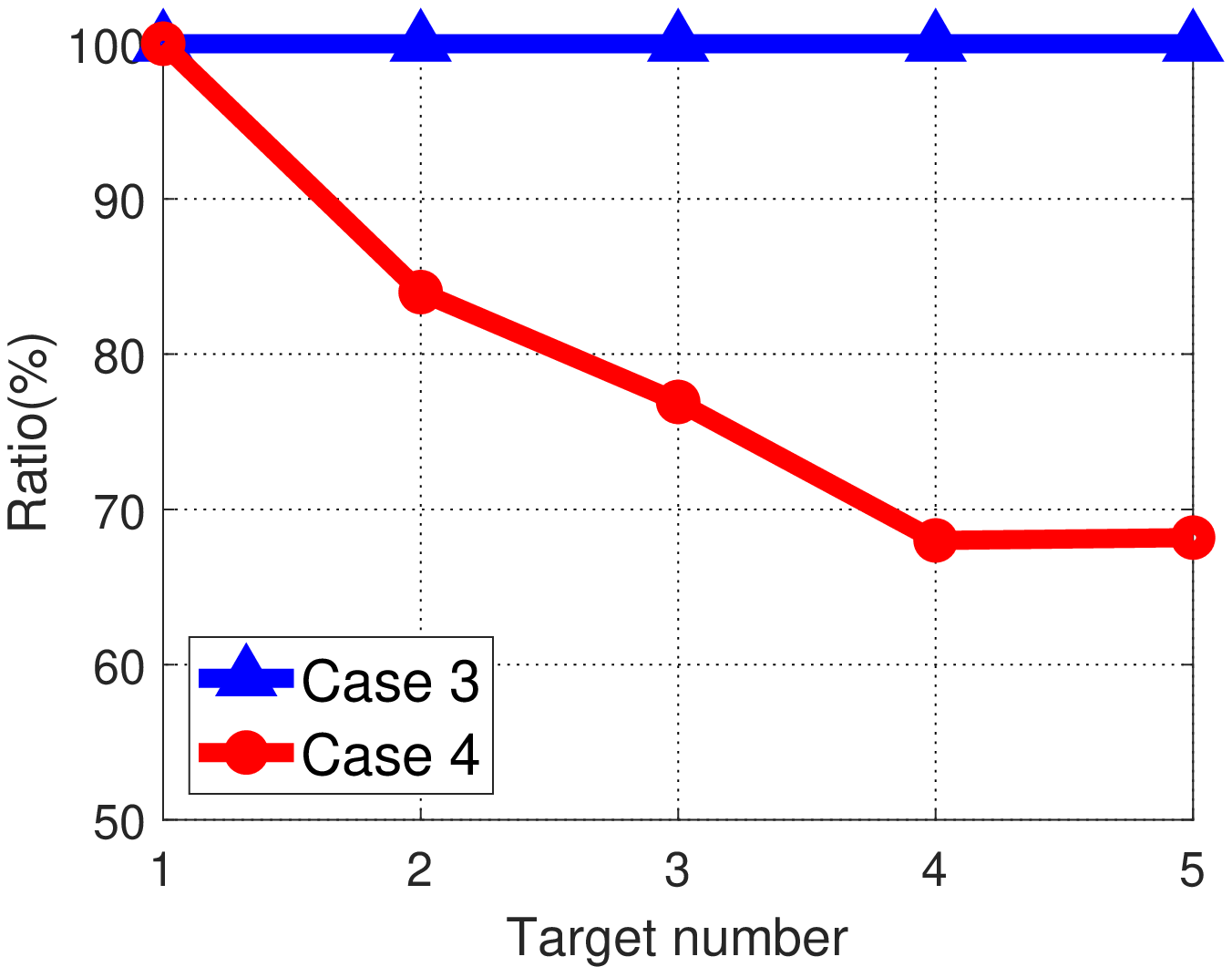}}
   \subfigure[$n=5$]
 { \label{fi::Dn5}
  \includegraphics[width=4.8cm]{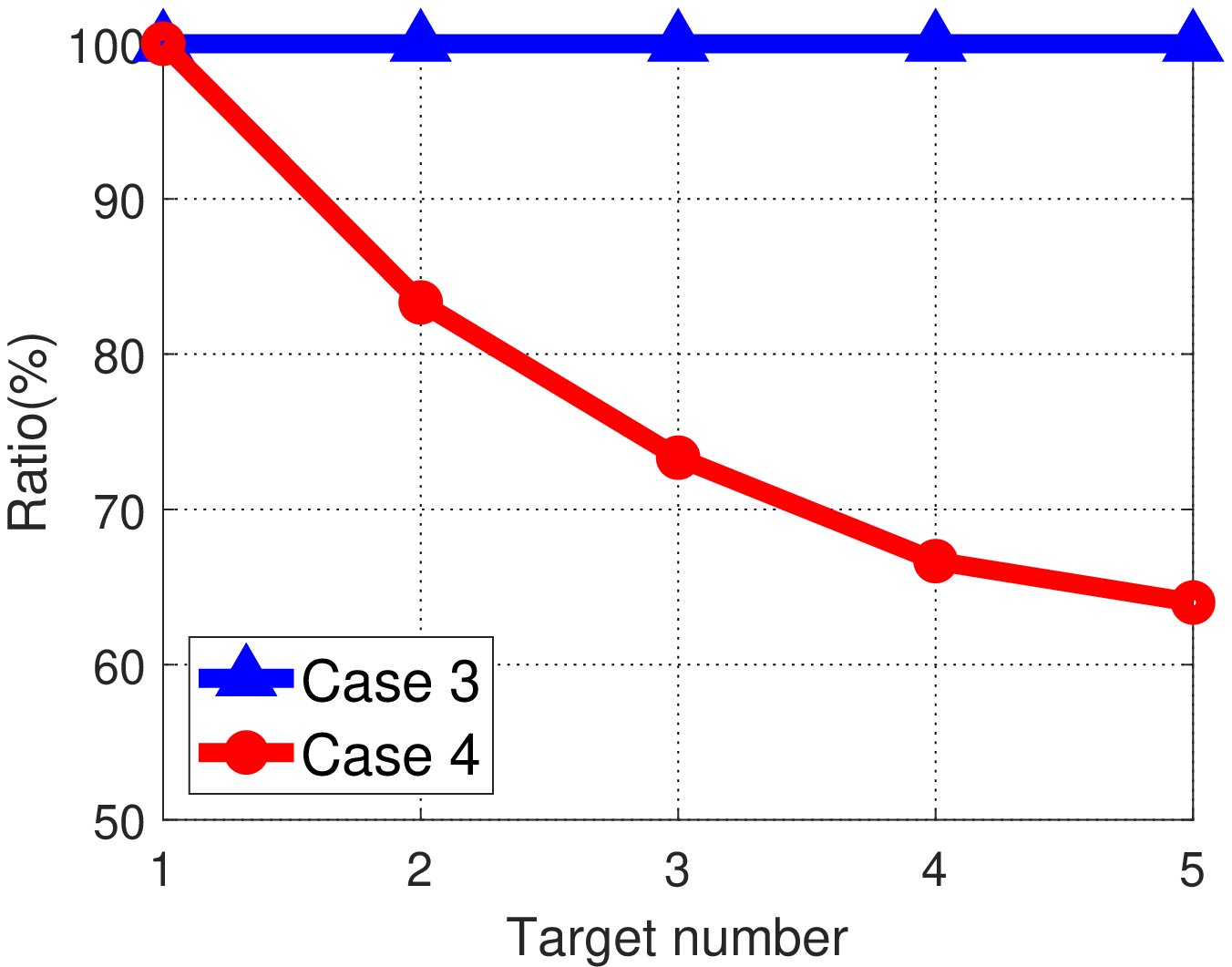}}
       \subfigure[$K=1$]
 { \label{fi::DK1}
  \includegraphics[width=4.8cm]{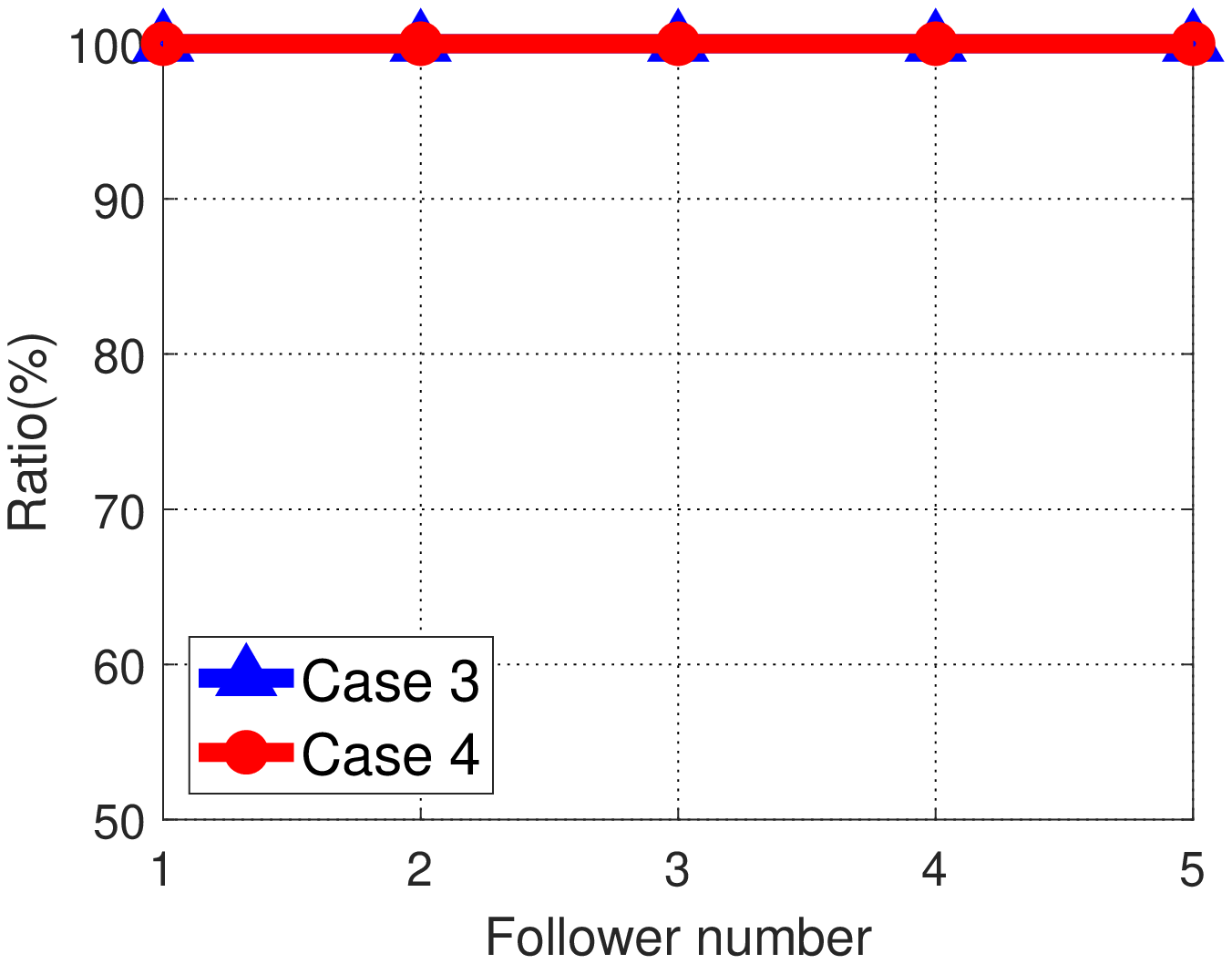}}
         \subfigure[$K=3$]
 { \label{fi::DK3}
  \includegraphics[width=4.8cm]{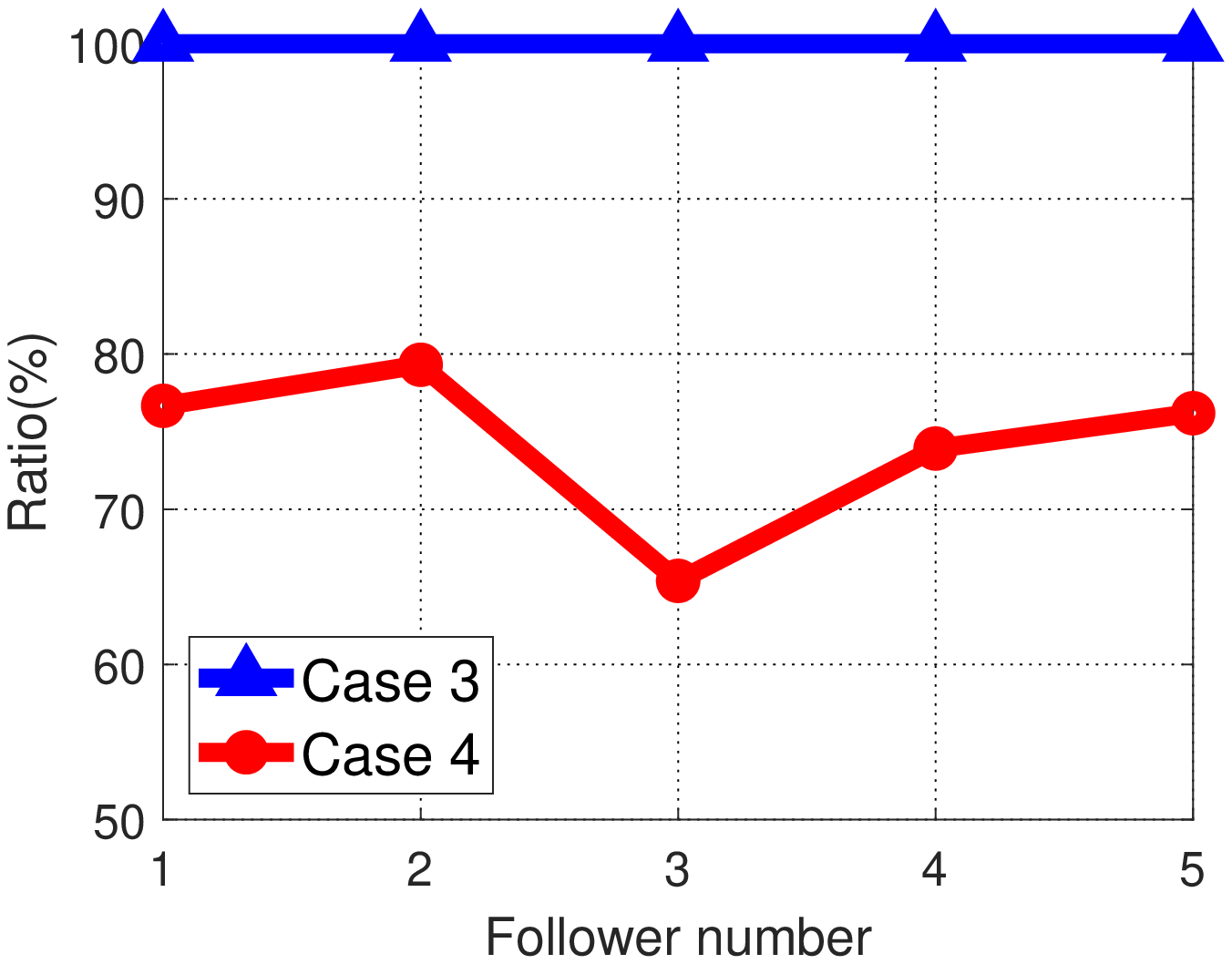}}
     \subfigure[$K=5$]
 { \label{fi::DK5}
  \includegraphics[width=4.8cm]{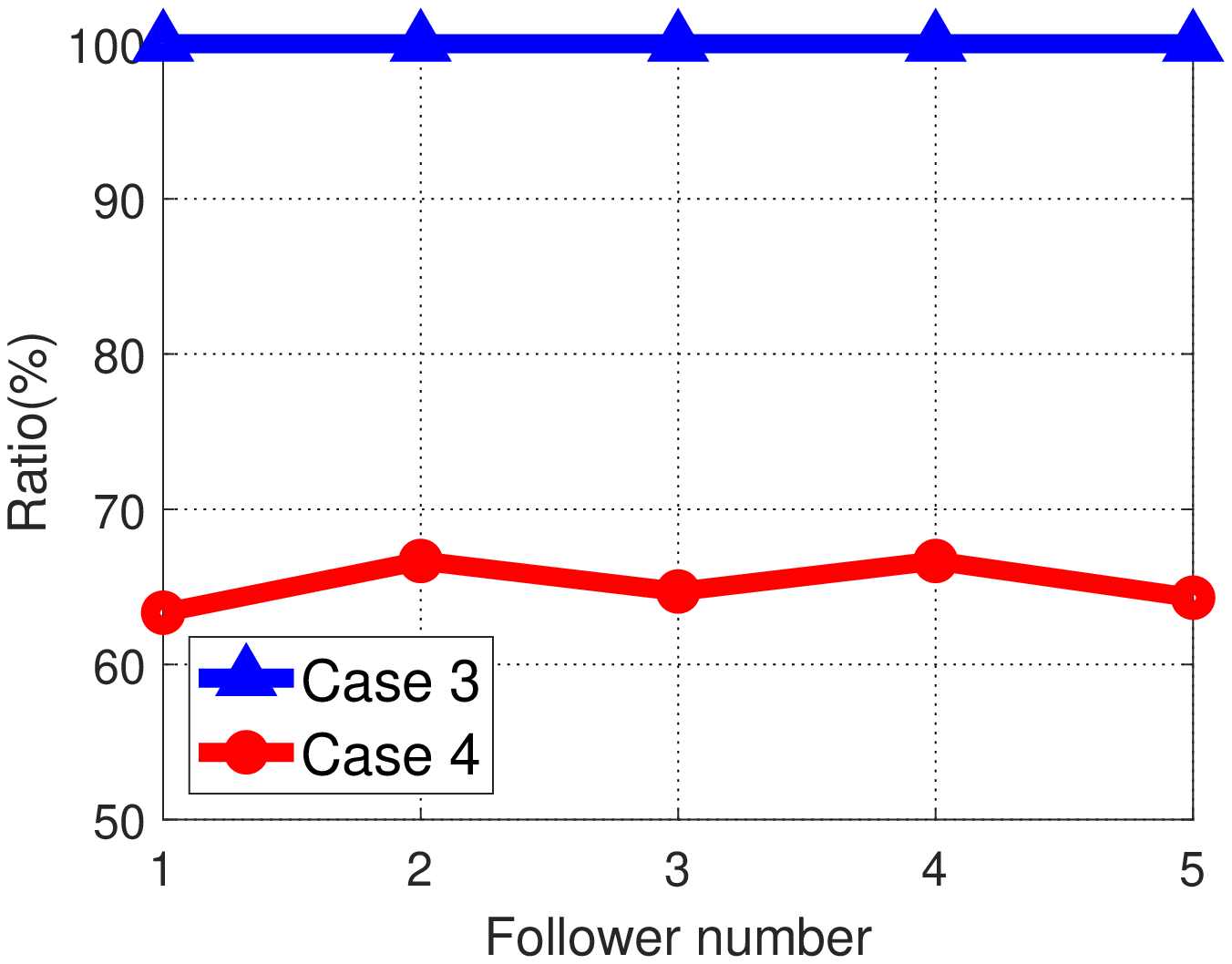}}

 \caption{Ratios of the two cases in deception.
The $x$-axis is for the target number in (a)-(c), and for the follower number in (d)-(f). $y$-axis is for ratios of \textit{Case 3} and \textit{Case 4} in 30 instances, which are depicted in blue and red, respectively.}
 \label{fig:subfig_DSSE}
 \end{figure*}


Conclusion 1) of Theorem \ref{th::robust-leader-1}  shows that there is always a nonempty subset of the observation parameter such that the leader does not  implement deception, since  tiny deception does not bring the leader more benefits.
Moreover, 2) of Theorem  \ref{th::robust-leader-1}  gives a lower bound if utility functions satisfy the convexity and Lipschitz continuity. From the perspective of energy, if the leader wants more benefits from deception,   it needs to   pay    energy  no less than the lower bound $\Delta\theta$. Therefore, it can be regarded as a tradeoff for the leader.

\begin{remark}
The robustness of the MSSE indicates that followers can ignore the misperception, while the robustness of the DSSE means the leader does not  implement deception in this region. On the other hand, the proofs of robustness under misperception  and deception are  different.  We consider Assumption \ref{as::ass6} in the robustness analysis of deception since the deceptive strategy is affected by followers' utility functions in different targets. Moreover, the proof for the misperception focuses on followers' profits under the fixed leader's strategy, while the proof to handle the deception counts in   the influence of the leader's deceive strategy on followers' actions.
\end{remark}

\section{Experiment}

In this section, we provide  numerical simulations for the stability and  robustness of MSSE and DSSE.

\subsection{Stable Condition for MSSE}

1) Inspired by the single-leader-single-follower game in  infrastructures protection problems\cite{korzhyk2011stackelberg}, with  misperception $\theta'$, we verify Theorem \ref{th::h1} by a numerical simulation. 
We consider models for  $K=10,15,\dots,50$ when $n=5,10,15$,  and other models for $n=10,15,\dots,50$ when $K=5,10,15$, respectively.  In each model, we randomly generate $100$ instances as follows. For the leader, $U_l^c(t_k)$ and $U_l^u(t_k)$ are uniformly  generated in the ranges $[5,10]$ and $[0,5]$, while for the $i$th follower,  $U_i^c(\theta',t_k)$ and $U_i^u(\theta',t_k)$ are uniformly  generated in the ranges $[0,5]$ and $[5,10]$. $R_l$ and $R_i$ are uniformly  generated in the range $[1,5]$. Moreover,  we compute MSSE    of $\mathcal{H}^2(\theta')$ by the extension of the mixed-integer linear program \cite{korzhyk2011stackelberg}:
\begin{equation}\label{eq::MILP-MSSE}
\begin{aligned}
\max\limits_{x,\boldsymbol{y},a}  &\sum\limits_{k=1}^K (\sum\limits_{i=1}^n  y^k_iR_i)(x^k U_l^c(t_k)\!+\!(R_l-x^k) U_l^u(t_k)),\\
s.t. \
&0\leqslant  a_i-R_ig_i(x,\theta',k) \leqslant (1-y_i^k)M,
\\
&\sum\limits_{k=1}^K x^k=R_{l},x^k\geqslant 0,
 \sum\limits_{k=1}^K y_i^k=1,y_i^k\in\{0,1\},
\\
& a=[a_1,\dots,a_n]^T\in \mathbb{R}^n,
\forall i\in\mathbf{P},k\!=\!1,\dots\!,K,
\end{aligned}
\end{equation}
where $M=10^9$ is a sufficiently large number.
Take the MATLAB toolbox YALMIP [48] to solve (\ref{eq::MILP-MSSE}) with the  terminal condition $\frac{\overline{U}^q-{\underline{U}}^q}{|\underline{U}^q|}<10^{-6}$, where $\overline{U}^q$ and $\underline{U}^q$ are the upper and lower bounds of the objective function in $q$th iteration. 
Set $(x_{\textit{\tiny MSSE}},\boldsymbol{y}_{\textit{\tiny MSSE}})$ as the MSSE strategy of  each  instance.



\textit{Case 1}:  $(x_{\textit{\tiny MSSE}},\boldsymbol{y}_{\textit{\tiny MSSE}})$ is a HNE when $\text{SOL}(\boldsymbol{y}_{\textit{\tiny MSSE}},\theta')$ is nonempty.

 \textit{Case 2}:  $\text{SOL}(\boldsymbol{y}_{\textit{\tiny MSSE}},\theta')$ is nonempty when $(x_{\textit{\tiny MSSE}},\boldsymbol{y}_{\textit{\tiny MSSE}})$ is a HNE.

In Fig.  \ref{fig:subfig}, the ratio of \textit{Case 1} is always $100\%$, which verifies Theorem  \ref{th::h1}. Also, the ratio of \textit{Case 2} is always larger than $85\%$. Therefore,   when $(x_{\textit{\tiny MSSE}},\boldsymbol{y}_{\textit{\tiny MSSE}})$ is a HNE,   the stable condition of Theorem  \ref{th::h1} can cover most instances.

 \begin{figure*}
\centering

   \subfigure[$U^c_l(t_1)=3,U^c_l(t_2)=6.$]
 { \label{fi::deception_compare_not}
  \includegraphics[width=4.8cm]{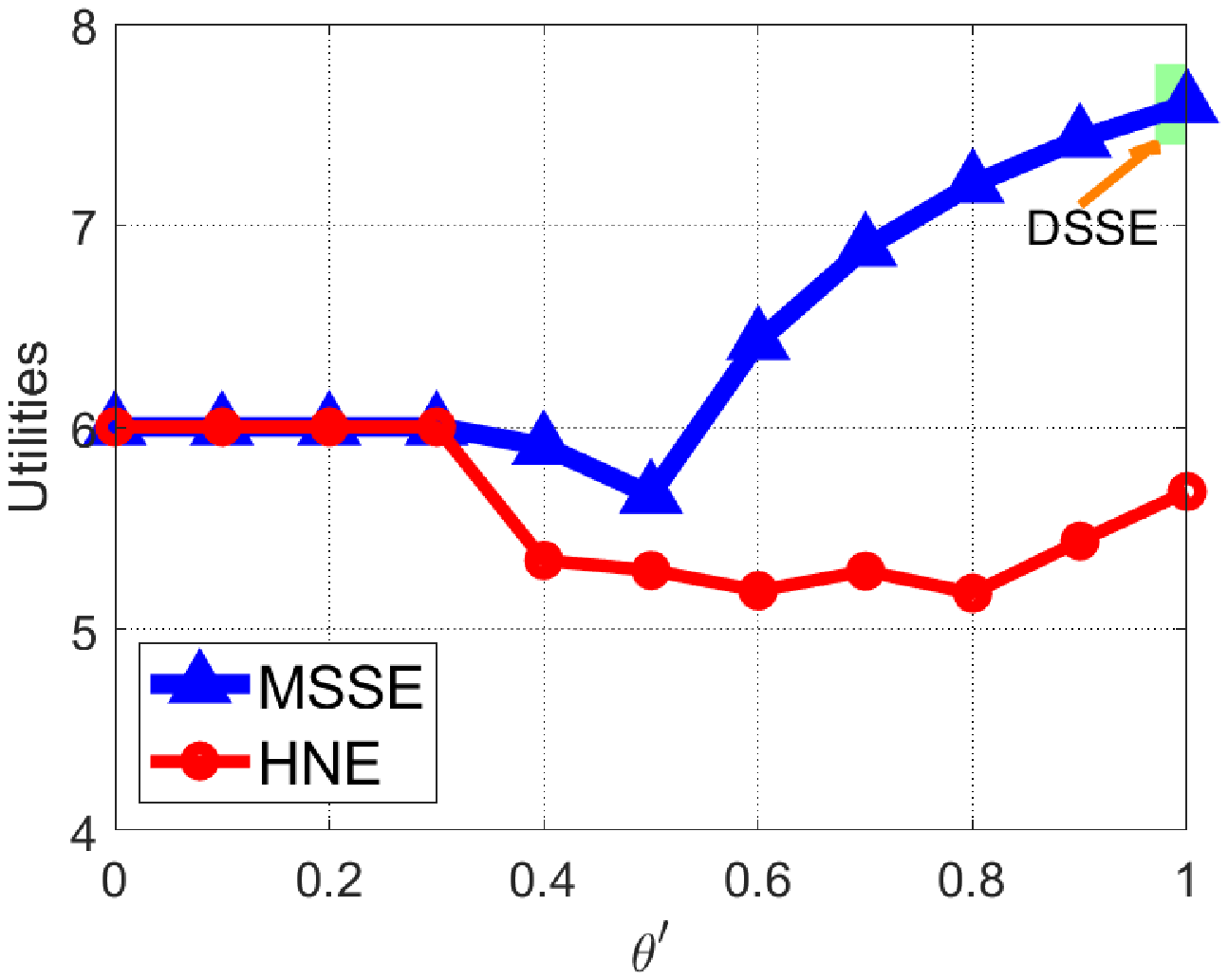}}
\subfigure[$U^c_l(t_1)=3.2,U^c_l(t_2)=2.$]{
\label{fi::deception_compare_more}
 \includegraphics[width=4.8cm]{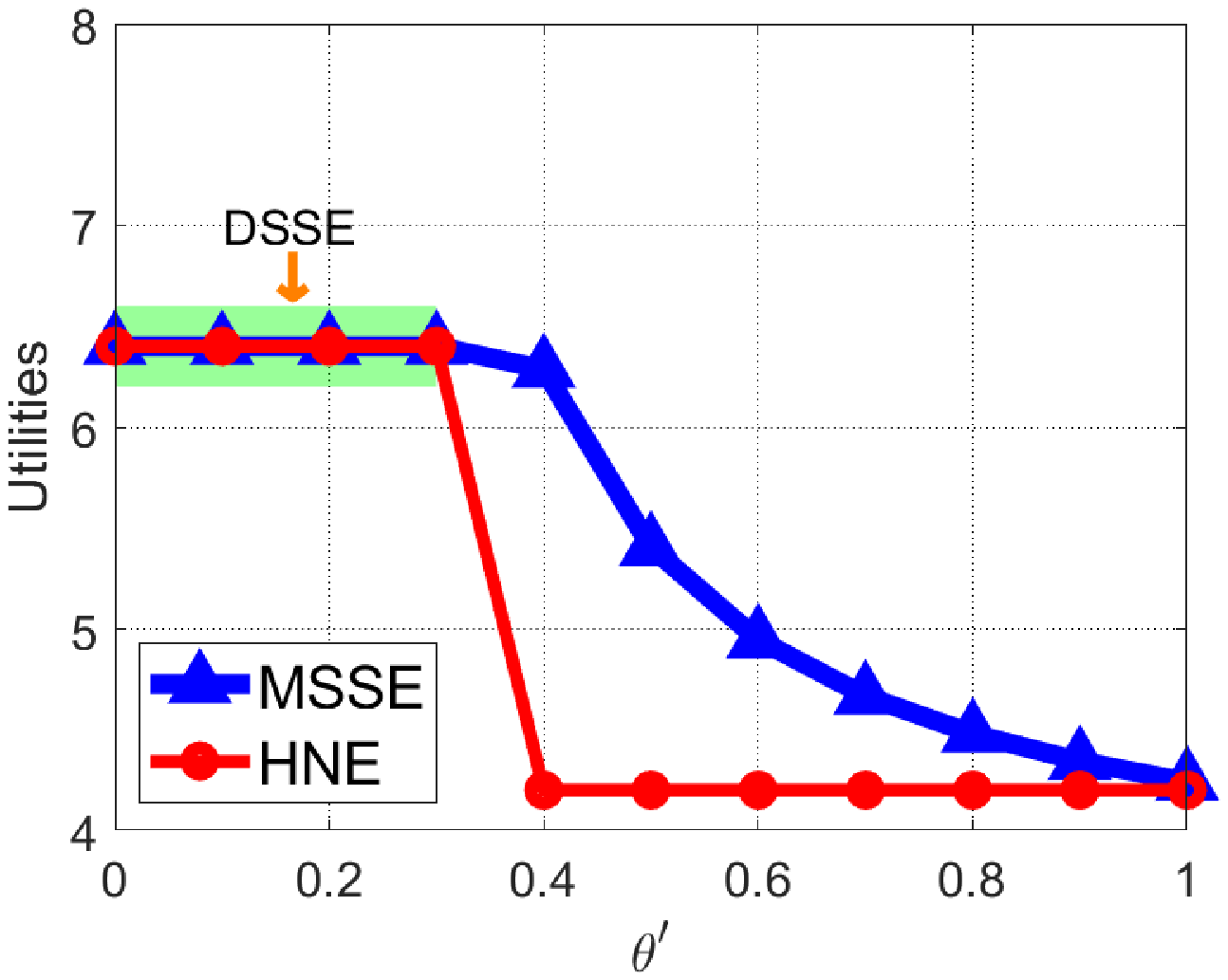}}
 \caption{To find cognitively stable DSSE with different environment settings. 
The blue line describes for the leader's utility $U_L$ of MSSE with different $\theta'\in \Theta$, while the red line describes the leader's utility $U_L$ of HNE with different $\theta'\in \Theta$. Besides, 
 the light green region exhibits $U_L$ of DSSE since the leader aims to maximize its own utility among all possible $\theta'$. 
}
 \label{fig:deception_compare}
 \end{figure*}

2) Consider a single-leader-two-followers model in MTD problems \cite{feng2017signaling}. Take $n=K=2$, $R_l=R_i=1$, $\Theta=[0,1]$,  $
U^c_2(\theta',t_1)\!=\!0.041((\theta'-0.5)^2-10+a)^2+4.305,$
$U^c_2(\theta',t_2)\!=\!0,$
$U^u_2(\theta',t_1)=-0.05((\theta'-0.5)^2+10-a)^2+5.1532$, and 
 $U^u_2(\theta',t_2)=-0.004((\theta'-0.5)^2-10+a)^2+0.82$, where $a\in \mathbb{R}$ is a parameter in attackers' migration cost. 
Set $a=0.2$, $0.3$, and $0.4$ in Fig. \ref{fi::mis_compare_not}, \ref{fi::mis_compare_1}, and \ref{fi::mis_compare_more}, respectively. 
In Fig. \ref{fi::mis_compare_not}, MSSE is always not HNE, and no cognitively stable MSSE can be found by Theorem \ref{th::h1}. 
Further, in Fig. \ref{fi::mis_compare_1}, there is only one cognitively stable MSSE when $\theta'=0.5$. 
In this case, it is usually hard for the player to reach the cognitively stable MSSE in MTD problems \cite{feng2017signaling}, since the probability for finding such a singleton is zero. However, by verifying the stable condition in Theorem \ref{th::h1}, we obtain a stable MSSE precisely and conveniently.  Fig. \ref{fi::mis_compare_more} shows a similar result, and we can improve the efficiency to find a cognitively stable MSSE  once the stable condition in Theorem \ref{th::h1} is verified.

\subsection{Stable Condition for DSSE}
1) Similar to security problems in deployed systems \cite{ijcai2019-75}, we verify Theorem \ref{th::h2} by a numerical simulation.  We consider models for  $K=1,\dots,5$ when $n=1,3,5$,  and other models for $n=1,\dots,5$ when $K=1,3,5$, respectively.  
 In each model, we randomly generate $30$ instances as follows. $U_l^u(t_k),U_i^c(t_k),R_i$, and $R_l$ are uniformly  generated in the range $[5,10]$, and $U_l^c(t_k),U_i^u(t_k)$ are uniformly  generated in the range $[0,5]$. Take $\Theta=[0,5]^{nk}\subset \mathbf{R}^{nK}$. Concretely, for any $\theta'\in\Theta$, $\theta'=[\theta'_{1,1},\dots,\theta'_{1,K},\dots,\theta'_{n,1},\dots,\theta'_{n,K}]^T$, where $\theta'_{i,k}\in [0,5]$. For the followers under the observation $\theta'$, set $U_i^c(\theta',t_k)=U_i^c(t_k)+\theta'_{i,k}$ and $U_i^u(\theta',t_k)=U_i^u(t_k)+\theta'_{i,k}$.   We compute DSSE of $\mathcal{H}^2(\Theta)$, similar to (\ref{eq::MILP-MSSE}). 
Take $(x_{\textit{\tiny DSSE}},\boldsymbol{y}_{\textit{\tiny DSSE}},\theta^*)$ as the DSSE strategy of  each  instance.

\textit{Case 3}:   $(x_{\textit{\tiny DSSE}},\boldsymbol{y}_{\textit{\tiny DSSE}})$ is a HNE when the leader is able to trick followers into attacking target $t_{K_{max}}$, where $K_{max}\in \mathop{\text{argmax}}\limits_{k=1,\dots,K} U_l^c(t_k)$.

\textit{Case 4}: the leader is able to trick followers into attacking target $t_{K_{max}}$ when the DSSE strategy is a HNE.

\noindent The above two cases are represented in   blue lines and  red lines.

In Fig. \ref{fig:subfig_DSSE}, the ratio of \textit{Case 3} is always $100\%$, which verifies Theorem  \ref{th::h2}. Also, the ratio of \textit{Case 4} is always larger than $60\%$. Hence, when the DSSE  is  a HNE,  the stable condition of Theorem  \ref{th::h2} can cover many instances.

2) Consider a single-leader-two-followers model in infrastructures protection problems \cite{ijcai2017-516}. Take $n\!=\!\!K\!=\!2$, $R_l\!\!=\!R_i\!=\!\!1$, $\Theta\!=\![0\!,1\!]$, $U^u_l(t_1)\!=\!2,U^u_l(t_2)\!=\!1,U_1^c(\theta',\!t_1)\!=\!3,$ $U_1^c(\theta',\!t_2)\!=\!1$, $U^u_1(\theta',t_1)=4,U^u_1(\theta',t_2)=2$,
$U^c_2(\theta'\!,t_1)=-2.52\theta'+1.428$, $U^c_2(\theta'\!,t_2)=0$,  $U^u_2(\theta',t_1)=-0.4\theta'+2$, and  $U^u_2(\theta',t_2)=-0.16\theta'+0.504$. Also, take $U^c_l(t_1)=3,U^c_l(t_2)=6$ in Fig. \ref{fi::deception_compare_not} and $U^c_l(t_1)=3.2,U^c_l(t_2)=2$ in Fig. \ref{fi::deception_compare_more}, where $U^c_l(t_k)$, the reward for protecting $t_k$, is different in situations with different leader's forms. 
In the environment setting of Fig. \ref{fi::deception_compare_not}, no DSSE is HNE. Neither can the previous work \cite{ijcai2017-516} find the cognitively stable DSSE,  nor can our proposed condition in Theorem \ref{th::h2}  be verified. However, the phenomenon changes in Fig. \ref{fi::deception_compare_more}, because we can find a cognitively stable DSSE, \textit{i.e.}, HNE, once the stable condition in Theorem \ref{th::h2} is satisfied. Thus, our proposed framework and conclusion in Theorem \ref{th::h2} actually provide a way to tell the differences among various environment settings when DSSE is HNE.

\subsection{Robustness of MSSE: in Counterterrorism Problems}

Inspired by the counterterrorism problems with multiple attack forms \cite{zhang2019modeling}, we consider that the American government  wants to defend against the criminals with different attack forms, including armed assaults, bombing/explosion, assassinations, facility/infrastructure attacks, hijackings, and  hostage taking.  Regard the government as a leader and the criminals with $6$ attack forms  as followers. Besides, `New York City,' `Los Angeles', 'SanFrancisco', `Washington, D.C.', and `Chicago' are ranked as the top five risky urban areas in America. Then we regard the $5$ cities as  targets such as the first target for `New York City'. Suppose that all players have \$$1$ millinon budgets, \textit{i.e.}, $R_l=1$ and $R_i=1$ for $i\in\mathbf{P}$. Take $U_l^c(t_k),U_l^u(t_k),U_i^c(t_k),U_i^u(t_k)\in[0,0.7]$ as utilities under the true observation. Also,   followers have a success probability of 0.2, considering that the United States can interdict some attack plots \cite{zhang2019modeling}. Therefore, followers have a false observation of the success rate as $p_{i,k}(\theta')=d_{i,k}\theta'+0.2$, where $\theta'\in \Theta=[-0.2,0.2]$, $\theta_0=0$ and $d_{i,k}$ is generated in the range  $[-1,1]$. Hence, $U_i^c(\theta',t_k)=p_{i,k}(\theta')U_i^c(t_k) $ and $U_i^u(\theta',t_k)=p_{i,k}(\theta')U_{i}^u(t_k)$ are the utilities perceived by the $i$th follower.

\begin{figure*}
\centering
\subfigure[Utilities of the $1$st follower]{
\label{fi::terr1}
 \includegraphics[width=4.8cm]{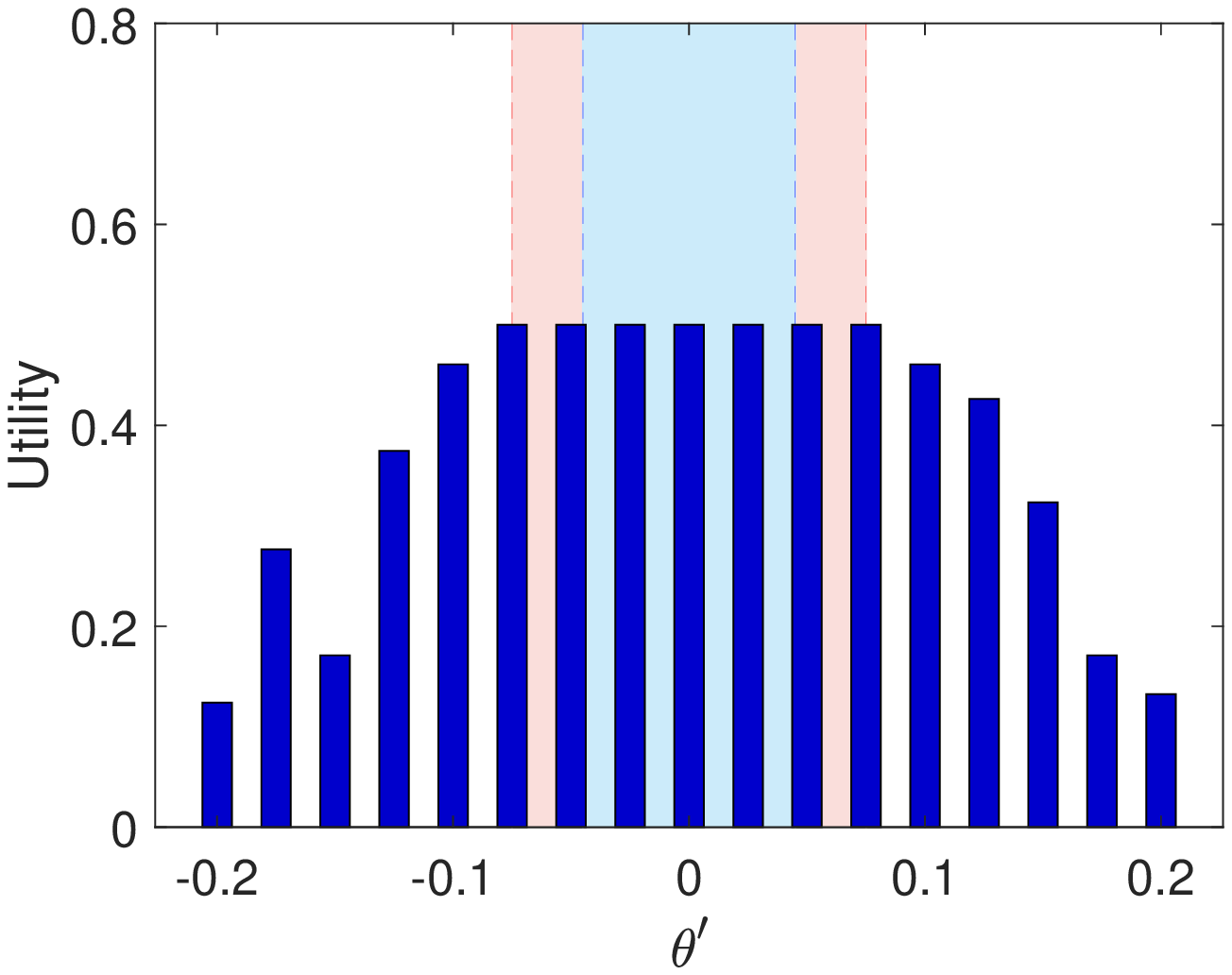}}
 \subfigure[Utilities of the $2$nd follower]
 { \label{fi::terr2}
  \includegraphics[width=4.8cm]{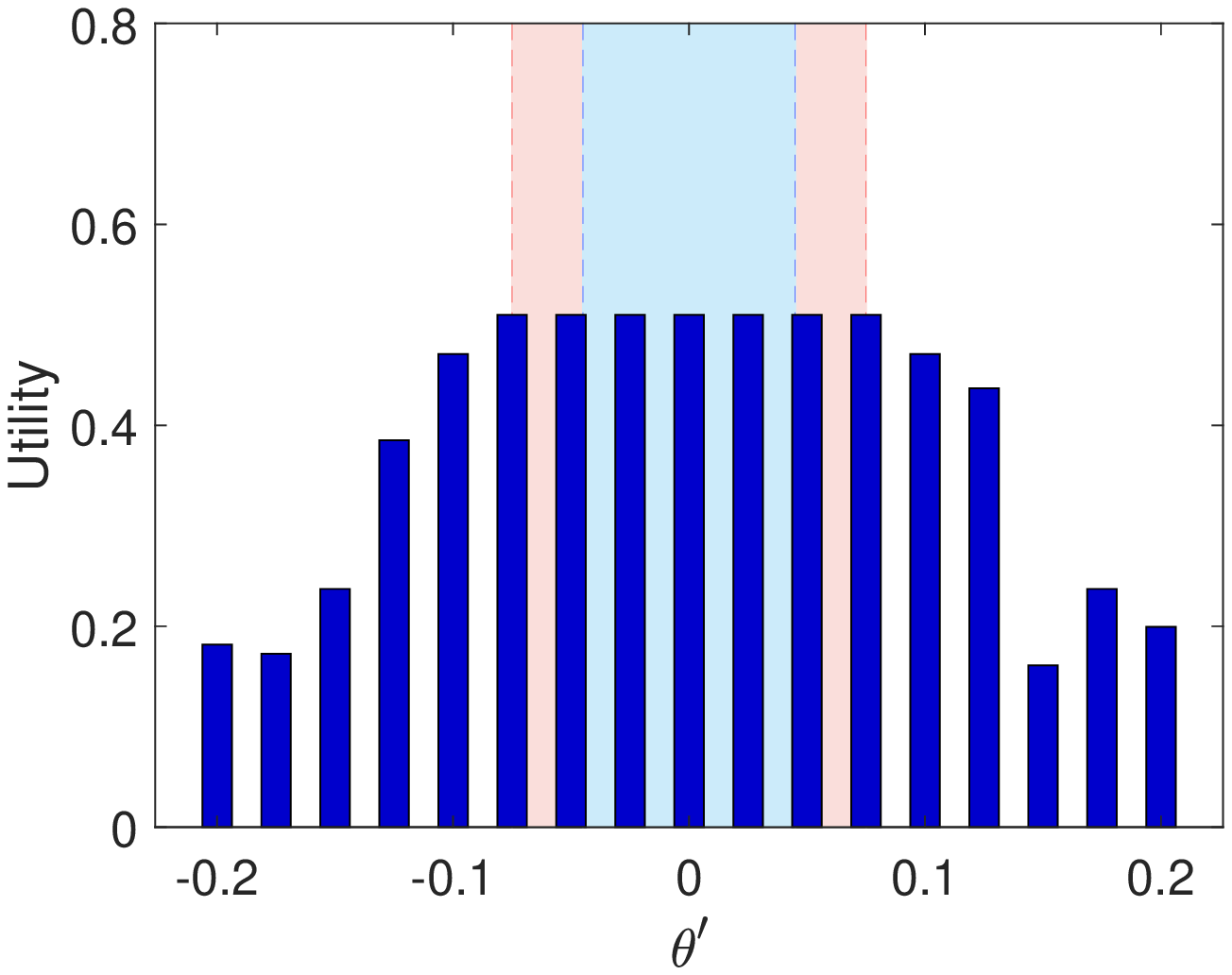}}
   \subfigure[Utilities of the $3$rd follower]
 { \label{fi::terr3}
  \includegraphics[width=4.8cm]{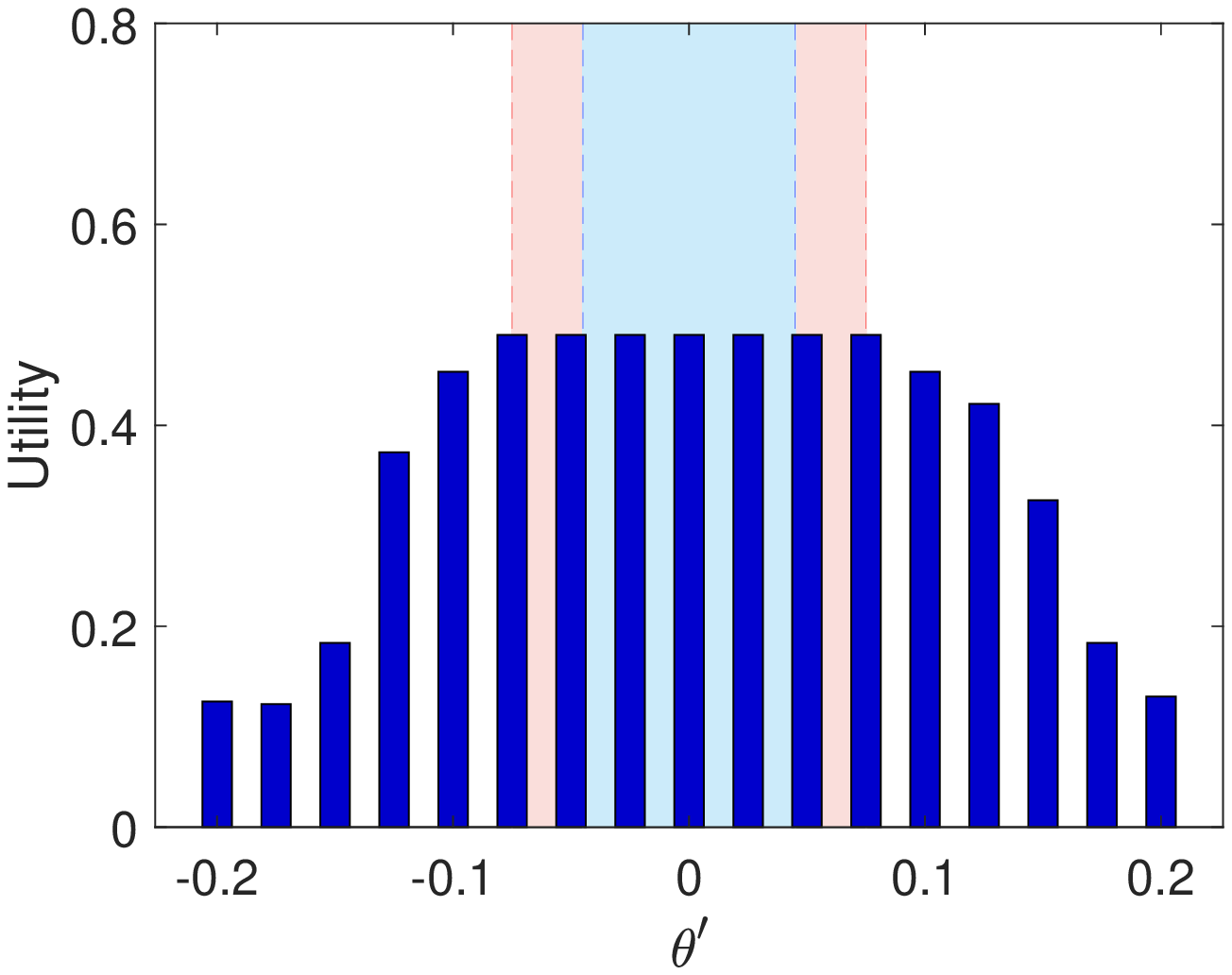}}
  \subfigure[Utilities of the $4$th follower]{
\label{fi::terr4}
 \includegraphics[width=4.8cm]{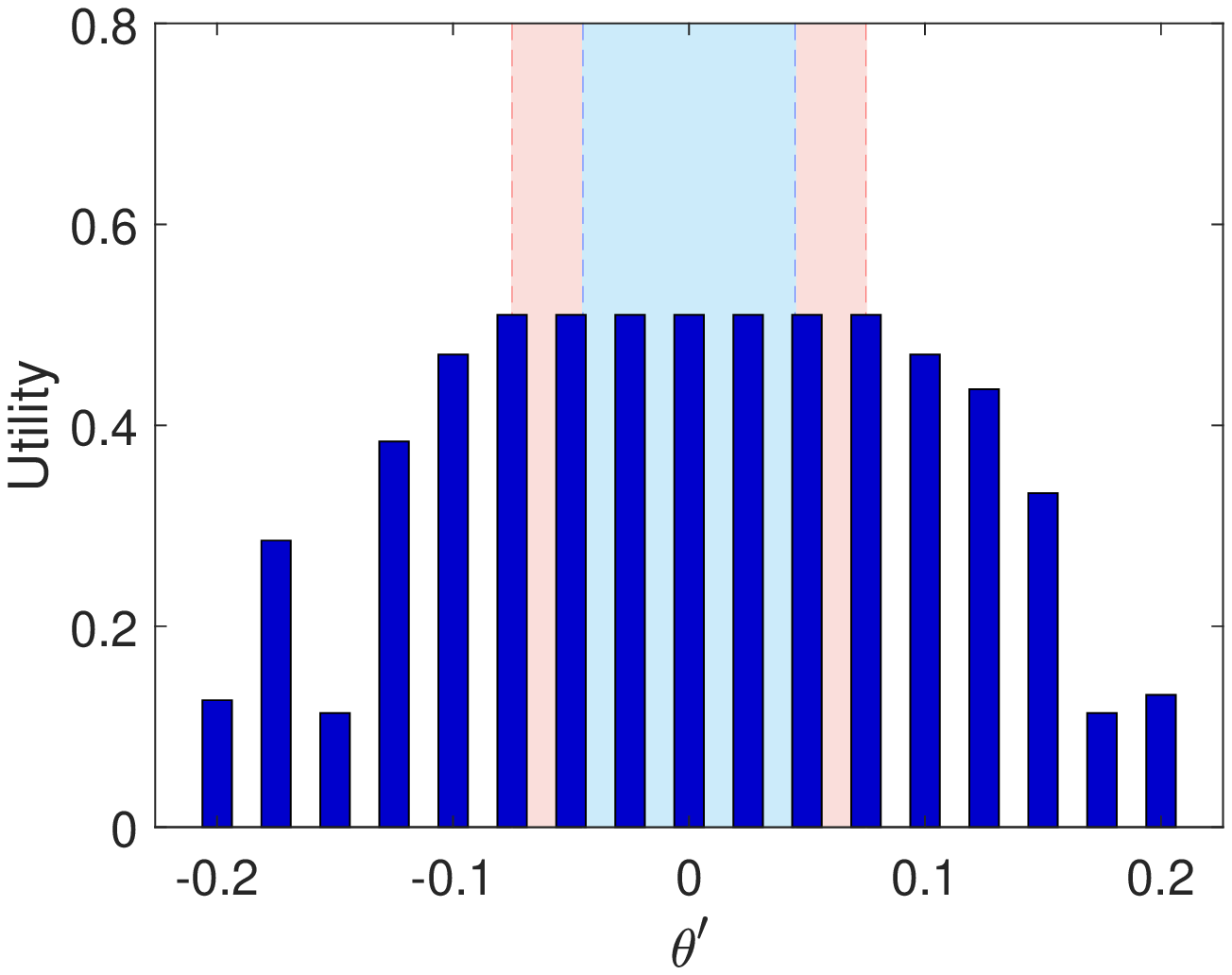}}
 \subfigure[Utilities of the $5$th follower]
 { \label{fi::terr5}
  \includegraphics[width=4.8cm]{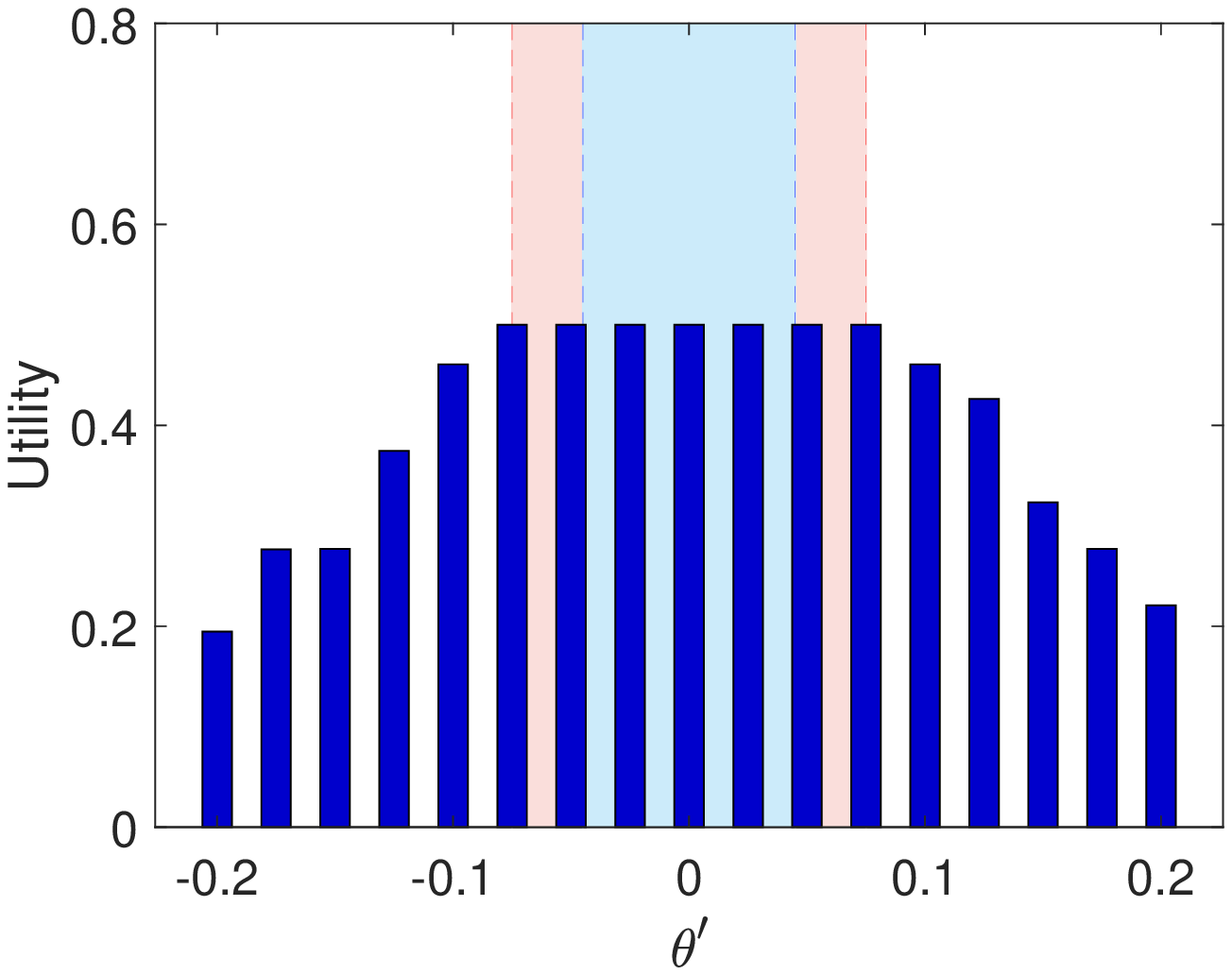}}
   \subfigure[Utilities of the $6$th follower]
 { \label{fi::terr6}
  \includegraphics[width=4.8cm]{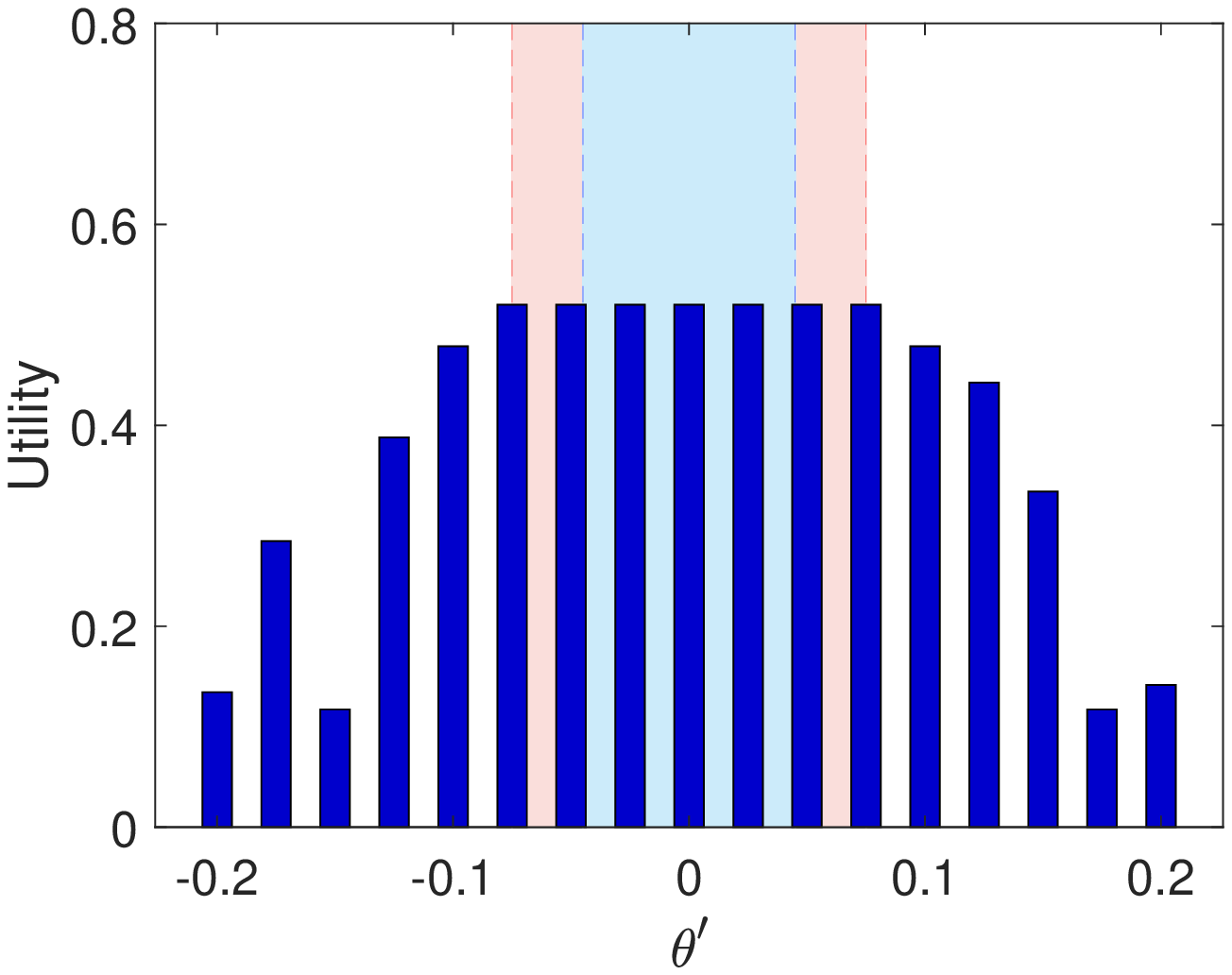}}
 \caption{Utilities of followers in the counterterrorism problem. $\theta'$ denotes the followers' false observation of the parameter, which affects success rate $p_{i,k}(\theta')$ for all $i,k$. The light blue region describes robust bounds of MSSE according to Theorem \ref{th::robust-followers-1}, and the light red region shows bounds referring to $\theta'$, where all followers' utilities are invariant in this instance.}
\label{fi::terrorutility}
 \end{figure*}

Fig. \ref{fi::terrorutility} shows the utilities of all followers,  where the $x$-axis represents the value of $\theta'$ and the $y$-axis is for the true utility of each follower under the observation $\theta'$. The blue cylinders are  followers' true utilities when they select the MSSE strategy under different $\theta'$. 
In Fig. \ref{fi::terrorutility}, the light blue region is in $|\theta'|\leqslant 0.045$ and the light red region is in $|\theta'|\leqslant 0.075$. 
 Actually, the utility under $\theta'=0$ is the real one with no misperception. Notice that all followers' utilities are unchanged when $|\theta'|\leqslant 0.075$, which is  consistent with Theorem \ref{th::robust-followers-1} since $\Delta\theta=0.045<0.075$.

\subsection{Robustness of  DSSE: in CPS}

Similar to the CPS with two players \cite{schlenker2018deceiving}, we consider one network administrator (leader) and one hacker (follower). The follower  invades the  leader with many attack methods such as `malware', `web-based attacks', `denial-of-service', `malicious insiders', `phishing and social enginering', `malicious code', `stolen devices', `ransomware', and `botnets'. Regard the $9$ attack methods as $9$ targets such as the first target for `malware'. Each player has \$$1$ million budgets, \text{i.e.}, $R_l=1$, $R_f=1$ and they allocate funds to the $9$ targets.
Denote  $U_l^c(t_k),U_l^u(t_k),U_f^c(t_k),U_f^u\in [0,2.5]$ as the values for different targets. Moreover, the network administrator makes some observable properties of a system such as TCP/IP stack appear different from what it actually is, and then the hacker probes the system. Concretely, denote $\Theta=[-1,1]$ as the deceptive set and $\theta_0=0$ as the true value. For $\theta'\in\Theta$, $U_f^u(\theta',t_k)=U_f^u(t_k)+d_{k}\theta'^2$ and $ U_f^c(\theta',t_k)=U_f^c(t_k)$ are utilities perceived by the follower, where $d_{k}$ is  generated in the range $D\subset\mathbb{R}$.


Fig. \ref{fi::cyberutility} shows the utilities of the leader with $D=(0,1),(0,2)$, and $(0,3)$, respectively.
The blue cylinders are the leader's utilities if the leader  deceives as  $\theta'$. 
The light blue regions  are in $|\theta'|\leqslant0.45$ in  \ref{fi::cyber1}, $|\theta'|\leqslant0.225$ in  \ref{fi::cyber2} and $|\theta'|\leqslant0.15$ in  \ref{fi::cyber3}. Also, the light red regions are in $|\theta'|\leqslant0.9$ , $|\theta'|\leqslant0.6$, and  $|\theta'|\leqslant0.5$, respectively.  
Besides, the blue cylinder under $\theta'=0$ is the utility if the leader does not deceive. Notice that in \ref{fi::cyber1}, the leader's utility under $|\theta'|\leqslant 0. 9$ is no larger than that under $\theta'=0$. It is  consistent with Theorem \ref{th::robust-leader-1} that the  DSSE strategy is robust for the leader since $\Delta\theta= 0.45<0.9$.  Moreover, in Fig. \ref{fi::cyber1}, if the leader wishes to benefit more from deception,  the deception strategy   needs to exceed $|\theta'|\geqslant0.9>\Delta\theta$.  Similar conclusions can be found in  Fig. \ref{fi::cyber2} and  \ref{fi::cyber3}. In fact, the robust boundary decreases as the bound of the parameter set $D$ increases, which is also consistent with Theorem  \ref{th::robust-leader-1}.
\section{conclusions}
In this paper, we have investigated the  SLMF Stackelberg security game by virtue of the second-level Stackelberg hypergame. We  have provided a   novel criterion  to evaluate both the strategic and cognitive stability of games with misinformation  based on HNE. Moreover, we have  provided two different stable conditions to connect MSSE and   DSSE  with   HNE.  Also, we have analyzed the influences of misperception and deception by the robustness of the MSSE and DSSE strategies. Finally,  we have presented numerical  experiments for the   validity  and  broad applicability of our results.

 \begin{figure*}
\centering
\subfigure[$D=(0,1).$]{
\label{fi::cyber1}
 \includegraphics[width=4.8cm]{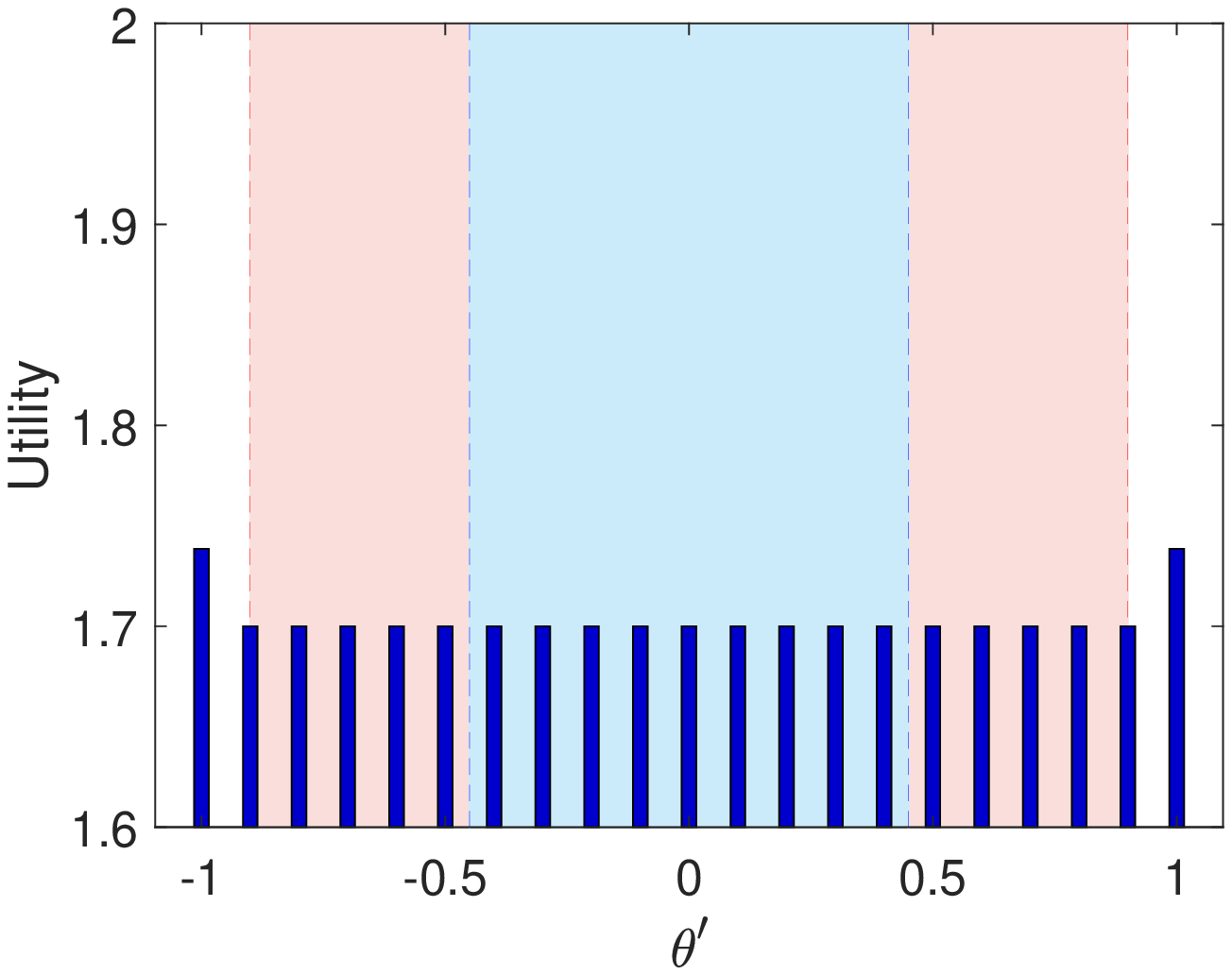}}
 \subfigure[$D=(0,2).$]
 { \label{fi::cyber2}
  \includegraphics[width=4.8cm]{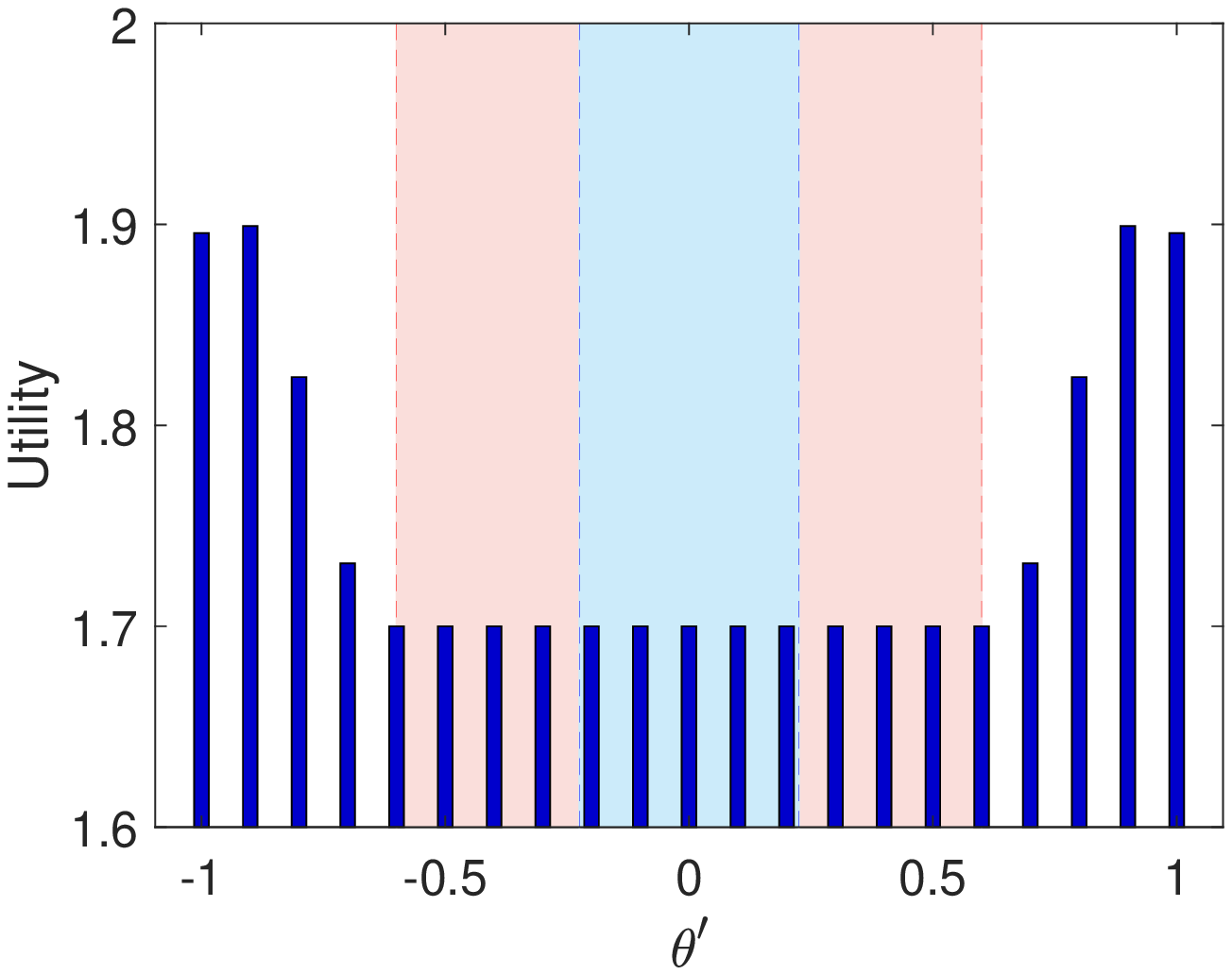}}
   \subfigure[$D=(0,3).$]
 { \label{fi::cyber3}
  \includegraphics[width=4.8cm]{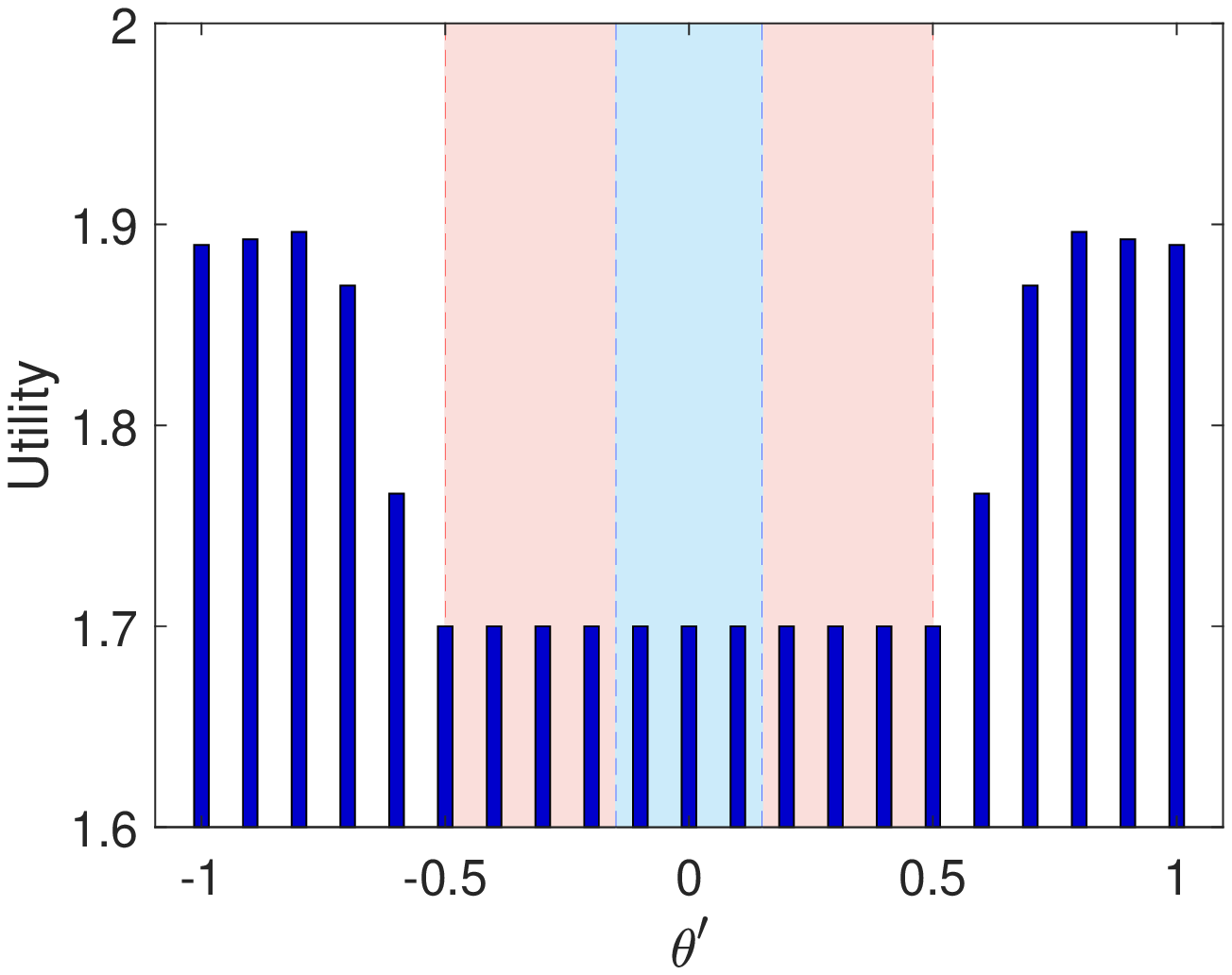}}
 \caption{Utilities of the leader in CPS. $\theta'$ leads to followers’ perceived utility $U_f^u(\theta',t_k)$ for all $k$. The light blue region describes robust bounds of DSSE according to Theorem \ref{th::robust-leader-1}, and the light red region describes bounds referring to $\theta'$, where all followers' utilities are invariant in this instance.} 
\label{fi::cyberutility}
 \end{figure*}

\appendices

\section{Proof of Theorem \ref{th::h1}}\label{ap::th::h1}

Denote $E_i(x,\theta')\!=\!\max\limits_{y_i\in \Omega_i}\! U_i(x,y_i,\theta')$,  $E(x,\theta')\!=\!\sum\limits_{i=1}^n \!E_i(x,\theta')$, and   $E^*(\theta')=\min\limits_{x\in \Omega_l} E(x,\theta')$. The leader's strategy $x\in\Omega_l$ is said to be a \textit{Minimax Strategy} if  $E(x,\theta')=E^*(\theta')$. Then the following proof consists of three steps. Step 1 shows the relationship between the leader's utility function and   followers' ones.  Step 2 reveals that $x_{\textit{\tiny MSSE}}$ is a leader's \textit{Minimax Strategy}. Step 3 shows that $(x_{\textit{\tiny MSSE}},\boldsymbol{y}_{\textit{\tiny MSSE}})$ is   HNE.

\textbf{Step 1:}
Since $({\boldsymbol{y}}',\lambda)\in \text{SOL}(\boldsymbol{y},\theta')$,  $\lambda>0$, and  $A_1(\theta'){\boldsymbol{y}}'=\lambda B \boldsymbol{y}$. Thus,
$$
\sum\limits_{i=1}^n(\boldsymbol{y}')_i^k\frac{U_i^u(\theta',t_k)-U_i^c(\theta',t_k)}{U_l^c(t_k)-U_l^u(t_k)}=\lambda \sum\limits_{i=1}^ny_i^k,\  \forall k=1,\dots,K.
$$
Then
$$
\begin{aligned}
&U_l(x,\boldsymbol{y})-U_l(x',\boldsymbol{y})\\
=& \sum\limits_{k=1}^K\sum\limits_{i=1}^ny_i^k\big(x^k-(x')^k\big)\big(U_l^c(t_k)-U_l^u(t_k)\big)\\
=&-\frac{1}{\lambda}  \sum\limits_{k=1}^K\sum\limits_{i=1}^n(\boldsymbol{y}')_i^k\big(x^k-(x')^k\big)\big(U_i^c(\theta',t_k)-U_i^u(\theta',t_k)\big).
\end{aligned}
$$
Clearly,
$$
\begin{aligned}
&\sum\limits_{i=1}^nU_i\big(x,\boldsymbol{y}'_i,\theta'\big)-\sum\limits_{i=1}^nU_i\big(x',\boldsymbol{y}'_i,\theta'\big)\\
=&\sum\limits_{k=1}^K\sum\limits_{i=1}^n(\boldsymbol{y}')_i^k\big(x^k-(x')^k\big)\big(U_i^c(\theta',t_k)-U_i^u(\theta',t_k)\big).
\end{aligned}
$$
Therefore,
$$
\begin{aligned}
U_l\!(x\!,\boldsymbol{y})\!-\!U_l\!(x'\!,\boldsymbol{y})\!=\! \frac{1}{\lambda}\big(\!\sum\limits_{i=1}^n\!U_i\!(x'\!,(\!\boldsymbol{y}')_i\!,\theta')\!\!-\!\!\sum\limits_{i=1}^n\!U_i\!(x\!,(\!\boldsymbol{y}')_i\!,\theta')\!\big).
\end{aligned}
$$
Since $\lambda >0$, $ U_l(x,\boldsymbol{y})>(=)U_l(x',\boldsymbol{y})$ for $x,x'\in \Omega_l $, and $\boldsymbol{y}\in \mathbf{\Omega}_f $
if and only if
$\sum\limits_{i=1}^nU_i\big(x,\boldsymbol{y}'_i,\theta'\big)<(=)\sum\limits_{i=1}^nU_i\big(x',\boldsymbol{y}'_i,\theta'\big)$ for  $x,x'\in \Omega_l $, $\boldsymbol{y}\in \mathbf{\Omega}_f , \theta' \in \mathbb{R}^m$, and $({\boldsymbol{y}}', \lambda) \in \text{SOL}(\boldsymbol{y},\theta')$.

\textbf{Step 2:} 
By the definition of $E^*$, $E(x_{\textit{\tiny MSSE}},\theta')\geqslant E^*(\theta')$ is always true.  
Suppose
\begin{equation}
\begin{aligned}\label{eq::h1pf-wq0}
E(x_{\textit{\tiny MSSE}},\theta')>E^*(\theta').
\end{aligned}
\end{equation}
Consider   $x^*$  as the leader's  \textit{Minimax Strategy}, where $E(x^*,\theta')=E^*(\theta')$.
For $(\boldsymbol{y}', \lambda) \in \text{SOL}(\boldsymbol{y}_{\textit{\tiny MSSE}},\theta')$,
$$
\sum\limits_{i=1}^n U_i(x^*,\boldsymbol{y}'_i,\theta')\leqslant E^*(\theta').
$$
Denote $S(\boldsymbol{y})=\{(i,k)|y_i^k\neq0,i\in\mathbf{P}, k=1,\dots,K\}$. 
For any $ (i, k_1),(i,k_2)\in S(\boldsymbol{y}_{\textit{\tiny MSSE}})$, the target $k_1$ and target $k_2$ have the same appeal to the $i$th follower.
For any $(i, k)\in S(\boldsymbol{y}_{\textit{\tiny MSSE}})$, set
\begin{equation}\label{eq::th1pf-Mix}
\begin{aligned}
M_i(x_{\textit{\tiny MSSE}})\! =\! x_{\textit{\tiny MSSE}}^k U_i^c(\theta',t_k)\!+\!(R_l\!-\!x_{\textit{\tiny MSSE}}^k) U_i^u(\theta',t_k).
\end{aligned}
\end{equation}
Additionally,  $A_2 \boldsymbol{y}'=0$ implies $(\boldsymbol{y}')_i^k=0$ if $(\boldsymbol{y}_{\textit{\tiny MSSE}})_i^k=0$. Then $S(\boldsymbol{y}')\subseteq S(\boldsymbol{y}_{\textit{\tiny MSSE}})$. Therefore, for any $(i,k)\in S(\boldsymbol{y}')$, (\ref{eq::th1pf-Mix}) also holds. Then
$$
\begin{aligned}
&U_i(x_{\textit{\tiny MSSE}},\boldsymbol{y}'_i,\theta')\\
=&\sum\limits_{k=1}^n (\boldsymbol{y}')_i^k\big(x_{\textit{\tiny MSSE}}^k U_i^c(\theta',t_k)+(R_l-x_{\textit{\tiny MSSE}}^k) U_i^u(\theta',t_k)\big)\\
=&\sum\limits_{(\boldsymbol{y}')_i^k\neq 0} M_i(x_{\textit{\tiny MSSE}}) = R_i M_i(x_{\textit{\tiny MSSE}}).
\end{aligned}
$$
Similarly, 
$
\begin{aligned}
U_i(x_{\textit{\tiny MSSE}},\boldsymbol{y}_{\textit{\tiny MSSE}},\theta')
=R_iM_i(x_{\textit{\tiny MSSE}})
\end{aligned}
$. Thus,
\begin{equation}\label{eq::th1-dsdadasdaaf}
\begin{aligned}
U_i\big(x_{\textit{\tiny MSSE}},\boldsymbol{y}'_i,\theta'\big)=U_i(x_{\textit{\tiny MSSE}},(\boldsymbol{y}_{\textit{\tiny MSSE}})_i,\theta').
\end{aligned}
\end{equation}
As a result,  $\boldsymbol{y}'$ is the   followers' best respose strategy to $x_{\textit{\tiny MSSE}}$ under the observation $\theta'$.
Then
$$
\begin{aligned}
E(x_{\textit{\tiny MSSE}},\theta')=&\sum\limits_{i=1}^n \max\limits_{y_i\in \Omega_i}U_i(x_{\textit{\tiny MSSE}},y_i,\theta')=\sum\limits_{i=1}^n U_i(x_{\textit{\tiny MSSE}},\boldsymbol{y}'_i,\theta').
\end{aligned}
$$

Consequently, 
$
\sum\limits_{i=1}^n U_i(x^*,\boldsymbol{y}'_i,\theta')<\sum\limits_{i=1}^n U_i(x_{\textit{\tiny MSSE}},\boldsymbol{y}'_i,\theta').
$
According to \cite{korzhyk2011stackelberg}, there exists $x'$ such that $\boldsymbol{y}_{\textit{\tiny MSSE}}\in \boldsymbol{\textbf{BR}}(x',\theta')$ and $
\sum\limits_{i=1}^n U_i(x',\boldsymbol{y}'_i,\theta')<\sum\limits_{i=1}^n U_i(x_{\textit{\tiny MSSE}},\boldsymbol{y}'_i,\theta')$. Recalling \textbf{Step 1},  
$
U_l(x',\boldsymbol{y}_{\textit{\tiny MSSE}})>U_l(x_{\textit{\tiny MSSE}},\boldsymbol{y}_{\textit{\tiny MSSE}}),
$ which contradicts that $x_{\textit{\tiny MSSE}}$ is the leader's MSSE strategy. Thus,  (\ref{eq::h1pf-wq0}) does not hold. As a result,
$
E(x_{\textit{\tiny MSSE}},\theta')=E^*(\theta'),
$
which indicates that $x_{\textit{\tiny MSSE}}$ is the leader's \textit{Minimax Strategy} and $\boldsymbol{y}_{\textit{\tiny MSSE}}$ is the corresponding strategies of followers.

\textbf{Step 3}: 
Note that $E(x_{\textit{\tiny MSSE}},\theta')=E^*(\theta')$ and $\boldsymbol{y}_{\textit{\tiny MSSE}}\in \boldsymbol{\textbf{BR}}(x_{\textit{\tiny MSSE}},\theta')$. By  (\ref{eq::th1-dsdadasdaaf}),  $\boldsymbol{y}'\in \boldsymbol{\textbf{BR}}(x_{\textit{\tiny MSSE}},\theta')$.
Define another associated zero-sum game $\bar{\mathcal{G}}$ with two players denoted as $\{1,2\}$ in $\bar{\mathcal{G}}$. The  strategy set of player $1$ is $\Omega_l$ and the strategy set of player $2$ is $\mathbf{\Omega}_f$. For any $x\in \Omega_l, \boldsymbol{y}\in \mathbf{\Omega}_f$,  $\bar{U}_1(x,\boldsymbol{y},\theta')=-\sum\limits_{i=1}^nU_i(x,y_i,\theta')$ and $\bar{U}_2(x,\boldsymbol{y},\theta')=\sum\limits_{i=1}^nU_i(x,y_i,\theta')$ are the  utility functions of player $1$ and player $2$, respectively. Each player aims at  maximizing its utility functions.

For any $\boldsymbol{y}\in \mathbf{\Omega}_f$, since $\boldsymbol{y}'_i\in \text{BR}_i(x_{\textit{\tiny MSSE}},\theta')$,
$\boldsymbol{y}'$ is the best response strategy to $x_{\textit{\tiny MSSE}}$ in $\bar{\mathcal{G}}$. Moreover,
$$
\begin{aligned}
E(x,\theta')=&\sum\limits_{i=1}^n\max\limits_{y_i\in \Omega_i} U_i(x,y_i,\theta')=\max\limits_{y\in \mathbf{\Omega}_f}\bar{U}_2(x,\boldsymbol{y},\theta').
\end{aligned}
$$
Since $E(x_{\textit{\tiny MSSE}},\theta')=E^*(\theta')$,
$$
x_{\textit{\tiny MSSE}}\in \mathop{\text{argmin}}_{x\in\Omega_l} \max\limits_{y\in \mathbf{\Omega}_f}\bar{U}_2(x,\boldsymbol{y},\theta').
$$
Then $x_{\textit{\tiny MSSE}}$ is the \textit{Minimax Strategy} in $\bar{\mathcal{G}}$. By Theorem 3.2 in \cite{fudenberg1991game}, $(x_{\textit{\tiny MSSE}},\boldsymbol{y}')$ is also  a NE of $\bar{\mathcal{G}}$. Then $x_{\textit{\tiny MSSE}}$ is also the best response strategy to $\boldsymbol{y}'$ in $\bar{\mathcal{G}}$. For any $x\in \Omega_l$,
$
\begin{aligned}
\bar{U}_1(x_{\textit{\tiny MSSE}},\boldsymbol{y}',\theta')\geqslant \bar{U}_1(x,\boldsymbol{y}',\theta').
\end{aligned}
$
Therefore,
$
\sum\limits_{i=1}^nU_i\big(x_{\textit{\tiny MSSE}},\boldsymbol{y}'_i,\theta'\big)\leqslant \sum\limits_{i=1}^nU_i\big(x,\boldsymbol{y}'_i,\theta'\big).
$
Then by \textbf{Step 1},
\begin{equation}\label{eq::th1::renyixing}
\begin{aligned}
U_l(x_{\textit{\tiny MSSE}},\boldsymbol{y}_{\textit{\tiny MSSE}},\theta')\geqslant U_l(x,\boldsymbol{y}_{\textit{\tiny MSSE}},\theta').
\end{aligned}
\end{equation}

Because  (\ref{eq::th1::renyixing}) holds for any $x\in \Omega_l$, $x_{\textit{\tiny MSSE}}$ is the best response strategy to $\boldsymbol{y}_{\textit{\tiny MSSE}}$ in $\mathcal{H}^2(\theta')$. Then $(x_{\textit{\tiny MSSE}},\boldsymbol{y}_{\textit{\tiny MSSE}})$ is   HNE of $\mathcal{H}^2(\theta')$. \hfill $\square$ \par

\section{Proof of Theorem \ref{th::h2}}\label{ap::th::h2}
Clearly,  there exists  $\theta^{*}\in \Theta$ such that $ (\boldsymbol{y}^*)_i^{K_{max}}=R_i$, $x^*\in \Omega_l$, and $\boldsymbol{y}^*\in \boldsymbol{\textbf{BR}}(x^*,\theta^*)$, where $(x^*,\boldsymbol{y}^*)$ is the decision result under the observation $\theta $.
Thus, $\text{SOL}(\boldsymbol{y}^*,\theta^*)$ has a solution
$$
\begin{aligned}
\left\{
  \begin{array}{l}
\lambda = \frac{\sum\limits_{i=1}^n R_i\frac{ U_i^u(\theta^*,t_{K_{max}})-U_i^c(\theta^*,t_{K_{max}})}{ U_l^c(t_{K_{max}})- U_l^u(t_{K_{max}})}}{\sum\limits_{i=1}^n R_i},\\
(\boldsymbol{y}^*)_i^{K_{max}}=R_i, \forall i\in\mathbf{P},\\
(\boldsymbol{y}^*)_i^k=0, \forall l\neq {K_{max}}, i\in\mathbf{P}.
  \end{array}
\right.
\end{aligned}
$$
By Theorem \ref{th::h1}, $x^*$ is the best response strategy to $\boldsymbol{y}^*$ and $(x^*)^{K_{max}}=R_l$.
Thus,
$$
\begin{aligned}
&U_l(x^*,\boldsymbol{y}^*)\\
=& \sum\limits_{k=1}^K\! \big(\sum\limits_{i=1}^n(\boldsymbol{y}^*)^k_i\big)\!\big((x^*)^k U_l^c(t_k)\!+\!(R_l\!-\!(x^*)^k) U_l^u(t_k)\big)\\
=& \sum\limits_{i=1}^nR_i R_lU_l^c(t_{K_{max}})=  R_lU_l^c(t_{K_{max}})(\sum\limits_{i=1}^nR_i).
\end{aligned}
$$
Since $K_{max}\in \mathop{\text{argmax}}\limits_{k\in K} U_l^c(t_k)$, for all $k=1,\dots,K$,
$
U_l^c(t_{K_{max}})\geqslant U_l^c(t_k).
$
By Assumption \ref{as::ass1},
$
U_l^c(t_{K_{max}})\geqslant U_l^u(t_k).
$
Then, for any $\theta\in \Theta, \boldsymbol{y}\in \mathbf{\Omega}_f,x\in \Omega_l$,  we have
$$
\begin{aligned}
&U_l(x,\boldsymbol{y})\\
\leqslant & \sum\limits_{k=1}^K \sum\limits_{i=1}^ny^k_i\big(x^k U_l^c(t_{K_{max}})+(R_l-x^k) U_l^c(t_{K_{max}})\big)\\
=& \sum\limits_{k=1}^K \big(x^k U_l^c(t_{K_{max}})+(R_l-x^k) U_l^c(t_{K_{max}})\big)\sum\limits_{i=1}^ny^k_i\\
=&\sum\limits_{k=1}^KR_l U_l^c(t_{K_{max}})\sum\limits_{i=1}^ny^k_i=  R_lU_l^c(\theta,t_{K_{max}})(\sum\limits_{i=1}^nR_i).
\end{aligned}
$$
Thus, $U_l(x,\boldsymbol{y}) \leqslant U_l(x^*,\boldsymbol{y}^*)$ for any $x\in \Omega_l,\boldsymbol{y}\in\mathbf{\Omega}_f$.
Therefore,  $ \theta^* \in\mathop{\text{argmax}}\limits_{\theta' \in \Theta}\max\limits_{x\in \Omega_l, \boldsymbol{y}\in \boldsymbol{\textbf{BR}}(x,\theta')}U_l(x,\boldsymbol{y} )$ is the optimal deception. Also, $(x^*,\boldsymbol{y}^*)$ is a DSSE   of $\mathcal{H}^2(\Theta)$. Besides, $x^*$ is the best response strategy to $\boldsymbol{y}^*$ since $U_l(x,\boldsymbol{y}^*)\leqslant U_l(x^*,\boldsymbol{y}^*)$ for any $x\in \Omega_l$. Thus, the conclusion follows. \hfill $\square$

\section{Proof of Theorem \ref{th::robust-followers-1}}\label{ap::th::robust-followers-2}

1) For any $\alpha_i\in\Gamma_i^1(x_{\textit{\tiny SSE}},\theta_0), l\notin\Gamma_i^1(x_{\textit{\tiny SSE}},\theta_0)$,
$
\begin{aligned}
g_i(x_{\textit{\tiny SSE}},\theta_0,\alpha_i)&>g_i(x_{\textit{\tiny SSE}},\theta_0,l).
\end{aligned}
$
Since $U_i^c(\theta,t_{\alpha_i})$ and $U_i^c(\theta,t_{\alpha_i})$ are differentiable in $\theta\in \Theta$ by Assumption \ref{as::ass5}, there is a convex set  $\delta_\theta^i $ such that, for all $ \theta \in \delta_{\theta}^i$,
$
\begin{aligned}
g_i(x_{\textit{\tiny SSE}},\theta,\alpha_i)&>g_i(x_{\textit{\tiny SSE}},\theta,l).
\end{aligned}
$
Let $\delta_\theta=\cap_{i=1}^{n}\delta_\theta^i$. Then $\text{int}(\delta_\theta)$ is nonempty.
For any $\theta' \in \delta_\theta $, the $i$th follower attacks the target in $\Gamma_i$, which leads to  the same profit as $\boldsymbol{y}_{\textit{\tiny SSE}}$. Thus,
$U_i\big(x_{\textit{\tiny SSE}},\!(\boldsymbol{y}_{\textit{\tiny SSE}})_i,\!\theta_0\big)\!=\!U_i\big(x_{\textit{\tiny SSE}},\!(\boldsymbol{y}_{\textit{\tiny MSSE}})_i(\theta'),\!\theta_0\big).$

2) For any $\alpha_i\in\Gamma_i^1(x_{\textit{\tiny SSE}},\theta_0),\beta_i\in\Gamma_i^2(x_{\textit{\tiny SSE}},\theta_0),l\notin\Gamma_i^1(x_{\textit{\tiny SSE}},\theta_0)$, we have $g_i(x_{\textit{\tiny SSE}},\theta_0,\alpha_i)>g_i(x_{\textit{\tiny SSE}},\theta_0,l)$ and $g_i(x_{\textit{\tiny SSE}},\theta_0,\beta_i)\geqslant g_i(x_{\textit{\tiny SSE}},\theta_0,l)$.
Since $U_i^c(\theta,t_k)$ and $U_i^u(\theta,t_k)$ are $\varsigma$-Lipschitz continuous  in $\theta\in \Theta$, for any $\theta,\theta'\in\Theta$,
$$
\begin{aligned}
|U_i^c(\theta,t_k)-U_i^c(\theta',t_k)|&\leqslant\varsigma \parallel\theta-\theta'\parallel,\\
|U_i^u(\theta,t_k)-U_i^u(\theta',t_k)|&\leqslant\varsigma \parallel\theta-\theta'\parallel.
\end{aligned}
$$
Thus,
$$
\begin{aligned}
&|g_i(x_{\textit{\tiny SSE}},\theta,k)-g_i(x_{\textit{\tiny SSE}},\theta',k)|\\
=&|x_{\textit{\tiny SSE}}^k\big(U_i^c(\theta,t_k)-U_i^c(\theta',t_k)\big)\\
&+(R_l-x_{\textit{\tiny SSE}}^k)|U_i^u(\theta,t_k)-U_i^u(\theta',t_k)|\\
\leqslant&x_{\textit{\tiny SSE}}^k\varsigma \parallel\theta-\theta'\parallel+(R_l-x_{\textit{\tiny SSE}}^k)\varsigma \parallel\theta-\theta'\parallel\\
=&\varsigma R_l \parallel\theta-\theta'\parallel.
\end{aligned}
$$
Therefore, for any $k$, $g_i(x_{\textit{\tiny SSE}},\theta,k)$ is $\varsigma R_l$-Lipschitz continuous in $\theta\in \Theta$. Then
\begin{equation}\label{eq::subsset}
\begin{aligned}
g_i(x_{\textit{\tiny SSE}},\theta,l)\leqslant& g_i(x_{\textit{\tiny SSE}},\theta_0,l)+\varsigma R_l \parallel\theta-\theta'\parallel\\
\leqslant& g_i(x_{\textit{\tiny SSE}},\theta_0,\beta_i)+\varsigma R_l \parallel\theta-\theta'\parallel.
\end{aligned}
\end{equation}

Also, since $U_i^c(\theta,t_k)$ and $U_i^u(\theta,t_k)$ are convex and differentiable in $\theta$, $g_i(x_{\textit{\tiny SSE}},\theta,k)$ is convex in $\theta\in \Theta$. Thus,
$$
\begin{aligned}
g_i(x_{\textit{\tiny SSE}},\theta,k)\!-\!g_i(x_{\textit{\tiny SSE}},\theta_0,k)\!\geqslant\!\bigtriangledown_{\theta} g_i(x_{\textit{\tiny SSE}},\theta_0,k)^T\!(\theta\!-\!\theta_0)\!.
\end{aligned}
$$
Take $\bigtriangledown_{i}^{\alpha_i}=\bigtriangledown_{\theta} g_i(x_{\textit{\tiny SSE}},\theta_0,\alpha_i)$. If $(\bigtriangledown_{i}^{\alpha_i})^T\bigtriangledown_{i}^{\alpha_i}\neq0$, then, with taking $ q_\theta=\frac{(\bigtriangledown_{i}^{\alpha_i})^T(\theta-\theta_0)}{(\bigtriangledown_{i}^{\alpha_i})^T\bigtriangledown_{i}^{\alpha_i}} $, $|q_\theta|\leqslant\parallel\theta-\theta_0\parallel $. Thus,
$$
\begin{aligned}
&g_i(x_{\textit{\tiny SSE}},\theta,\alpha_i)-g_i(x_{\textit{\tiny SSE}},\theta_0,\alpha_i)\\
\geqslant &\bigtriangledown_{\theta} g_i(x_{\textit{\tiny SSE}},\theta_0,\alpha_i)^T(\theta-\theta_0)\\
=&-q_\theta (\bigtriangledown_{i}^{\alpha_i})^T\bigtriangledown_{i}^{\alpha_i}\\
\geqslant & -\parallel\theta-\theta_0\parallel(\bigtriangledown_{i}^{\alpha_i})^T\bigtriangledown_{i}^{\alpha_i}.
\end{aligned}
$$
Obviously, if $(\bigtriangledown_{i}^{\alpha_i})^T\bigtriangledown_{i}^{\alpha_i}=0$,
$$
\begin{aligned}
g_i(x_{\textit{\tiny SSE}},\theta,\alpha_i)\!-\!g_i(x_{\textit{\tiny SSE}},\theta_0,\alpha_i)
\!\geqslant \!0\!
=\!  -\!\parallel\theta\!-\!\theta_0\parallel(\bigtriangledown_{i}^{\alpha_i})^T\bigtriangledown_{i}^{\alpha_i}\!.
\end{aligned}$$
Recalling (\ref{eq::subsset}),
$$
\begin{aligned}
&g_i(x_{\textit{\tiny SSE}},\theta,\alpha_i)-g_i(x_{\textit{\tiny SSE}},\theta,l)\\
\geqslant &g_i(x_{\textit{\tiny SSE}},\!\theta_0,\!\alpha_i)\!-\! g_i(x_{\textit{\tiny SSE}},\!\theta_0,\!\beta_i)\!-\!\big((\bigtriangledown_{i}^{\alpha_i})^T\!\bigtriangledown_{i}^{\alpha_i}\!+\!\varsigma R_l\big)\! \parallel\!\theta\!-\!\theta'\!\parallel\!.
\end{aligned}
$$
Since $
\parallel\theta-\theta'\parallel<\Delta\theta
=\min\limits_{i\in\mathbf{P}}\frac{\hat{g}_i^1- \hat{g}^2_i}{\bigtriangledown^*_i+\varsigma R_l},
$
$$
\begin{aligned}
\parallel\theta-\theta'\parallel<
\frac{g_i(x_{\textit{\tiny SSE}},\theta_0,\alpha_i)-g_i(x_{\textit{\tiny SSE}},\theta_0,\beta_i)}{(\bigtriangledown_{i}^{\alpha_i})^T\bigtriangledown_{i}^{\alpha_i}+\varsigma R_l}.
\end{aligned}
$$
For any $\alpha_i \in \Gamma_i^1(x_{\textit{\tiny SSE}},\theta_0),l\notin \Gamma_i^1(x_{\textit{\tiny SSE}},\theta_0)$,
$$
g_i(x_{\textit{\tiny SSE}},\theta,\alpha_i)>g_i(x_{\textit{\tiny SSE}},\theta,l).
$$
For $i\in\mathbf{P}$,
$
\begin{aligned}
U_i\big(x_{\textit{\tiny SSE}},\!(\boldsymbol{y}_{\textit{\tiny SSE}})_i,\!\theta_0\big)\!=\!U_i\big(x_{\textit{\tiny SSE}},\!(\boldsymbol{y}_{\textit{\tiny MSSE}})_i(\theta'),\!\theta_0\big).
\end{aligned}
$
\hfill $\square$

\section{Proof of Theorem \ref{th::robust-leader-1}}\label{ap::th::robust-leader-1}
1) By Theorem \ref{th::robust-followers-1}, the leader does not change its strategy under $\delta_\theta$ according to \cite{ijcai2019-75}. Thus, the leader's profit does not change, and $U_l(x_{\textit{\tiny SSE}},\boldsymbol{y}_{\textit{\tiny SSE}})=U_l(x_{\textit{\tiny DSSE}},\boldsymbol{y}_{\textit{\tiny DSSE}}).$

2) By Assumption \ref{as::ass6}, $\Gamma_i^1(x_{\textit{\tiny SSE}},\theta_0)$ has  the unique element. Take $\alpha_i\in\Gamma_i^1(x_{\textit{\tiny SSE}},\theta_0),\beta_i\in\Gamma_i^2(x_{\textit{\tiny SSE}},\theta_0),l\notin\Gamma_i^1(x_{\textit{\tiny SSE}},\theta_0)$. As shown in the proof of Theorem \ref{th::robust-followers-1}, $g_i(x_{\textit{\tiny SSE}},\theta,l)$ is $\varsigma R_l$-Lipschitz continuous in $\theta\in \Theta$. Then
$$
\begin{aligned}
g_i(x_{\textit{\tiny SSE}},\theta,l)\leqslant& g_i(x_{\textit{\tiny SSE}},\theta_0,l)+\varsigma R_l \parallel\theta-\theta'\parallel.
\end{aligned}
$$
Since $\beta_i\in\Gamma_i^2(x_{\textit{\tiny SSE}},\theta_0)$,
$
\begin{aligned}
g_i(x_{\textit{\tiny SSE}},\theta_0,\beta_i)&\geqslant g_i(x_{\textit{\tiny SSE}},\theta_0,l).
\end{aligned}
$
Then 
$
g_i(x_{\textit{\tiny SSE}},\theta,l)
\leqslant g_i(x_{\textit{\tiny SSE}},\theta_0,\beta_i)+\varsigma R_l \parallel\theta-\theta'\parallel.
$ 
Also,
$$
\begin{aligned}
g_i(x_{\textit{\tiny SSE}},\theta,\alpha_i)\geqslant& g_i(x_{\textit{\tiny SSE}},\theta_0,\alpha_i)-\varsigma R_l \parallel\theta-\theta'\parallel.
\end{aligned}
$$
Therefore,
$$
\begin{aligned}
&g_i(x_{\textit{\tiny SSE}},\theta,\alpha_i)-g_i(x_{\textit{\tiny SSE}},\theta,l)\\
\geqslant &g_i(x_{\textit{\tiny SSE}},\theta_0,\alpha_i)- g_i(x_{\textit{\tiny SSE}},\theta_0,\beta_i)-2\varsigma R_l\parallel\theta-\theta'\parallel.
\end{aligned}
$$
For any $\theta$ with $\parallel\theta-\theta_0\parallel<\Delta\theta$, since
$
\begin{aligned}
\Delta\theta=\min\limits_{i\in\mathbf{P}}\frac{\hat{g}_i^1- \hat{g}^2_i}{2\varsigma R_l},
\end{aligned}
$
$$
\begin{aligned}
\parallel\theta-\theta_0\parallel<
\frac{g_i(x_{\textit{\tiny SSE}},\theta_0,\alpha_i)-g_i(x_{\textit{\tiny SSE}},\theta_0,\beta_i)}{2\varsigma R_l}.
\end{aligned}
$$
Therefore, for $i\in \mathbf{P}$,
\begin{equation}\label{eq::robustness-theta-deltatheta-2}
\begin{aligned}
g_i(x_{\textit{\tiny SSE}},\theta,l)<g_i(x_{\textit{\tiny SSE}},\theta,\alpha_i).
\end{aligned}
\end{equation}
According to \cite{ijcai2019-75} and Assumption \ref{as::ass6},  (\ref{eq::robustness-theta-deltatheta-2}) holds for any $x_{\textit{\tiny DSSE}}\in\Omega_l$
. Thus, the leader does not change its strategy under $\delta_\theta$, and $U_l(x_{\textit{\tiny SSE}},\boldsymbol{y}_{\textit{\tiny SSE}})=U_l(x_{\textit{\tiny DSSE}},\boldsymbol{y}_{\textit{\tiny DSSE}}).$ \hfill $\square$

\ifCLASSOPTIONcaptionsoff
  \newpage
\fi



%
\bibliographystyle{IEEEtran}

\bibliography{autosam}

%

\end{document}